\documentclass[11pt,a4paper]{amsart}
\usepackage[foot]{amsaddr}
\usepackage{ifxetex}
\ifxetex
  \usepackage[no-math]{fontspec}
\else
\fi
\usepackage{amsmath}
\usepackage{amsfonts}
\usepackage{amssymb}
\usepackage{amsthm}
\usepackage{fullpage}
\usepackage{microtype}
\usepackage[libertine]{newtxmath}
\usepackage[tt=false]{libertine} %% tt is ugly 
\usepackage{caption}
\usepackage{bbm}
\usepackage{hyperref, color}
\hypersetup{colorlinks=true,citecolor=blue, linkcolor=blue, urlcolor=blue}
\usepackage[linesnumbered,boxed,ruled,vlined]{algorithm2e}
\usepackage{bm}
\usepackage{bbm}
\usepackage[numbers]{natbib}
\usepackage{xcolor}
\usepackage{enumerate} 
\usepackage{enumitem}
\usepackage{tabularx}
\usepackage{array}
\usepackage{cleveref}
\newcolumntype{L}[1]{>{\raggedright\arraybackslash}p{#1}}
\newcolumntype{C}[1]{>{\centering\arraybackslash}m{#1}}
\newcolumntype{R}[1]{>{\raggedleft\arraybackslash}p{#1}}

\usepackage{makecell}
\usepackage{footnote}
\makesavenoteenv{tabular}

\newcommand{\Nv}{N_{\mathrm{vtx}}}
\newcommand{\cnon}{\mathcal{C}_{\mathrm{non}}}
\newcommand{\cadp}{\mathcal{C}_{\mathrm{apt}}}

\newcommand{\Vfro}{V_{\mathrm{frozen}}}
\newcommand{\Efro}{\mathcal{E}_{\mathrm{frozen}}}
\newcommand{\vol}{\mathrm{Vol}}

\newcommand{\DTV}[2]{d_{\mathrm{TV}}\left({#1},{#2}\right)}
\newcommand{\dist}{\mathrm{dist}}

\renewcommand{\epsilon}{\varepsilon}

\newcommand{\Lin}{\mathrm{Lin}}
\newcommand{\Vcol}{V_{\mathrm{set}}}

\newcommand{\twotree}{S_{\mathrm{tree}}}

\newcommand{\vst}{v_{\star}}

\newcommand{\vbl}[1]{\mathsf{vbl}\left(#1\right)}

\newcommand{\sample}{\textnormal{\textsf{InvSample}}}

\newcommand{\Tmix}{\left\lceil 2n \log \frac{4n}{\epsilon} \right\rceil}

\newcommand{\tmix}{T_{\textsf{mix}}}

\newtheorem{theorem}{Theorem}[section]

\newtheorem*{claim*}{Claim}
\newtheorem{condition}[theorem]{Condition}

\newtheorem{lemma}[theorem]{Lemma}
\newtheorem{proposition}[theorem]{Proposition}
\newtheorem{corollary}[theorem]{Corollary}
\theoremstyle{definition}

\newtheorem{definition}[theorem]{Definition}

\newtheorem*{remark*}{Remark}

\def\Pr{\mathop{\mathbf{Pr}}\nolimits}

\renewcommand{\emptyset}{\varnothing}

\newcommand{\abs}[1]{\left\vert#1\right\vert}
\newcommand{\set}[1]{\left\{#1\right\}}
 \newcommand{\tuple}[1]{\left(#1\right)} 
 
 \newcommand{\tp}{\tuple}

\newcommand{\defeq}{\triangleq}

\newcommand{\ctp}[1]{\left\lceil{#1}\right\rceil}
\newcommand{\ftp}[1]{\left\lfloor{#1}\right\rfloor}

\def\*#1{\mathbf{#1}} % Use \*A for \mathbf{A}
\def\+#1{\mathcal{#1}} % Use \+A for \mathcal{A}
\def\-#1{\mathrm{#1}} % Use \-A for \mathrm{A}

\usepackage{todonotes}

\usepackage{xifthen}

\renewcommand{\Pr}[2][]{ \ifthenelse{\isempty{#1}}
  {\mathbf{Pr}\left[#2\right]} {\mathbf{Pr}_{#1}\left[#2\right]} } % Use \Pr[a]{b} for \mathbf{Pr}_a[b], \Pr{b} for  \mathbf{Pr}[b]
\newcommand{\E}[2][]{ \ifthenelse{\isempty{#1}}
  {\mathbf{\mathbf{E}}\left[#2\right]}
  {\mathbf{\mathbf{E}}_{#1}\left[#2\right]} }
  \newcommand{\Var}[2][]{ \ifthenelse{\isempty{#1}}
  {\mathbf{\mathbf{Var}}\left[#2\right]}
  {\mathbf{\mathbf{Var}}_{#1}\left[#2\right]} }

\newcommand{\Dom}[1]{\boldsymbol{#1}}
\newcommand{\Cons}[1]{\mathcal{#1}}
\newcommand{\True}{\mathtt{True}}
\newcommand{\False}{\mathtt{False}}
\newcommand{\Ass}[1]{\boldsymbol{#1}}
\newcommand{\Proj}[1]{\boldsymbol{#1}}
\newcommand{\RDForm}[2]{{#1}^{\ftp{#2}}}
\newcommand{\PX}{\Proj{h}^{X}}
\newcommand{\PY}{\Proj{h}^{Y}}
\title{Sampling Constraint Satisfaction Solutions in the Local Lemma Regime}
\date{}

\author{Weiming Feng}
\author{Kun He}
\author{Yitong Yin}

\thanks{This research is supported by the National Key R\&D Program of China 2018YFB1003202 and the National Science Foundation of China under Grant Nos. 61722207 and 61672275.}
\address[Weiming Feng, Yitong Yin]{State Key Laboratory for Novel Software Technology, Nanjing University. \textnormal{E-mail: \url{fengwm@smail.nju.edu.cn} and \url{yinyt@nju.edu.cn}}.}

\address[Kun He]{Institute of Computing Technology, Chinese Academy of Sciences.
\textnormal{E-mail: \url{hekun.threebody@foxmail.com}}.}

\begin{document}

\maketitle

\begin{abstract}
We give a Markov chain based algorithm for sampling almost uniform solutions of constraint satisfaction problems (CSPs).
Assuming a canonical setting for the Lov\'asz local lemma,
where each constraint is violated by a small number of forbidden local configurations, 
our sampling algorithm is accurate in a local lemma regime,
and the running time is a fixed polynomial whose dependency on $n$ is close to linear, where $n$ is the number of variables.
Our main approach is a new technique called \emph{state compression}, 
which generalizes the ``mark/unmark'' paradigm of Moitra~\cite{Moi19},
%which makes the local-lemma-based sampling approaches~\cite{FGYZ20} 
and can give fast local-lemma-based sampling algorithms.
%For specific subclasses of CSPs, our approach improves the state-of-the-arts for almost-uniform samplers for hypergraph colorings and for CNF solutions.
As concrete applications of our technique, we give the current best almost-uniform samplers for hypergraph colorings and for CNF solutions.
\end{abstract}

\section{Introduction}
\label{section-intro}
The space of constraint satisfaction solutions is one of the most well-studied subjects in Computer Science.
Given a collection of constraints defined on a set of variables, a solution to the \emph{constraint satisfaction problem} (CSP) is an assignment of variables such that all constraints are satisfied.
A fundamental criterion for the existence of constraint satisfaction solutions is given by the \emph{Lov\'asz local lemma}~(LLL)~\cite{EL75}.
%
%Constructing a probability space for all assignments of variables and 
Interpreting the space of all assignment as a probability space and 
the violation of each constraint as a bad event,  %defined on independent variables, 
the local lemma characterizes a regime within which a constraint satisfaction solution always exists, by the tradeoff between: (1)~the chance for the occurrence of each bad event and (2)~the degree of dependency between them. %such that within this regime a constraint satisfaction solution always exists.

%
%The existence of CSP solutions is center topic in many research areas, one of the most successful tool is the \emph{Lov\'asz local lemma} (LLL).
%A center problem is to determine whether the CSP solution exists.
%One of the most successful tool is the \emph{Lov\'asz local lemma}~(LLL)~\cite{EL75}.
%%
%The LLL constructs a probability space and a collection of bad events, each of them indicates the violation of a single constraint. 
%%
%The LLL is captured by two parameters (1) the maximum probability of one bad event; (2) the ``dependency'' among the collection of all bad events.
%%
%If there is a trade-off between these two parameters, then the CSP solution exists in this local lemma regime.
%%

In Computer Science, the studies of the Lov\'asz local lemma are more focused on the \emph{algorithmic LLL} (also called \emph{constructive LLL}), 
which is concerned with not just  existence of a constraint satisfaction solution, but also how to find such a solution efficiently. 
The studies of algorithmic LLL constitute an important line of modern algorithm researches~\cite{beck1991algorithmic,alon1991parallel,molloy1998further,CS00, moser2009constructive,moser2010constructive,Kolipaka2011MoserAT,haeupler2011new,HS17,HS19}.
A major breakthrough was the Moser-Tardos algorithm~\cite{moser2010constructive}, which finds a satisfaction solution efficiently up to a sharp condition known as the Shearer’s bound~\cite{shearer85,Kolipaka2011MoserAT}. 

%The original LLL only gives non-construct proofs~\cite{EL75}.
%%
%Later on, there is a line of works on \emph{algorithmic-LLL}~\cite{beck1991algorithmic,alon1991parallel,molloy1998further,CS00, moser2009constructive,moser2010constructive,Kolipaka2011MoserAT,haeupler2011new,HS17,HS19}, which aims to construct a CSP solution in the local lemma regime.
%%
%An important breakthrough is the Moser-Tardos algorithm~\cite{moser2010constructive}, which finds a solution efficiently up to the Shearer’s bound~\cite{shearer85,Kolipaka2011MoserAT}. 

%The \emph{Lov\'asz local lemma (LLL)}~\cite{EL75} is a basic tool to prove the existence of certain combinatorial structures, and has been extensively studied in combinatorics, probability theory and computer science.
%%
%Given a collection of bad events, the LLL constructs a probability space, then proves that with a positive probability, all the bad events can be avoided. 
%%
%The original LLL only gives non-construct proofs~\cite{EL75,shearer85}.
%%
%Later on, there is a line of works on \emph{algorithmic Lov\'asz local lemma}~\cite{beck1991algorithmic,alon1991parallel,molloy1998further,CS00, moser2009constructive,moser2010constructive,Kolipaka2011MoserAT,haeupler2011new,HS17,HS19}, which aims to construct a solution under the LLL condition.
%%
%The most remarkable result is the Moser-Tardos algorithm~\cite{moser2010constructive}.
%%
%The algorithm finds a solution efficiently under conditions that match the Lovasz Local Lemma in the 
%%under the same condition of 
%the so-called variable-LLL framework~\cite{moser2010constructive,Kolipaka2011MoserAT}. 

In this paper, we are concerned with a problem that we call the \emph{sampling LLL},
which asks for the regimes in which a nearly uniform (instead of an arbitrary) satisfaction solution can be generated efficiently.
This is a \emph{distribution-sensitive} variant of the algorithmic LLL.
The problem is closely related to the problem of estimating the total number of satisfaction solutions, usually via standard reductions~\cite{jerrum1986random,vstefankovivc2009adaptive};
besides, it may also serve as a standard toolkit for solving the inference problems that are well motivated from machine learning applications~\cite{Moi19}.

%In this paper, we consider the problem of \emph{sampling-LLL}.
%This problem asks whether, in the local lemma regime, we can sample a CSP solution uniformly at random instead of merely constructing one.
%The sampling problem is more difficult than construction, but has many important applications. 
%For example, one can approximately count the number of CSP solutions through a sampling algorithm~\cite{jerrum1986random,vstefankovivc2009adaptive}; besides, the sampling problem is also closely related to the inference problems, which are  well-motivated  by the applications in  machine learning~\cite{Moi19}.

This sampling variant of algorithmic LLL is computationally more challenging than the conventional algorithmic LLL.
For example, for $k$-CNF formulas with variable-degree $d$, the Moser-Tardos algorithm for generating an arbitrary solution is known to be efficient when $k\gtrsim\log_2 d$, 
while the problem of generating a nearly uniform solution requires $k\gtrsim2\log_2 d$ to be tractable~\cite{BGGGS19}.

Meanwhile, much less positive progress was known for the sampling LLL.
A fundamental obstacle is that the space of satisfaction solutions may not be connected via local updates of variables~\cite{wigderson2019book}, whereas such connectivity is crucial for mainstream sampling techniques.
In~\cite{GJL19}, Guo, Jerrum and Liu proposed to study the sampling LLL, 
and %showed that a Moser-Tardos algorithm 
resolved the problem for the CSPs with extremal constraints.
%they gave an elegant Moser-Tardos style scheme for uniform sampling satisfaction solutions,
%which is efficient in the local lemma regime for satisfaction solutions to \emph{extremal} conditions.
%in addition to the local lemma condition, requires an extra condition on intersections between bad events.
%
In a major breakthrough~\cite{Moi19}, Moitra introduced a novel approach for approximately counting $k$-SAT solutions.
The approach utilizes the algorithmic LLL to properly \emph{mark/unmark} variables, 
which helps construct efficient linear programmings for estimating marginal probabilities.
For $k$-CNF formulas with variable-degree $d$ within a local lemma regime $k \gtrsim 60 \log d$, the algorithm approximately counts the total number of SAT solutions in time $n^{\mathrm{poly}(dk)}$.
Further extensions of Moitra's approach were made to hypergraph colorings~\cite{guo2019counting} and random CNF formulas~\cite{galanis2019counting}, where the running times are both $n^{\mathrm{poly}(dk)}$ for constraint-width $k$ and variable-degree $d$.
Recently, a much faster algorithm for sampling $k$-SAT solutions inspired by Moitra's algorithm was given in~\cite{FGYZ20}. 
It implements a Markov chain on the assignments of the marked variables chosen via Moitra's approach.
The resulting sampling algorithm enjoys a close-to-linear running time $\widetilde{O}(d^2k^3n^{1.000001})$ with an improved regime $k \gtrsim 20 \log d$.
It also formally confirms that the originally  disconnected solution space is changed to be very well connected after restricting onto a wisely chosen set of marked variables.
However, such approach of fast sampling seems rather restricted to CNF formulas, where the variables can be marked/unmarked non-adaptively to the assignments, whereas for CSPs with larger domains where marking/unmarking variables adaptively to their assignments is crucial~\cite{guo2019counting}, the current approach for fast sampling has met some fundamental barriers.
%extending such approach of fast sampling beyond the Boolean domain, e.g.~to hypergraph colorings, seems to be difficult, because the current Markov chain based approach relies on the variables being marked/unmarked non-adaptively to the assignments of variables, while marking/unmarking variables adaptively to the assignments is crucial for Moitra's approach to be applied for CSPs with larger domains~\cite{guo2019counting}.

For sampling general constraint satisfaction solutions,
we do not know whether the problem is tractable in a local lemma type of regime, 
neither do we know 
any general algorithmic approach that can achieve this.
New ideas beyond the  paradigm of marking/unmarking variables %and restricting the solution space onto marked variables 
are needed.

\subsection{Our results}
\label{section-results}
We consider the problem of uniform sampling constraint satisfaction solutions,
formulated by the variable-framework LLL with \emph{uniform random variable} and \emph{atomic bad events}.
Let $V$ be a collection of $n=\abs{V}$ mutually independent uniform random variables and $\+B$ be a collection of atomic bad events such that
\begin{itemize}
\item \textbf{uniform random variables}: the value of each $v \in V$ is uniformly drawn from a domain~$Q_v$;
\item \textbf{atomic bad events}: each $B \in \+B$ is determined by the variables in $\vbl{B} \subseteq V$, and $B$ occurs if the assignment of $\vbl{B}$ is as specified by the unique forbidden pattern $\sigma_B \in \bigotimes_{v \in \vbl{B}}Q_v$.
\end{itemize}

%and only if the variables in $\vbl{B}$ take a unique assignment $\sigma_B \in \bigotimes_{v \in \vbl{B}}Q_v$.

%In this paper, we model CSP as variable-LLL framework, and we further explore the problem of sampling-LLL.
%We give a new sampling algorithm for variable-LLL framework  with \emph{uniform random variable} and \emph{atomic bad events}.
%%
%Specifically,  let $V$ be a collection of mutually independent uniform random variables and $\+B$ be a collection of atomic bad events such that
%\begin{itemize}
%\item \textbf{uniform random variable}: each $v \in V$ takes value from a domain $Q_v$ uniformly at random;
%\item \textbf{atomic bad event}: each $B \in \+B$ is determined by the set of variables $\vbl{B} \subseteq V$, and $B$ occurs if and only if the variables in $\vbl{B}$ take a unique assignment $\sigma_B \in \bigotimes_{v \in \vbl{B}}Q_v$.	
%\end{itemize}

%The uniform distribution is a natural assumption for LLL

We assume uniform random variables because our goal is to uniformly sample constraint satisfaction solutions. 
Meanwhile, the atomicity of bad events is a natural and fundamental setting assumed in various studies of LLL~\cite{achlioptas2016random,HH17,harris2017algorithmic,kolmogorov2018commutativity,Harris19Obl,achlioptas2019beyond,HS19,harvey2020algorithmic}.

%The uniform random variable is a natural assumption employed in most applications of LLL,
%and the atomic bad event is an important setting in a lot of LLL-related works~\cite{achlioptas2016random,HH17,harris2017algorithmic,kolmogorov2018commutativity,Harris19Obl,achlioptas2019beyond,HS19,harvey2020algorithmic}.
Let $p = \max_{B \in \+B}\Pr{B}$, where the probability is taken over independent  random variables in $V$. %where the probability is taken over the product distribution of all variables in $V$.
Let $G=(\+B,E)$ be the \emph{dependency graph}, where each vertex is a bad event in $\+B$, and the neighborhood of each $B\in \+B$ in $G$ is $\Gamma(B) \triangleq \{B'\in \+B\setminus\{B\} \mid  \vbl{B} \cap \vbl{B'} \neq \emptyset \}$.
%
%For each $B \in \+B$, the neighborhood of $B$ in $G$ is defined by $\Gamma(B) \triangleq \{B'\in \+B \mid B'\neq B, \vbl{B} \cap \vbl{B'} \neq \emptyset \}$.
%
Let $D\triangleq\max_{B\in \+B}\abs{\Gamma(B)}$ denote the maximum degree of the dependency graph.
By the Lov\'asz local lemma, there exists a satisfying assignment that avoids all bad events in $\+B$ if
\begin{align}\label{eq:lll-regime}
\ln \frac{1}{p} \geq \ln D + 1.
\end{align}

Such an instance of LLL naturally specifies a uniform distribution over all satisfying assignments, called the \textbf{LLL-distribution}~\cite{harris2020New}.
Formally, it is the distribution of the independent random variables in $V$ conditioned on that none of the bad events in $\+B$ occurs. 	
%
%We give a fast algorithm for sampling from this distribution.

%We study the sampling problem of the Lov\'asz loca lemma. Given an LLL instance, our goal is to draw random samples from the \textbf{LLL-distribution} i.e. the product distribution of all variables in $V$ conditional on none of bad events in $\+B$ occurs. 	
%Let $\mu$ denote the uniform distribution over all assignments in $\otimes_{v \in V}Q_v$ that avoids all bad events in $\+B$ (a.k.a. LLL-distribution).
%Our first contribution is the following algorithm.
\begin{theorem} 
\label{theorem-main}
The following holds for any $0 < \zeta \leq 2^{-400}$.
There is an algorithm such that given a Lov\'asz local lemma instance with uniform random variables and atomic bad events, if
\begin{align}
\label{eq-condition-main-thm}
\ln \frac{1}{p} \geq  350 \ln D + 3 \ln \frac{1}{\zeta},
\end{align}
%\begin{align}
%\label{eq-condition-main-thm}
%p \leq  \frac{\zeta^3}{D^{1000}},
%\end{align}
then the algorithm outputs a random assignment $\Ass{X} \in \bigotimes_{v \in V}Q_v$ in time
%\begin{align*}
$\widetilde{O}\tp{(D^2 k +q)n \tp{\frac{n}{\epsilon}}^{\zeta}}$,
%\end{align*}
such that 
the distribution of $\Ass{X}$ is $\epsilon$-close to the LLL-distribution in total variation distance,
where $q = \max_{v \in V}\abs{Q_v}$, $k = \max_{B \in \+B} \abs{\vbl{B}}$, and $\widetilde{O}(\cdot)$ hides a factor of $\mathrm{polylog}(n,\frac{1}{\epsilon}, q, D)$. %and $\vbl{B} \subseteq V$ is the set variables that determines the bad event $B$.
\end{theorem}

%\begin{theorem}
%\label{theorem-main}
%The following holds for any $0 < \zeta \leq 2^{-1000}$.
%Let $V$ be a set of mutually independent variables, where each $v \in V$ takes values uniformly in a domain $Q_v$.
%Let $\+B$ be a collection of atomic bad events.
%Let $p = \max_{B \in \+B}\Pr{B}$ and $D$ denote the maximum degree of the dependency graph.
%Let $\mu$ denote the uniform distribution over all assignments in $\otimes_{v \in V}Q_v$ that avoids all bad events in $\+B$.
%There is an algorithm such that given any Lov\'asz local lemma instance with atomic bad events, if
%\begin{align}
%\label{eq-condition-main-thm}
%p \leq  \frac{\zeta^3}{D^{1000}},
%\end{align}
%then the algorithm draws a random sample $X \in \otimes_{v \in V}Q_v$ in time
%\begin{align*}
%\widetilde{O}\tp{D^2 k n \tp{\frac{n}{\epsilon}}^{\zeta}}
%\end{align*}
%such that 
%the distribution of $X$ is $\epsilon$-close to $\mu$ in total variation distance,
%where $\widetilde{O}(\cdot)$ hides a factor of $\mathrm{polylog}(n,\frac{1}{\epsilon}, q, D)$, $q = \max_{v \in V}\abs{Q_v}$, $k = \max_{B \in \+B} \abs{\vbl{B}}$, and $\vbl{B} \subseteq V$ is the set variables that determines the bad event $B$.
%\end{theorem}

This gives a unified approach for sampling uniform LLL-distributions.
It is achieved by a new technique called ``{state compression}'' (see \Cref{section-overview} and \Cref{section-proj-of-LLL}).
%Our result gives a unified approach for sampling from a broad class of uniform LLL-distributions,
%from a broad class of uniform LLL-distributions, 
%within a  local lemma regime that matches~\eqref{eq:lll-regime} up to a constant factor. %(while a constant factor to the condition~\eqref{eq:lll-regime}  is necessary for tractable sampling LLL even for  $k$-CNFs~\cite{BGGGS19}).
%
The time complexity of the sampling algorithm is controlled by a constant parameter  $\zeta$ which also controls the gap to the local lemma condition~\eqref{eq-condition-main-thm}, %in the condition~\eqref{eq-condition-main-thm}, 
so the running time can be arbitrarily close to linear in $n$ as $\zeta$ approaches $0$.
%
%As $\zeta$ approaches $0$ the running time of our sampling algorithm can be arbitrarily close to linear in $n=|V|$.

%Moreover, our algorithm runs in a fixed polynomial time.
%The running time is close to linear if the parameter  $\zeta$ approaches to $0$.
%%Using the simulated annealing reduction~\cite{bezakova2008accelerating,vstefankovivc2009adaptive,huber2015approximation,kolmogorov18faster,FGYZ20},
%%one can also obtain a fast approximate counting algorithm in $\widetilde{O}(\mathrm{poly}(Dkq)n^{2 + \zeta})$ time. %by employing our sampling algorithm as a subroutine.
%The condition in~\eqref{eq-condition-main-thm} matches the local lemma condition up to a constant factor of $\ln D$. 
%The hardness result in~\cite{BGGGS19} shows that it is NP-hard to approximately sample when $\ln \frac{1}{p} \leq 2 \ln D - O(1)$.
%This lower bound is obtained from sampling uniform $k$-CNF solutions, thus the hardness result still holds for LLL instances with uniform variables and atomic constraints.
%It remains to be an interesting open problem to close the gap on the constant factor of $\ln D$.

Though \Cref{theorem-main} is stated for uniform sampling, 
our main result can be extended to the LLL-distributions that arise from non-uniform random variables with arbitrary constant biases, a setting that corresponds to the statistical physics models with constant \emph{local fields}, which are considered  interesting for sampling and counting.
For such a general setting, \Cref{theorem-main} remains to hold by replacing the condition~\eqref{eq-condition-main-thm} with $\ln \frac{1}{p} \geq  C\ln({D}/{\zeta})$ where the constant factor $C$ depends on the maximum bias.
The formal proof of this general result is postpone to the full version of the paper.
%where the  for any $v\in V$ and $x \in Q_v$, the probability that $v$ takes the value $x$ is $\Theta(\frac{1}{|Q_v|})$.
%In this case, the condition in~\eqref{eq-condition-main-thm} becomes $\ln \frac{1}{p} \geq \Omega(\ln \frac{D}{\zeta})$, where the constant factor hidden in $\Omega(\cdot)$ depends on the constant factor hidden in $\Theta(\cdot)$.

On the other hand, any general non-atomic bad event can be seen as a union of disjoint atomic bad events.
Let $B$ be a bad event defined on $\vbl{B} \subseteq V$ and $\mathcal{N}_B\triangleq{\{\sigma\in \bigotimes_{v \in \vbl{B}}Q_v\mid B\text{ occurs at }\sigma\}}$ denote the set of assignments of  $\vbl{B}$ that make $B$ occur.
Event $B$ can thus be decomposed to $\abs{\mathcal{N}_B}$ atomic events, each corresponding to a forbidden assignment  $\sigma\in\mathcal{N}_B$.
Therefore, any general LLL instance with $p = \max_{B \in \+B}\Pr{B}$ and maximum degree $D$ of the dependency graph, can be equivalently represented as an LLL instance with atomic bad events, by blowing up each bad event $B\in\+B$ for at most $N\triangleq\max_{B\in\+B}\abs{\mathcal{N}_B}$ times.
The resulting LLL instance with atomic bad events can be constructed within  $\widetilde{O}(DNkn)$ time, such that every atomic bad event occurs with probability at most $p$ and has the degree of dependency at most $(D+1)N$.
%where such blow up may increase the maximum degree of the dependency graph by at most $N$ times while does not increase the probability of a bad event,
%
Hence, we have the following corollary.
\begin{corollary}
\label{corollary-main}
The following holds for any $0 < \zeta \leq 2^{-400}$.
There is an algorithm such that given a Lov\'asz local lemma instance with uniform random variables, 
if 
\begin{align*}
\ln \frac{1}{p} \geq  350 \ln (D+1) + 350 \ln N + 3 \ln \frac{1}{\zeta}, 
\end{align*}
%$p \leq  \frac{\zeta^3}{(DN)^{1000}}$,
then the  algorithm  outputs a random assignment $\Ass{X} \in \bigotimes_{v \in V}Q_v$ in time $\widetilde{O}\tp{(D^2 N^2 k+q) n \tp{\frac{n}{\epsilon}}^{\zeta}}$
such that 
the distribution of $\Ass{X}$ is $\epsilon$-close to the LLL-distribution in total variation distance,
where $q = \max_{v \in V}\abs{Q_v}$, $k = \max_{B \in \+B} \abs{\vbl{B}}$, and $\widetilde{O}(\cdot)$ hides a factor of $\mathrm{polylog}(n,\frac{1}{\epsilon}, q, D,N)$.
%$q = \max_{v \in V}\abs{Q_v}$ and $k = \max_{B \in \+B} \abs{\vbl{B}}$.
\end{corollary}
To the best of our knowledge, this is the first result that achieves efficient uniform sampling of general CSP solutions
within such a local lemma type of regime.
In the current result, both the regime and the complexity depend on an extra parameter $N$, namely the maximum number of violating local configurations for any bad event.
Whether such dependency is necessary is an open problem.

Our approach also produces sharper bounds for specific subclasses of LLL instances.
We consider the problem of uniformly sampling proper colorings of hypergraphes. 
Let $H=(V,\+E)$ be a $k$-uniform hypergraph i.e. $\abs{e} = k$ for all $e \in \+E$.
A proper hypergraph $q$-coloring $\Ass{X} \in [q]^V$ assigns each vertex a color such that no hyperedge is monochromatic. 
Let $\Delta$ denote the maximum degree of hypergraph, i.e.~each vertex belongs to at most $\Delta$ hyperedges. 
By LLL, a proper $q$-coloring exists if $q \geq C\Delta^{\frac{1}{k-1}}$ for some suitable constant $C$. 
%For each hyperedge $e \in \+E$, we define $q$ atomic bad events, each indicates the event that all vertex in $e$ take the same color $c \in [q]$.
We have  the following result for sampling hypergraph colorings.
\begin{theorem}
\label{theorem-coloring}
There is an algorithm such that given any $k$-uniform hypergraph on $n$ vertices with maximum degree $\Delta$ and a set of colors $[q]$, 
assuming $k\geq 30$ and $q \geq 15\Delta^{\frac{9}{k-12}} + 650$, 
the algorithm returns a random $q$-coloring $\Ass{X}\in [q]^V$
in time $\widetilde{O}(q^2k^3\Delta^2 n \tp{\frac{n}{\epsilon}}^{\frac{1}{q}})$, 
%where $\widetilde{O}(\cdot)$ hides a factor of $\mathrm{polylog}(n,\frac{1}{\epsilon}, q, \Delta,k)$,
such that 
the distribution of $\Ass{X}$ is $\epsilon$-close in total variation distance to the uniform distribution of all proper  $q$-colorings of the input hypergraph.
%
%a set of colors $[q]$ and a $k$-uniform hypergraph with maximum degree $\Delta$, if $k\geq 30$ and $q \geq 15\Delta^{\frac{9}{k-12}} + 650$, then the algorithm draws a random coloring $X\in [q]^V$ in time
%$\widetilde{O}(q^2k^3\Delta^2 n \tp{\frac{n}{\epsilon}}^{\frac{1}{q}})$
%such that 
%the distribution of $X$ is $\epsilon$-close in total variation distance to the uniform distribution of all proper  $q$-colorings, where $\widetilde{O}(\cdot)$ hides a factor of $\mathrm{polylog}(n,\frac{1}{\epsilon}, q, \Delta,k)$.
\end{theorem}
In fact, our algorithm works for a regime where $k\ge 13$ and $q\ge q_0(k)=\Omega(\Delta^{\frac{9}{k - 12}})$.
See~\Cref{theorem-coloring-gen} for a more technical statement.
The running time of our algorithm is always polynomially bounded for any bounded or unbounded $k$ and $\Delta$, and is getting arbitrarily close to linear in $n$ as $q$ grows.

%The condition in \Cref{theorem-coloring} is better than that in \Cref{corollary-main}, because we give a more careful analysis for concrete applications. 
%Actually, our algorithm works for $k \geq 13$ and $q = \Omega(\Delta^{\frac{9}{k - 12}})$. See~\Cref{theorem-coloring-gen} for a more technical statement.
%Our algorithm works in the local lemma regime, and the running time is a fixed polynomial in $n,\frac{1}{\epsilon}, q, \Delta$ and $k$. 
%Moreover, the exponent of $n$ is $1+\frac{1}{q}$, which implies the running time is more close to linear if more colors are provided. 

Hypergraph colorings are important combinatorial objects.
The classic local Markov chain on hypergraph colorings rapidly mixes in $O(n \log n)$ steps if $k \geq 4$  and $q > \Delta$~\cite{bordewich2006stopping,bordewich2008path}.
 For ``simple'' hypergraphs where any two hyperedges share at most one vertex, the mixing condition was improved to  $q \geq \max\{C_k \log n, 500k^3\Delta^{1/(k-1)}\}$~\cite{frieze2011randomly,frieze2017randomly}.
The first algorithm for sampling and counting hypergraph colorings that works in a local lemma regime was given in~\cite{guo2019counting}.
The algorithm is obtained by extending Moitra's approach~\cite{Moi19} to adaptively marking/unmarking hypergraph vertices, and runs in time $n^{\mathrm{poly}(\Delta k)}$ if $k \geq 28$ and $q > 798\Delta^{\frac{16}{k-16/3}}$.
Our algorithm both substantially improves the running time and improves the regime to $q \geq 15\Delta^{\frac{9}{k-12}} + O(1)$. 
Our algorithm utilizes a novel projection scheme instead of the mark/unmark strategy of Moitra, to transform the space of proper colorings.
And our algorithm implements a rapidly mixing Markov chain on the projected space.

A canonical subclass of CSPs are the CNF (conjunctive normal form) formulas.
In a $k$-CNF, each clause contains $k$ distinct variables.
And the maximum (variable-)degree $d$ is given by maximum number of clauses a variable appears in. 
%Each clause represents an atomic bad events.
By LLL, a satisfying assignment exists if $k \geq \log d + \log k + C$\footnote{Throughout the paper, we use $\log$ to denote the logarithm base 2.} for some suitable constant $C$.
We have the following result for uniform sampling $k$-CNF solutions.
%We give the following algorithm for sampling uniform $k$-CNF solutions. 
%In \Cref{theorem-main}, the condition in~\eqref{eq-condition-main-thm} requires $q > 20$.  
%In fact, our algorithm can also be applied to instances with small $q$.
%We give the following algorithm for sampling uniform $k$-CNF solutions. 
\begin{theorem}
\label{theorem-CNF}
The following holds for any $0< \zeta \leq 2^{-20}$. 
There is an algorithm such that given any $k$-CNF formula on $n$ variables with maximum degree $d$, 
%where each clause contains $k$ variables and each variable belongs to at most $d$ clauses, 
assuming \mbox{$k \geq 13 \log d + 13 \log k + 3 \log \frac{1}{\zeta}$},
the algorithm returns a random assignment $\Ass{X}\in \{\True,\False\}^V$ 
in time
$\widetilde{O}(d^2k^3 n \tp{\frac{n}{\epsilon}}^{\zeta})$
such that 
the distribution of $\Ass{X}$ is $\epsilon$-close in total variation distance to the uniform distribution of all satisfying assignments. %where $\widetilde{O}(\cdot)$ hides a factor of $\mathrm{polylog}(n,\frac{1}{\epsilon},d,k)$.
\end{theorem}
%The condition in \Cref{theorem-CNF} is also better than that in \Cref{theorem-main}.
A more detailed version is stated  as \Cref{theorem-CNF-gen}.
%Actually, \Cref{theorem-CNF} is a simplified version of more technical statement in \Cref{theorem-CNF-gen}.
%The running time of our algorithm is a fixed polynomial. 
%If the parameter $\zeta$ approaches to $0$, the condition in \Cref{theorem-CNF} becomes stronger, and the running time becomes more close to linear.
The regime $k \gtrsim 13 \log d$ in \Cref{theorem-CNF} improves the state-of-the-art regime $k \gtrsim 20 \log d$ in~\cite{FGYZ20} with the same running time.

\subsection{Implications to approximate counting}
All our sampling results imply efficient algorithms for approximate counting.
Given an LLL instance $\Phi$ with uniform random variables, let $Z_{\Phi}$ denote the total number of satisfying assignments that avoid all bad events.
For any $0<\delta <1$, the problem $\+P_{\mathrm{count}}(\Phi,\delta)$ asks to output a random number $\widehat{Z}$ such that $ \widehat{Z} \in (1 \pm \delta)Z_{\Phi}$ with probability at least $\frac{3}{4}$.

%Given an LLL instance $\Phi$ with uniform random variables, let $\mu_{\Phi}$ denote its LLL-distribution, which is the uniform distribution over all satisfying assignments that avoid all the bad events.
%
In our results~(\Cref{theorem-main},
\Cref{corollary-main},
\Cref{theorem-coloring},
and \Cref{theorem-CNF}),
for several subclasses of LLL instances, we give such sampling algorithms  that given an LLL instance $\Phi$ and an error bound $\epsilon>0$,  a random $\Ass{X}$ is returned in time $T(\epsilon)=T_{\Phi}(\epsilon)$ such that $\Ass{X}$ is $\epsilon$-close in total variation distance to the LLL-distribution of $\Phi$, which is the uniform distribution over all satisfying assignments for $\Phi$.

It is well known that one can solve the approximate counting problem $\+P_{\mathrm{count}}(\Phi,\delta)$ by calling to such oracles for nearly uniform sampling, either via the self-reducibility~\cite{jerrum1986random} that adds one bad event at a time, or via the simulated annealing approach~\cite{bezakova2008accelerating,vstefankovivc2009adaptive,huber2015approximation,kolmogorov18faster} that alters a temperature.
The simulated annealing gives more efficient reduction.
Specifically, by routinely going through the annealing process in~\cite{FGYZ20}, one can obtain a non-adaptive simulated annealing strategy to solve the approximate counting problem $\+P_{\mathrm{count}}(\Phi,\delta)$ in time 
$O\tp{\frac{m}{\delta^2}T(\epsilon)\log \frac{m}{\delta}}$, where $\epsilon = \Theta\tp{ \frac{\delta^2}{m \log (m / \delta)} } $, and $m$ denotes the number of bad events in $\Phi$.

\subsection{Technique overview} 
\label{section-overview}
As addressed in~\cite{wigderson2019book}, in general, the space of SAT solutions may not be connected via local updates of variables, even when the existence of SAT solutions is guaranteed by the local lemma.
A major challenge for efficiently sampling constraint satisfaction solutions in a local lemma regime is to bypass such connectivity barrier.

Several previous works that have successfully bypassed this fundamental barrier fell into the same ``mark/unmark'' paradigm initiated by Moitra in~\cite{Moi19}.
Let $V$ be the set of variables, and let $\mu$ denote the uniform distribution over all satisfying assignments. 
The paradigm effectively constructs a random pair $(M,X_M)$ where $M\subseteq V$ is a set of marked variables and $X_M$ is a random assignment of the marked variables in $M$, 
such that the random pair $(M,X_M)$ satisfies the so-called ``\emph{pre-Gibbs}'' property~\cite{guo2019counting}, which means that if we complete $X_M$ to an assignment $\Ass{X}$ of all variables in $V$ by sampling the complement $X_{V\setminus M}$ according to the marginal distribution induced by $\mu$ on ${V\setminus M}$ conditioning on $X_M$, then the resulting $\Ass{X}$ indeed follows the correct distribution $\mu$.
The paradigm may construct the marked set $M$ either non-adaptively to the random $X_M$ (as in~\cite{Moi19,FGYZ20,galanis2019counting} for CNFs), or adaptively to it (as in~\cite{guo2019counting} for hypergraph colorings).
The random pair $(M,X_M)$ can thus be jointly distributed, so that being pre-Gibbs does not necessarily mean that $X_M$ is distributed as the marginal distribution~$\mu_M$.
Indeed, it can be much more complicated than that.

In this paper, we introduce a novel technique called ``\textbf{\emph{state compression}}'' to bypass the connectivity barrier for general spaces of satisfaction solutions and obtain fast sampling algorithms.

For each variable $v\in V$ with domain $Q_v$, we construct a projection $h_v:Q_v\to\Sigma_v$ that maps from domain $Q_v$ to an alphabet $\Sigma_v$,
so that each assignment $\Ass{X}\in \Dom{Q} \triangleq \bigotimes_{v \in V}Q_v$ is mapped to a string $\Proj{h}(\Ass{X})\triangleq(h_v(\Ass{X}_v))_{v \in V}$ in $\Dom{\Sigma} \triangleq \bigotimes_v\Sigma_v$. 
Therefore, the LLL-distribution $\mu$ over satisfying assignments, is transformed to a joint distribution $\nu$ over $\Dom{\Sigma}$ as:
\[
\forall \Ass{Y} \in \Sigma, \quad \nu(\Ass{Y}) = \Pr[\Ass{X} \sim \mu]{ \Proj{h}(\Ass{X}) = \Ass{Y}}. 
\]

Our algorithm first simulates the Glauber dynamics with stationary distribution $\nu$ to draw a sample $\Ass{Y} \in \Sigma$ approximately according to $\nu$.
%The Glauber dynamics starts from a proper initial $\Ass{Y} \in \Sigma$, 
%then repeats the following transition  for $O(n \log n)$ steps:
At each transition, the Glauber dynamics:
\begin{itemize}
\item picks a variable $v$ uniformly at random;
\item updates $Y_v$ by a random value sampled according to $\nu^{Y_{V \setminus \{v\}}}_v$, %where $\nu^{Y_{V \setminus \{v\}}}_v$ is 
which stands for the marginal distribution at $v$ induced by $\nu$ conditioned on the assignment on $V \setminus \{v\}$ being fixed as $Y_{V \setminus \{v\} }$.
\end{itemize}
After running the Glauber dynamics for a sufficiently many $O(n\log n)$ steps, 
the algorithm generates a random string $\Ass{Y} \in \Sigma$ which hopefully is distributed approximately as $\nu$.
Finally, the algorithm still needs to ``invert'' the sampled string $\Ass{Y} \in \Sigma$ to a random satisfying assignment $\Ass{X} \in \Dom{Q}$ that follows the LLL-distribution $\mu$ conditioning on $\Proj{h}(\Ass{X})=\Ass{Y}$.

Both in the final step of the algorithm and at each transition of the Glauber dynamics, we are in fact trying to invert a completely specified string $\Ass{Y} \in \Sigma$ (or an almost completely specified string $Y_{V \setminus \{v\}}$) to a uniform random satisfying assignment $\Ass{X} \in \Dom{Q}$ within its pre-image $\Proj{h}^{-1}(\Ass{Y})$ (or that of $Y_{V \setminus \{v\}}$).
%
%Note that sampling from the marginal distribution $\nu^{Y_{V \setminus \{v\}}}_v$ can be achieved by the above inverse sampling because once we have sampled an $X_v\in Q_v$ in an assignment $\Ass{X}$ inverted from $Y_{V \setminus \{v\}}$ as described above, we can map it to $Y_v=h_v(X_v)$ which is distributed as $\nu^{Y_{V \setminus \{v\}}}_v$.

Therefore, the efficiency of above algorithmic framework for sampling  relies on that:
\begin{enumerate}
\item\label{item:rapid} the Glauber dynamics for $\nu$ mixes in $O(n\log n)$ steps;
\item\label{item:implement} there is a procedure that can efficiently invert a completely (or almost completely) specified string $\Ass{Y}$ to a uniform random satisfying assignment $\Ass{X} \in \Dom{Q}$ within the pre-image $\Proj{h}^{-1}(\Ass{Y})$.
\end{enumerate}

As we know, the original space of satisfying assignments $\Ass{X}\in\Dom{Q}$ may not be connected via the local updates used by the Glauber dynamics.
To achieve above \cref{item:rapid}, intuitively, the projection $\Proj{h}$ should be able to map many far-apart solutions $\Ass{X}, \Ass{X}' \in\Dom{Q}$ to the same $\Proj{h}(\Ass{X})=\Proj{h}(\Ass{X}')$, so the random walk in the projected space becomes well connected.
This suggests that
%in order to ensure the Glauber dynamics on the projected space  mixing rapidly, 
%\begin{quote}
\emph{the projection $\Proj{h}$ should substantially compress the original state space}.
%\end{quote}
%
On the other hand, the above \cref{item:implement} is easier to solve when the projection $\Proj{h}$  is somehow close to a one-to-one mapping, because in such case, by assuming $\Proj{h}(\Ass{X})=\Ass{Y}$, the original LLL instance is very likely to be decomposed into small clusters.
%
%In particular, the inverse sampling asked by \cref{item:implement} becomes trivial to solve when $\Proj{h}$ is a one-to-one mapping.
%
%And we will see that if $\Proj{h}$  is close to a one-to-one mapping, conditioning on a projected solution $\Ass{Y}$, the LLL instance is very likely to be decomposed into disjoint clusters.
%
This suggests that
%\begin{quote}
\emph{the projection $\Proj{h}$ should not compress the original state space too much}.
%\end{quote}

The above two seemingly contradicting requirements can in fact be captured by a set of simple and local entropy constraints, formulated in~\Cref{condition-projection}.
A good projection $\Proj{h}$ satisfying these requirements can thus be constructed by algorithmic LLL.

The original mark/unmark paradigm can be treated as a special case of our approach of state compression.
Recall that the paradigm generates a pre-Gibbs pair $(M,X_M)$, where each variable $v\in V$ is either marked ($v\in M$) so that its value $X_v$ is revealed, or is unmarked ($v\not\in M$) so that its value $X_v$ is unrevealed.
This can be represented by a projection $\Proj{h}$ where for each marked $v$, the projection $h_v:Q_v\to\Sigma_v$ is a one-to-one mapping to $\Sigma_v$ where $\abs{\Sigma_v}=\abs{Q_v}$; and for each unmarked $v$, the projection $h_v:Q_v\to\Sigma_v$ is a all-to-one mapping to $\Sigma_v$ of size $\abs{\Sigma_v}=1$.
General projections provide a broad middle ground 
between the two extremal cases for the one-to-one and the all-to-one mappings, 
so that our technique is applicable to more general settings.
And for large enough $Q_v$'s, it indeed is such middle ground $h_v:Q_v\to\Sigma_v$ with $\abs{\Sigma_v}\approx\abs{Q_v}^{3/4}$ 
that resolves the problem well.

\subsection{Open problems}
An open problem is to remove the assumption on the atomicity of bad events.
In general, the LLL is defined by arbitrary bad events on arbitrary probability space.
The LLL distribution can thus be generalized. 
%as the distribution obtained by conditioning on none of the bad event occurs.
%
And the sampling LLL corresponds to the problems of sampling from non-uniform distributions or distributions arising  from  global constraints.
%
%The latter is highly open for sampling algorithms.

It is well-known that the Shearer’s bound is tight for general LLL~\cite{shearer85}.
A central open problem for sampling LLL is to find the ``Shearer’s bound'' for sampling LLL, namely, to give a tight condition under which one can efficiently draw random samples from general LLL distributions.

Even for interesting special classes of LLL instances such as $k$-CNFs or hypergraph colorings, the critical thresholds for the computational phase transition for sampling  are major open problems in the field of sampling algorithms.

\subsection{Organization of the paper}
Models and preliminaries are described in \Cref{section-model}.
The rules for state compression are given in \Cref{section-proj-of-LLL}.
The main sampling algorithm is described in \Cref{section-algorithm}.
In \Cref{section-proof-main}, we prove all main results in \Cref{section-intro}.
In \Cref{section-projection-construction}, we give the algorithms for constructing projections.
In \Cref{section-subroutine}, we analyze the inverse sampling subroutine.
The rapid mixing of the Markov chain is proved in \Cref{section-mixing}.

\pagebreak
\section{Models and preliminaries}
\label{section-model}
\subsection{CSP formulas defined by atomic bad events}
Let $V$ be a set of variables with finite domains $(Q_v)_{v\in V}$, where each $v \in V$ takes its value from $Q_v$ with $\abs{Q_v}\geq 2$.
Let $\Dom{Q}\triangleq\bigotimes_{v \in V}Q_v$ denote the space for all assignments, and for any subset $\Lambda\subseteq V$, denote $Q_{\Lambda}\triangleq\bigotimes_{v \in \Lambda}Q_v$.
Let $\Cons{C}$ be a collection of \emph{local constraints}, where each $c\in\Cons{C}$ is defined on a subset of variables $\vbl{c}\subseteq V$ that  maps every assignment $\Ass{x}_{\vbl{c}}\in Q_{\vbl{c}}$ to a $\True$ or $\False$, which indicates whether $c$ is  \emph{satisfied} or \emph{violated}.
A CSP (constraint-satisfaction problem) formula $\Phi$ is specified by the tuple $(V, \Dom{Q}, \Cons{C})$ such that:
\begin{align*}
\forall\Ass{x}\in\Dom{Q},\qquad \Phi(\Ass{x}) = \bigwedge_{c \in \Cons{C}}c\left(\Ass{x}_{\vbl{c}}\right),
\end{align*}
where $\Ass{x}_{\vbl{c}}$ denotes the restriction of $\Ass{x}$ on $\vbl{c}$.
In LLL's language, each $c\in\Cons{C}$ corresponds to a \emph{bad event} $A_c$ defined on $\vbl{c}$ that occurs if $c$ is violated, and $\Phi$ is satisfied by $\Ass{x}$ if and only if none of these bad events occurs.

In this paper, we restrict ourselves to the CSP formulas defined by {atomic bad events}. 
A constraint $c$ defined on $\vbl{c}$ is called \emph{atomic} if $|c^{-1}(\False)|=1$, that is, if $c$ is violated by a unique ``forbidden configuration'' in $Q_{\vbl{c}}$.
%there is a unique $\sigma_{\vbl{c}}\in\Omega_{\vbl{c}}$ such that
%\begin{align*}
%c\left(\Ass{x}_{\vbl{c}}\right) = 
%\begin{cases}
%\False & \text{if }\Ass{x}_{\vbl{c}}=\sigma_{\vbl{c}},\\
%\True & \text{otherwise.}
%\end{cases}
%\end{align*}
%
Such CSP formulas with atomic constraints have drawn studies in the context of LLL~\cite{achlioptas2016random,HH17,harris2017algorithmic,kolmogorov2018commutativity,Harris19Obl,achlioptas2019beyond,HS19,harvey2020algorithmic}. 
Similar classes of CSP formulas have also been studied under the name ``{multi-valued/non-Boolean CNF formulas}'' in the field of classic Artificial Intelligence~\cite{LiuKM03,FrischP01}.
Clearly, any general constraint $c$ on $\vbl{c}$ can be simulated by $|c^{-1}(\False)|$ atomic constraints, each forbidding a configuration in $c^{-1}(\False)$.

The \emph{dependency graph} of a CSP formula $\Phi=(V,\Dom{Q},\Cons{C})$ is defined on the vertex set $\Cons{C}$, such that any two constraints $c,c'\in C$ are adjacent if $\vbl{c}$ and $\vbl{c'}$ intersect.
We use $\Gamma(c)\triangleq\{c'\in\Cons{C}\setminus\{c\}\mid  \vbl{c} \cap \vbl{c'} \neq \emptyset\}$ to denote the neighborhood of $c\in\Cons{C}$ and let
\begin{align*}
D =D_{\Phi}\triangleq \max_{c \in \Cons{C}}\abs{\Gamma(c)}	
\end{align*}
denote the maximum degree of the dependency graph.

The followings are some typical special cases of CSP formulas with atomic constraints.

\subsubsection{$k$-CNF formula}
The CNF formulas $\Phi=(V,\Dom{Q},\Cons{C})$ are formulas with atomic constraints on Boolean domains $Q_v=\{\True,\False\}$, for all $v\in V$.
Now each constraint $c\in\Cons{C}$ is a \emph{clause}.
For $k$-CNF formulas, we have $\abs{\vbl{c}}=k$ for all clauses $c\in\Cons{C}$.

\subsubsection{Hypergraph coloring}
Let $H=(V,\+E)$ be a $k$-uniform hypergraph, where every hyperedge $e\in\+E$ has $|e|=k$.
Let $[q] = \{1,2,\ldots,q\}$ be a set of $q$ colors.
A proper hypergraph coloring $\Ass{X} \in [q]^V$ assigns each vertex $v \in V$ a color $X_v$ such that no hyperedge is monochromatic.

Define the following set  $\Cons{C}$ of atomic constraints. 
For each hyperedge $e \in \+E$ and color $i \in  [q]$, add an atomic constraint  $c_{e,i}$ into $\Cons{C}$, where $c_{e,i}$ is defined as $\vbl{c_{e,i}} = e$ and for any $\Ass{x}\in[q]^e$, $c_{e,i}(\Ass{x}) = \False $ if and only if ${x}_v = i$ for all $v \in e$.
It is straightforward to see that there is a one-to-one correspondence between the proper $q$-colorings in $H$ and the satisfying assignments to $\Phi = (V,[q]^V, \Cons{C})$.

\subsection{Lov\'asz local lemma}
Let $\mathcal{R}=\{R_1,R_2,\ldots,R_n\}$ be a collection of mutually independent random variables.
For any event $E$, denote by $\vbl{E}\subseteq \+R$ the set of variables determining $E$.
In other words, changing the values of variables outside of $\vbl{E}$ does not change the truth value of $E$.
Let $\mathcal{B} = \{B_1,B_2,\ldots, B_n\}$ be a collection of ``bad'' events.
%
%Each event $B_i$ is determined by a set of variables, denoted by $\vbl{B_i} \subseteq \mathcal{P}$.
%
For each event $B \in \mathcal{B}$, we define $\Gamma(B) \defeq \set{B'\in \mathcal{B} \mid B' \neq B \text{ and } \vbl{B'} \cap \vbl{B} \ne \emptyset }$.
For any event $A\notin \mathcal{B}$ and its determining variables $\vbl{A} \subseteq \mathcal{R}$,
%where $A$ is determined by the variables in $\vbl{A} \subseteq \mathcal{P}$, 
we define $\Gamma(A) \triangleq \{B \in \mathcal{B} \mid \vbl{A} \cap \vbl{B} \ne \emptyset \}$.
%
%Let $\Phi = (V, C)$ be a CNF formula. 
%%
%For each clause $c \in \Cons{C}$, we use 
%$\Gamma(c) \triangleq \{ b \in C \mid b \neq c \land \vbl{b} \cap \vbl{c} \neq \emptyset \}$
%to denote the set of other clauses that intersect $c$. 
%%
%For any event $B$ that depends only on the set of variables $\vbl{B}$, we use $\Gamma(B) \triangleq \{ c \in \Cons{C} \mid \vbl{c} \cap \vbl{B} \neq \emptyset \}$ to denote the set of clauses that contain any of the variables in $\vbl{B}$.
%%
%Let $\Pr[\cdot]$ denote the product distribution that every variable in $V$ takes a value from $\{0,1\}$ uniformly and independently.
%
Let $\Pr[\+D]{\cdot}$ denote the product distribution of variables in
$\mathcal{R}$. %The following version of the Lov\'asz local lemma is from ~\cite{haeupler2011new}.
The following version of the Lov\'asz local lemma will be used in this paper.
\begin{theorem}[~\cite{haeupler2011new}]
\label{theorem-LLL}
If there is a function $x: \mathcal{B} \rightarrow (0,1)$ such that for any $B \in \mathcal{B}$,
\begin{align}\label{eqn:LLL}
  \Pr[\+D]{B} \leq x(B) \prod_{B' \in \Gamma(B)}(1-x(B')),	
\end{align}
then it holds that
\begin{align*}
  \Pr[\+D]{\bigwedge_{B \in \mathcal{B}} \overline{B} } \geq \prod_{B \in \mathcal{B}}(1 - x(B)) > 0. 	
\end{align*}
Thus, there exists an assignment of all variables that avoids all the bad events.

Moreover, for any event $A$, it holds that
\begin{align*}
\Pr[\+D]{A \,\big|\, \bigwedge_{B \in \mathcal{B}} \overline{B} } \leq \Pr[\+P]{A} \prod_{B \in \Gamma(A)}(1- x(B))^{-1}.	
\end{align*}
\end{theorem}

\subsection{Coupling, Markov chain and mixing time}
Let $\Omega$ be a state space. Let $\mu$ and $\nu$ be two distributions over $\Omega$.
The \emph{total variation distance} between $\mu$ and $\nu$ are defined by
\begin{align*}
\DTV{\mu}{\nu} \triangleq \frac{1}{2} \sum_{x \in \Omega}\abs{\mu(x) - \nu(x)}.	
\end{align*}
A coupling of $\mu$ and $\nu$ is a joint distribution $(X, Y) \in \Omega \times \Omega$ such that the marginal distribution of $X$ is $\mu$ and the marginal distribution of $Y$ is $\nu$.
The following coupling lemma is well-known. 
\begin{lemma}[coupling lemma~\text{\cite[Proposition~4.7]{levin2017markov}}]
\label{lemma-coupling-ineq}
For any coupling $(X,Y)$ between $\mu$ and $\nu$,
\begin{align*}
\DTV{\mu}{\nu} \leq \Pr[]{X \neq Y}.	
\end{align*}
Moreover, there exists an \emph{optimal coupling} that achieves the equality.
\end{lemma}

A \emph{Markov chain} is a random sequence $(X_t)_{t \geq 0}$ over a state space $\Omega$ such that the transition rule is specified by the \emph{transition matrix} $P: \Omega \times \Omega \to \mathbb{R}_{\geq 0}$.
We often use the transition matrix to denote the corresponding Markov chain.
The Markov chain $P$ is \emph{irreducible} if for any $X, Y \in \Omega$, there exists $t > 0$ such that $P^t(X, Y) > 0$. 
The Markov chain $P$ is \emph{aperiodic} if $\gcd\{t \mid P^t(X,X) > 0\} = 1$ for all $X \in \Omega$.
A distribution $\pi$ over $\Omega$ is a \emph{stationary distribution} of $P$ if $\pi P = \pi$.
If a Markov chain is irreducible and aperiodic, then it has a unique stationary distribution.
The Markov chain $P$ is \emph{reversible} with respect to the distribution $\pi$ if the following \emph{detailed balance equation} holds
\begin{align*}
\forall X, Y \in \Omega: \quad \pi(X)P(X, Y) = \pi(Y)P(Y,X),	
\end{align*}
which implies $\pi$ is a stationary distribution of $P$.
Given a Markov chain $P$ with the unique stationary distribution $\pi$, the \emph{mixing time} of $P$ is defined by
\begin{align*}
\forall 0 < \epsilon < 1, \quad \tmix(\epsilon) \triangleq \max_{X_0 \in \Omega} \min\{t \mid \DTV{P^t(X_0,\cdot)}{ \pi} \leq \epsilon \}.	
\end{align*}

A coupling of Markov chain $P$ is a joint random process $(X_t,Y_t)_{t \geq 0}$ such that both $(X_t)_{t \geq 0}$ and $(Y_t)_{t \geq 0}$ follow the transition rule of $P$ individually, and if $X_s = Y_s$, then $X_k = Y_k$ for all $k \geq s$. The coupling is a widely-used tool to bound the mixing times of Markov chains, because by the coupling lemma, it holds that $\max_{X_0 \in \Omega}\DTV{P^t(X_0,\cdot)}{\pi} \leq \max_{X_0,Y_0 \in \Omega}\Pr[]{X_t \neq Y_t}$.

The \emph{path coupling}~\cite{bubley1997path} is a powerful tool to construct the coupling of Markov chains. Assume $\Omega = \bigotimes_{v \in V}Q_v$, where $\abs{V} = n$ and each $Q_v$ is a finite domain.  
For any $X,Y \in \Omega$,
define the \emph{Hamming distance} between $X$ and $Y$ by
\begin{align*}
d_{\mathrm{ham}}(X,Y) \triangleq \abs{\{v \in V \mid X_v \neq Y_v\}}.	
\end{align*}
In this paper, we will use the following simplified version of path coupling.
\begin{lemma}[path coupling~\cite{bubley1997path}]
\label{lemma-path-coupling}
Let $0 < \delta < 1$ be a parameter. 
Let $P$ be an irreducible and aperiodic Markov chain over the state space $\Omega = \bigotimes_{v \in V}Q_v$, where $\abs{V} = n$ . If there is a coupling of Markov chain $(X, Y) \to (X',Y')$ defined over all $X, Y \in \Omega$ with $d_{\mathrm{ham}}(X, Y) = 1$ such that 
\begin{align*}
\E[]{d_{\mathrm{ham}}(X', Y') \mid X, Y} \leq 1 - \delta,
\end{align*}
then the mixing time of the Markov chain satisfies
\begin{align*}
\tmix(\epsilon) \leq \ctp{\frac{1}{\delta}\log\frac{n}{\epsilon}}.	
\end{align*}
\end{lemma}
Readers can refer to the textbook~\cite{levin2017markov} for more backgrounds of Markov chains and mixing times.

\pagebreak

\section{state compression}
\label{section-proj-of-LLL}
A CSP formula $\Phi=(V,\Dom{Q},\Cons{C})$ with uniformly distributed random variables defines an LLL instance.
\begin{definition}[LLL-distribution]
\label{def:LLL-distr}
For each $v\in V$, let $\pi_v$ denote the uniform distribution over domain $Q_v$. 
Let $\pi\triangleq\bigotimes_{v\in V}\pi_v$ be the uniform distribution over $\Dom{Q}$.
Let $\mu=\mu_{\Phi}$ denote the distribution of $\Ass{X}\sim\pi$ conditioned on $\Phi(\Ass{X})$, that is, the uniform distribution over satisfying solutions of $\Phi$.
\end{definition}
This distribution $\mu$ over satisfying solutions of $\Phi$ is what we want to sample from.
In order to do so, this uniform probability space of satisfying solutions is transformed by a projection.
A \emph{projection scheme} $\Proj{h}=(h_v)_{v\in V}$ specifies for each $v\in V$, a mapping from $v$'s domain $Q_v$ to a finite \emph{alphabet} $\Sigma_v$:
\begin{align*}
h_v:Q_v\to\Sigma_v.
\end{align*}
Let $\Dom{\Sigma}\triangleq\bigotimes_{v\in V} \Sigma_v$, 
and for any $\Lambda\subseteq V$, we denote $\Sigma_\Lambda\triangleq\bigotimes_{v\in \Lambda} \Sigma_v$.
%In particular, we write $\Dom{\Sigma}=\Sigma_V$.
%
%Every assignment $\Ass{x}\in\Omega$ is then mapped to a \emph{configuration} $\Proj{h}(\Ass{x})\triangleq (h_v(x_v))_{v\in V}\in\Sigma$. 
%
%For any $\Lambda\subseteq V$ and $\Ass{x}\in Q_\Lambda$, we define $\Proj{h}(\Ass{x})\triangleq(h_v(x_v))_{v\in \Lambda}$.

We  also naturally interpret $\Proj{h}$ as a function on (partial) assignments such that %for any $\Lambda\subseteq V$, %and $\Ass{x}\in {Q}_{\Lambda}$,
\begin{align}
\label{eq-def-proj-set}
\forall \Lambda\subseteq V, \forall \Ass{x}\in Q_\Lambda,\qquad
%\forall\Ass{x}\in {Q}_\Lambda,\qquad
\Proj{h}(\Ass{x})\triangleq(h_v(x_v))_{v\in \Lambda}.
\end{align}
%Moreover, for any $\Lambda\subseteq V$ and $\Ass{x}\in {Q}_{\Lambda}$, let $\Proj{h}(\Ass{x})\triangleq(h_v(x_v))_{v\in \Lambda}$.

%A formula $\Phi=(V,\Dom{Q},\Cons{C})$ naturally gives rise to a uniform probability space of satisfying solutions. 
%
%This probability space is changed by the projection scheme $\Proj{h}=(h_v)_{v\in V}$.

%This gives rise to several natural distributions.

\begin{definition}[projected LLL-distribution]
\label{def:proj-distr}
For each $v\in V$, let $\rho_v$ be the distribution of $Y_v=h_v(X_v)$ where $X_v\sim\pi_v$. 
%denote the distribution projected by $h_v$ onto $\Sigma_v$ from the uniform distribution $\pi_v$ over $Q_v$, i.e.~$\rho_v$ is the distribution of $Y_v=h_v(X_v)$ where $X_v\sim\pi_v$. 
%
Let $\rho\triangleq\bigotimes_{v\in V}\rho_v$ be the product distribution over $\Dom{\Sigma}$.

For each $v\in V$ and any $y_v\in\Sigma_v$, let $\pi_v^{y_v}$ denote the distribution of $X_v\sim\pi_v$ conditioned on $h_v(X_v)=y_v$.
For any $\Lambda\subseteq V$ and ${y}_\Lambda\in\Sigma_\Lambda$, let $\pi^{{y}_{\Lambda}}$ be the distribution of $\Ass{X}\sim\pi$ conditioned on $\Proj{h}(\Ass{X}_{\Lambda})={y}_{\Lambda}$.
%$\triangleq\bigotimes_{v\in\Lambda}\pi_v^{y_v}\bigotimes_{v\in V\setminus\Lambda}\pi_v$.

Let $\nu=\nu_{\Phi,\Proj{h}}$ denote the distribution of $\Ass{Y}=\Proj{h}(\Ass{X})$ where $\Ass{X}\sim\mu$. 
\end{definition}

Note that the original LLL-distribution $\mu$ is a \emph{Gibbs} distribution~\cite{mezard2009information}, defined by local constraints on independent random variables.
Whereas, the distribution $\nu$ of projected satisfying solution, is a joint distribution over $\Sigma$, which may no longer be a Gibbs distribution nor can it be represented as any LLL instance, because $\Ass{x},\Ass{x}'\in Q$ with $\Phi(\Ass{x})\neq\Phi(\Ass{x}')$ may be mapped to the same $\Proj{h}(\Ass{x})=\Proj{h}(\Ass{x}')$. 
%both satisfying solutions and violating solutions to the original formula $\Phi$ may be mapped to the same configuration in $\Sigma$ by the projection scheme $\Proj{h}$.
%\end{remark}

In the algorithm, a projection scheme $\Proj{h}=(h_v)_{v\in V}$ is accessed through the following oracle.
\begin{definition}[projection oracle]\label{def:projection-oracle}
A \emph{projection oracle} with query cost $t$ for a projection scheme $\Proj{h}=(h_v)_{v\in V}$ is a data structure that can answer each of the following two types of queries within time $t$:
\begin{itemize}
\item \emph{evaluation}: given an input value $x_v \in Q_v$ of a variable $v\in V$, output  $h_v(x_v) \in \Sigma_v$;
\item \emph{inversion}: given a projected value $y_v\in \Sigma_v$ of a variable $v \in V$, return a random $X_v\sim\pi_v^{y_v}$.
\end{itemize}
\end{definition}

Our algorithm for sampling a uniform random satisfying solution is then outlined below.

  \par\addvspace{.5\baselineskip}
%\noindent
\framebox{
  \noindent
  \begin{tabularx}{14cm}{@{\hspace{\parindent}} l X c}
    \multicolumn{2}{@{\hspace{\parindent}}l}{\underline{Algorithm for sampling from $\mu$}} \\% Title
  1. & Construct a good projection scheme $\Proj{h}$ (formalized by Condition~\ref{condition-projection});\\
% 2. & sample a random $\Ass{Y}\in\Dom{\Sigma}$ according to the product distribution $\rho$;\\
  2. & sample a uniform random $\Ass{X}\sim\pi$ and let $\Ass{Y}=\Proj{h}(\Ass{X})$;\\
  3. & (Glauber dynamics on $\nu$) repeat the followings for sufficiently many iterations:\\
      & \qquad pick a $v\in V$ uniformly at random;\\
      & \qquad update $Y_v$ by redrawing its value independently according to $\nu_v^{Y_{V\setminus\{v\}}}$;\\
  4. & sample $\Ass{X}\sim\mu$ conditioned on $\Proj{h}(\Ass{X})=\Ass{Y}$.
  \end{tabularx}
 }
\par\addvspace{.5\baselineskip}

The algorithm simulates a Markov chain (known as the Glauber dynamics) on space $\Dom{\Sigma}$ for drawing a random configuration $\Ass{Y}\in\Sigma$ approximately according to the joint distribution $\nu$,
after which, the algorithm ``inverts'' $\Ass{Y}$ to a uniform random satisfying assignment $\Ass{X}$ for $\Phi$ within the pre-image $\Proj{h}^{-1}(\Ass{Y})$.

The key to the effectiveness of this sampling algorithm is that we should be able to sample accurately and efficiently from $\nu_v^{Y_{V\setminus\{v\}}}$ (which is the marginal distribution at $v$ induced by $\nu$ conditioning on that the configuration on ${V\setminus\{v\}}$ being fixed as $Y_{V\setminus\{v\}}$) as well as from $\mu^{\Ass{Y}}$ (which is the distribution of $\Ass{X}\sim\mu$ conditioned on that $\Proj{h}(\Ass{X})=\Ass{Y}$). 
In fact, both of these are realized by sampling generally from the following marginal  distribution $\mu_S^{y_{\Lambda}}$, for $S\subseteq V$ and $y_{\Lambda}\in\Sigma_{\Lambda}$, where either $\Lambda= V$ or $|\Lambda|= |V|-1$.
\begin{align}
\mu_{S}^{{y}_{\Lambda}}
&:
\text{ distribution of ${X}_S$, where $\Ass{X}\in\Dom{Q}$ is drawn from $\mu$ conditioning on that $\Proj{h}({X}_{\Lambda})={y}_{\Lambda}$.}\label{eq:conditional-marginal-distribution}
\end{align}
The distribution $\mu^Y$ corresponds to the special case of $\mu_{S}^{{y}_{\Lambda}}$ with $S=\Lambda=V$.
And also we can sample from $\nu_v^{Y_{V\setminus\{v\}}}$ by first sampling a $X_v\sim\mu_{v}^{Y_{V\setminus\{v\}}}\triangleq \mu_{\{v\}}^{Y_{V\setminus\{v\}}}$ and then outputting $h_v(X_v)$.

Since $y_{\Lambda}$ is either completely or almost completely specified on $V$, sampling from $\mu_{S}^{{y}_{\Lambda}}$ is essentially trying to invert $y_{\Lambda}$ according to distribution $\mu$.
And this task becomes tractable when the projection $\Proj{h}$ is somehow close to a 1-1 mapping, i.e.~when $\Proj{h}(\Ass{X})$'s entropy remains significant compared to $\Ass{X}\sim\mu$.

On the other hand, the efficiency of the sampling algorithm relies on the mixing of the Markov chain for sampling from $\nu$.
It was known that the original state space of all satisfying solutions might not be well connected through single-site updates~\cite{wigderson2019book,FGYZ20}.
The projection may increase the connectivity of the state space by mapping many far-apart satisfying solutions to the same configuration in $\Sigma$, but this means that  the projection $\Proj{h}$ should not be too close to a 1-1 mapping.
In other words, the projection $\Proj{h}(\Ass{X})$ shall reduce the entropy of $\Ass{X}\sim\mu$ by a substantial amount.

These two seemingly contradicting requirements are formally captured by the following condition.

\begin{condition}[entropy criterion]
\label{condition-projection}
Let $0<\beta< \alpha<1$ be two parameters.%\ytodo{ Beware: Meaning of $\alpha$ has been changed to $1-\alpha$.}
The followings hold for the CSP formula $\Phi = (V,\Dom{Q},\Cons{C})$ and the projection scheme $\Proj{h}$.
For each $v\in V$, let $q_v\triangleq|Q_v|$ and $s_{v}\triangleq|\Sigma_v|$.
The projection $\Proj{h}$ is \emph{balanced}, which means for any $v\in V$ and $y_v\in \Sigma_v$, 
\begin{align*}
\ftp{\frac{q_v}{s_v}} \leq \abs{h^{-1}_v(y_v)} \leq \ctp{\frac{q_v}{s_v}}.
\end{align*}
And for any constraint $c \in \Cons{C}$, it holds that
\begin{align}
\sum_{v \in \vbl{c}}\log\ctp{\frac{q_v}{s_v}}	 &\leq  \alpha\sum_{v \in \vbl{c}}\log q_v, \label{eq:entropy-lower-bound}\\
\sum_{v \in \vbl{c}}\log\ftp{\frac{q_v}{s_v}} &\geq \beta\sum_{v \in \vbl{c}}\log q_v.\label{eq:entropy-upper-bound}
\end{align}
\end{condition}
Note that for uniform random variable $X_v\in Q_v$, the entropy $H(X_v)=\log q_v$, and for $Y_v=h_v(X_v)$ where $\Proj{h}$ is balanced, we have $\log\frac{q_v}{\ctp{q_v/s_v}}\le H(Y_v)\le \log\frac{q_v}{\ftp{q_v/s_v}}$.
Therefore, the two inequalities \eqref{eq:entropy-lower-bound} and \eqref{eq:entropy-upper-bound} are in fact slightly stronger versions of the entropy upper and lower bounds for $\Ass{X}\sim\pi$:
\[
(1-\alpha)\sum_{v \in \vbl{c}}H(X_v)\le \sum_{v \in \vbl{c}}H(h_v(X_v)) \le (1-\beta)\sum_{v \in \vbl{c}}H(X_v).
\]

So how may such a projection satisfying Condition~\ref{condition-projection}  change the properties of a solution space and help sampling?
Next, we introduce two consequent conditions of Condition~\ref{condition-projection} to explain this.

Recall that after projection, the joint distribution $\nu$ over projected solutions may no longer be represented by any LLL instance.
Nevertheless, we can modify it to a valid LLL instance by proper rounding.

\begin{definition}[the ``round-down'' CSP formula]
\label{def:round-down}
Given a CSP formula $\Phi=(V,\Dom{Q},\Cons{C})$ and a projection scheme $\Proj{h}=(h_v)_{v\in V}$, % where $h_v:Q_v\to\Sigma_v$ for every $v\in V$, 
%=======
%Given a formula $\Phi=(V,\Dom{Q},\Cons{C})$ and a projection scheme $\Proj{h}=(h_v)_{v\in V}$ where $h_v:Q_v\to\Sigma_v$ for every $v\in V$, 
%>>>>>>> 1d11c4743c82399ac2ca8fea9516df1ad30550d3
%
let CSP formula $\RDForm{\Phi}{\Proj{h}}=(V,\Dom{\Sigma},\RDForm{\Cons{C}}{\Proj{h}})$ be constructed as follows:
\begin{itemize}
\item 
the variable set is still $V$ and each variable $v\in V$ now takes values from $\Sigma_v$;
\item 
corresponding to each constraint $c\in\Cons{C}$ of $\Phi$, a constraint $c'\in\RDForm{\Cons{C}}{\Proj{h}}$ %on $\vbl{c'}=\vbl{c}$ 
is constructed as follows:
\begin{align*}
&\vbl{c'}=\vbl{c}
\quad\text{and}\\
%\quad\text{ and }\quad
&\forall \Ass{y}\in\Sigma_{\vbl{c'}},\quad
c'\left(\Ass{y}\right)
=
%\bigwedge_{\Ass{x}\in\Omega_{\vbl{c}}\atop \Proj{h}(\Ass{x})=\Ass{y}}c(\Ass{x}).
%\ftp{\Pr[\Ass{X}\sim\pi_{\vbl{c}}^{\Ass{y}}]{c\left(\Ass{X}\right)}},
%\ftp{\Pr{c\left(\Ass{X}\right)}},
\begin{cases}
\True & \text{if }c(\Ass{x})\text{ for all }\Ass{x}\in\Omega_{\vbl{c}}\text{ that }\Proj{h}(\Ass{x})=\Ass{y},\\
\False & \text{if }\neg c(\Ass{x})\text{ for some }\Ass{x}\in\Omega_{\vbl{c}}\text{ that }\Proj{h}(\Ass{x})=\Ass{y}.
%\text{otherwise}.
\end{cases}
\end{align*}
%where $\Ass{X}\sim\pi_{\vbl{c}}^{\Ass{y}}$.
\end{itemize}
\end{definition}
The CSP formula $\RDForm{\Phi}{\Proj{h}}$ is considered a ``\emph{round-down}'' version of the CSP formula $\Phi$ under projection $\Proj{h}$, because it always holds that
%\begin{align*}
$c'(\Ass{y})=\ftp{\Pr[\Ass{X}\sim\pi]{c\left(\Ass{X}_{\vbl{c}}\right)\mid \Proj{h}\left(\Ass{X}_{\vbl{c}}\right)=\Ass{y}}}$ for all $\Ass{y}\in\Sigma_{\vbl{c'}}=\Sigma_{\vbl{c}}$.
%\end{align*}  

Recall that the following ``LLL condition'' is assumed for the LLL instance defined by  CSP formula $\Phi=(V,\Dom{Q},\Cons{C})$ on uniform random variables $\Ass{X}\sim\pi$:
\begin{align}
\ln\frac{1}{p}>A\ln D+B, \quad \text{(for some suitable constants $A$ and $B$)}\label{eq:abstract-LLL-condition}
\end{align}
where $p\triangleq\max_{c \in\Cons{C}} \Pr[\Ass{X}\sim\pi]{\neg c\left(\Ass{X}_{\vbl{c}}\right)}$ denotes the maximum probability that a constraint $c\in\Cons{C}$ is violated and $D$ denotes the maximum degree of the dependency graph. %and $A, B$ are some suitable constants.

For CSP formula $\Phi$ defined by atomic constraints, 
the LLL condition~\eqref{eq:abstract-LLL-condition} and the inequality~\eqref{eq:entropy-lower-bound} in Condition~\ref{condition-projection} together imply the following condition.

\begin{condition}[round-down LLL criterion]\label{round-down-LLL-condition}
The LLL instance defined by the round-down CSP formula $\RDForm{\Phi}{\Proj{h}}=(V,\Dom{\Sigma},\RDForm{\Cons{C}}{\Proj{h}})$ on variables distributed as $\rho$, satisfies that
\begin{align*}
\ln{\frac{1}{p}} >(1-\alpha)(A \ln D+B),
%Weiming: change alpha to 1 - alpha
%\label{eq:rounded-down-LLL-condition}
\end{align*}
where $p\triangleq \max_{c \in \RDForm{\Cons{C}}{\Proj{h}}}\Pr[\Ass{Y}\sim\rho]{\neg c\left(\Ass{Y}_{\vbl{c}}\right)}$ and $D$ denotes the maximum degree of the dependency graph.
\end{condition}

The projection $\Proj{h}$ may map both satisfying $\Ass{x}\in\Dom{Q}$ and unsatisfying $\Ass{x}'\in\Dom{Q}$ to the same $\Proj{h}(\Ass{x})=\Proj{h}(\Ass{x}')\in\Sigma$, which causes ambiguity for classifying those ``satisfying'' $\Ass{y}\in\Sigma$.
The round-down CSP formula resolves such ambiguity with a pessimistic mindset: it refutes any $\Ass{y}\in\Sigma$ whenever even a single $\Ass{x}\in \Proj{h}^{-1}(\Ass{y})$ is unsatisfying.
Condition~\ref{round-down-LLL-condition} basically says that an LLL condition holds even up to such a pessimistic interpretation.
This is crucial for sampling from $\mu_S^{y_{\Lambda}}$ defined in~\eqref{eq:conditional-marginal-distribution}, because within such regime, the probability space of $\mu^{y_{\Lambda}}$ is decomposed into small clusters of sizes $O(\log n)$.

Meanwhile, the LLL condition~\eqref{eq:abstract-LLL-condition} and the inequality~\eqref{eq:entropy-upper-bound} in Condition~\ref{condition-projection} together imply the following condition.

\begin{condition}[conditional LLL criterion]\label{conditional-LLL-condition}
For any $\Lambda\subseteq V$ and $y_{\Lambda}\in\Sigma_{\Lambda}$, the LLL instance defined by CSP formula $\Phi=(V,\Dom{Q},\Cons{C})$ on variables distributed as $\pi^{y_{\Lambda}}$, satisfies that
\begin{align*}
\ln{\frac{1}{p}} >\beta (A \ln D +B),
%\label{eq:conditional-LLL-condition}
\end{align*}
where $p\triangleq \max_{c \in \Cons{C}}\Pr[\Ass{X}\sim\pi^{y_{\Lambda}}]{\neg c\left(\Ass{X}_{\vbl{c}}\right)}$ and $D$ denotes the maximum degree of the dependency graph.
\end{condition}

Condition~\ref{conditional-LLL-condition} is basically a self-reducibility property.
A major obstacle for sampling satisfying solution is that the regime~\eqref{eq:abstract-LLL-condition} for the original CSP formula $\Phi$ may not be self-reducible: it is not closed under pinning of variables to arbitrary evaluations.
Condition~\ref{conditional-LLL-condition} states that the self-reducibility property is achieved under projection: the LLL regime is closed under pinning of variables to arbitrary projected evaluations.
This is crucial for rapid mixing of the Markov chain on projected space $\Sigma$.

We have efficient procedures for constructing the projection scheme satisfying \Cref{condition-projection}.

\begin{theorem}[projection construction]
\label{theorem-projection-general}
Let $0< \beta< \alpha < 1$ be two parameters.
Let $\Phi=(V,\Dom{Q},\Cons{C})$ be a CSP formula where all constraints in $\Cons{C}$ are atomic.
Let $D$ denotes the maximum degree of its dependency graph and $p \triangleq \max_{c \in \Cons{C}} \prod_{v \in \vbl{c}}\frac{1}{\abs{Q_v}}$.
%For any formula $\Phi=(V,\Dom{Q},\Cons{C})$ where all constraints in $\Cons{C}$ are atomic, 
If $\log \frac{1}{ p} \geq \frac{25}{(\alpha-\beta)^3} \tp{\log D + 3}$,
%where $D$ denotes the maximum degree of its dependency graph and $p \triangleq 
%\max_{c \in\Cons{C}} \Pr[\Ass{X}\sim\pi]{\neg c\left(\Ass{X}_{\vbl{c}}\right)}
%=
%\max_{c \in \Cons{C}} \prod_{v \in \vbl{c}}\frac{1}{\abs{Q_v}}$,
then for any $0<\delta<1$, with probability at least $1-\delta$ 
a projection oracle (Definition~\ref{def:projection-oracle}) with query cost $O(\log q)$ can be successfully constructed within time $O(n(Dk + q) \log \frac{1}{\delta} \log q)$,
where $q \triangleq \max_{v \in V}\abs{Q_v}$, $k \triangleq \max_{c \in \Cons{C}}\abs{ \vbl{c}}$ and the oracle is for a projection scheme $\Proj{h}=(h_v)_{v\in V}$ that satisfies \Cref{condition-projection} with parameters $(\alpha,\beta)$.
%a projection oracle (Definition~\ref{def:projection-oracle}) with query cost $O(\log q)$ for a projection scheme $\Proj{h}=(h_v)_{v\in V}$ that satisfies \Cref{condition-projection} with parameters $(\alpha,\beta)$, can be successfully constructed within time $O(nDk\log\frac{1}{\delta}\log q)$, where $q \triangleq \max_{v \in V}\abs{Q_v}$ and $k \triangleq \max_{c \in \Cons{C}}\abs{ \vbl{c}}$.
%\begin{itemize}
%\item 
%there is a projection oracle (Definition~\ref{def:projection-oracle}) with query cost $O(\log q)$ for  a projection scheme $\Proj{h}=(h_v)_{v\in V}$ that satisfies \Cref{condition-projection} with parameters $(\alpha,\beta)$;
%\item
%for any $0<\delta<1$, with probability at least $1-\delta$ such projection oracle can be successfully constructed within $O(nDk\log\frac{1}{\delta}\log q)$ time, where $q \triangleq \max_{v \in V}\abs{Q_v}$ and $k \triangleq \max_{c \in \Cons{C}}\abs{ \vbl{c}}$.
%\end{itemize}
\end{theorem}

%\begin{theorem}[projection construction]
%\label{theorem-projection-general}
%Let $0< \beta< \alpha < 1$ be two parameters.
%Let $\Phi=(V,\Dom{Q},\Cons{C})$ be a CSP formula where all constraints in $\Cons{C}$ are atomic.
%Let $D$ denotes the maximum degree of its dependency graph and $p \triangleq \max_{c \in \Cons{C}} \prod_{v \in \vbl{c}}\frac{1}{\abs{Q_v}}$.
%Let $q \triangleq \max_{v \in V}\abs{Q_v}$ and $k \triangleq \max_{c \in \Cons{C}}\abs{ \vbl{c}}$.
%If 
%\[
%\log \frac{1}{ p} \geq \frac{25}{(\alpha-\beta)^3} \tp{\log D + 3},
%\]
%then for any $0<\delta<1$, with probability at least $1-\delta$ 
%a projection oracle (Definition~\ref{def:projection-oracle}) with query cost $O(\log q)$ can be successfully constructed within time $O(n(Dk + q) \log \frac{1}{\delta} \log q)$,
%where the oracle is for a projection scheme $\Proj{h}=(h_v)_{v\in V}$ that satisfies \Cref{condition-projection} with parameters $(\alpha,\beta)$.
%\end{theorem}

The above result can be strengthened for the \emph{$(k,d)$-CSP formulas}, where $\abs{\vbl{c}}=k$ for all $c\in\Cons{C}$ and each  $v\in V$ appears in at most $d$ constraints, on homogeneous domains ${Q}_v=[q]$ for all $v\in V$.
\begin{theorem}
\label{theorem-projection-uniform}
Let $0< \beta< \alpha < 1$ be two parameters.
The followings hold for any $(k,d)$-CSP formula $\Phi=(V,[q]^V,\Cons{C})$ where all constraints in $\Cons{C}$ are atomic:
\begin{itemize}
\item 
If $7\le q^{\frac{\alpha+\beta}{2}}\le \frac{q}{6}$ and $\log q \geq \frac{1}{\alpha-\beta}$,
then a projection oracle with query cost $O(\log q)$ for a projection scheme  $\Proj{h}$ satisfying \Cref{condition-projection} with parameters $(\alpha,\beta)$, can be constructed in time $O(n \log q)$.
\item 
If $k \geq \frac{2\ln 2}{(\alpha - \beta)^2}\log(2\mathrm{e}kd)$,
then for any $0<\delta<1$, with probability at least $1-\delta$ a projection oracle as above can be successfully constructed within time $O(ndk\log\frac{1}{\delta})$.
\end{itemize}
\end{theorem}
The proofs of \Cref{theorem-projection-general} and \Cref{theorem-projection-uniform} are given in \Cref{section-projection-construction}.

\pagebreak

\section{The sampling algorithm}
\label{section-algorithm}
%We formally describe the sampling algorithm.
%
Let $\Phi=(V,\Dom{Q},\Cons{C})$ be the input CSP formula with atomic constraints, which defines a uniform distribution $\mu$ over satisfying assignments as in \Cref{def:LLL-distr}.
Let $\epsilon > 0$ be an error bound.
The goal is to output a random assignment $\Ass{X} \in \Dom{Q}$ such that $\DTV{\Ass{X}}{\mu} \leq \epsilon$.
%The input contains an error bound $\epsilon > 0$ and a CSP formula $\Phi=(V,\Dom{Q},\Cons{C})$ with uniformly distributed random variables and atomic constraints.
%The output is a random $\Ass{X} \in \Dom{Q}$ such that $\DTV{\Ass{X}}{\mu} \leq \epsilon$, where $\mu$ is the LLL-distribution in \Cref{def:LLL-distr}.

Depending on the classes of CSP formulas, the algorithm first applies one of the procedures in \Cref{theorem-projection-general} and \Cref{theorem-projection-uniform} to construct a projection scheme $\Proj{h} = (h_v)_{v \in V}$, where $h_v:Q_v\to \Sigma_v$ for each $v\in V$, such that $\Proj{h}$ satisfies~\Cref{condition-projection} with parameters $(\alpha,\beta)$, where $0  < \beta < \alpha < 1$ are going to be fixed later in the analysis in~\Cref{section-proof-main}.
For randomized construction procedure, we set its failure probability to be $\frac{\epsilon}{4}$, and if it fails, the sampling algorithm simply returns an arbitrary $\Ass{X} \in \Dom{Q}$.

Suppose that the projection scheme $\Proj{h}$ is given.
%
%Recall the uniform distribution $\pi$ over $\Dom{Q}$ and the distribution $\nu$ of projected satisfying assignments in \Cref{def:proj-distr}. 
%
The sampling algorithm is described in \Cref{alg-mcmc}.

\begin{algorithm}[ht]
  \SetKwInOut{Input}{input} \SetKwInOut{Output}{output} 
  \Input{a CSP formula $\Phi=(V,\Dom{Q},\Cons{C})$ with atomic constraints, a projection scheme $\Proj{h} = (h_v)_{v \in V}$ satisfying \Cref{condition-projection} with parameters $(\alpha,\beta)$, and an error bound $\epsilon > 0$;}  
  \Output{a random assignment $\Ass{X} \in \Dom{Q}$;}  
  sample a uniform random $\Ass{X} \sim \pi$ and let $\Ass{Y} \gets \Proj{h}(\Ass{X})$\; 
  \For(\hspace{62pt}// \texttt{\small Glauber dynamics for $\Ass{Y}\in\Dom{\Sigma}$ }){each $t$ from $1$ to $T \triangleq \Tmix$} {
    pick a variable $v \in V$ uniformly at random\; 
    %\tcc{resample $X_{\+M}(v)$ from the distribution $\mu_v(\cdot \mid X_{t-1}({\+M} \setminus \{v\} ) )$}
    $X_v \gets \sample\tp{\Phi, \Proj{h}, \frac{\epsilon}{4(T+1)}, Y_{V \setminus \{v\}}, \{v\}}$\label{line-sample-1};\hspace{30pt}// \texttt{\small sample $X_v \in Q_v$ approx.~from  $\mu^{Y_{V \setminus \{v\}}}_v$}\\
    $Y_v \gets h_v(X_v)$;\label{line-map}
    %\hspace{156pt}// \texttt{\small $Y_v \in \Sigma_v$  approx.~follows $\nu^{Y_{V \setminus \{v\}}}_v$}
    %\quad\tcp{update $Y_v \in \Sigma_v$ via $X_v$ and $h_v$ }
    %$\forall u\in\+M$ and $u\neq v$, $X_{t}(u)\gets X_{t-1}(u)$\;
    }
  %\tcc{sample $X_{V\setminus\+M}$ from the distribution $\mu_{V \setminus {\+M}}(\cdot\mid X_{T})$} 
  $\Ass{X} \gets \sample\tp{\Phi, \Proj{h}, \frac{\epsilon}{4(T+1)}, \Ass{Y} , V}$;\label{line-sample-2}{\hspace{78pt}// \texttt{\small sample $\Ass{X} \in \Dom{Q}$ approx.~from  $\mu^{\Ass{Y}}$}}\\
  \Return{$\Ass{X}$\;}
  \caption{The sampling algorithm (given a proper projection scheme)}\label{alg-mcmc}
\end{algorithm}

\Cref{alg-mcmc} implements the sampling algorithm outlined in \Cref{section-proj-of-LLL}.
It first implements the Glauber dynamics on space $\Dom{\Sigma}$ for sampling from $\nu$, the distribution of projected satisfying assignments in \Cref{def:proj-distr}.
It simulates the Glauber dynamics for $T = \Tmix$ steps to draw a random $\Ass{Y}\in\Dom{\Sigma}$ distributed approximately as $\nu$.
At each step, $Y_v$ for a uniformly picked $v\in V$ is redrawn approximately from the marginal distribution $\nu^{Y_{V \setminus \{v\}}}_v$.
At last, the algorithm inverts the sampled $\Ass{Y}\in\Dom{\Sigma}$ to a random satisfying assignment $\Ass{X}\in\Dom{Q}$ distributed approximately as $\mu$ conditioning on that $\Proj{h}(\Ass{X})=\Ass{Y}$.

\Cref{alg-mcmc} relies on an \emph{Inverse Sampling} subroutine for sampling approximately from $\mu^{Y_{V \setminus \{v\}}}_v$ or~$\mu^{\Ass{Y}}$.
%for sampling $X_S \in Q_S$ approximately according to the  distribution $\mu_S^{y_{\Lambda}}$ defined in~\eqref{eq:conditional-marginal-distribution}.

\subsection{The \sample{} subroutine (\Cref{alg-sample})}
The goal of the subroutine $\sample\tp{\Phi, \Proj{h}, \delta, y_{\Lambda}, S}$, where $S\subseteq V$, $\Lambda\subseteq V$, and $y_{\Lambda}\in\Sigma_{\Lambda}$, is to sample a random $X_S \in Q_S$ according to the distribution $\mu_S^{y_{\Lambda}}$, as defined in~\eqref{eq:conditional-marginal-distribution}.
In principle, computing the distribution $\mu_S^{y_{\Lambda}}$ involves computing some nontrivial partition function, which is intractable in general.
Here, for an error bound $\delta>0$, we only ask for that with probability at least $1-\delta$, the subroutine returns a random sample that is $\delta$-close  to $\mu_S^{y_{\Lambda}}$ in total variation distance, where the probability is taken over the randomness of the input $y_{\Lambda}$.
 
%\Cref{alg-mcmc} first uses Glauber dynamics for $T$ steps to sample a random image $X$ from the projected distribution $\mu_{\+P}$.
% In each step of Glauber dynamics, after picking the variable $v \in V$, we first use the subroutine $\sample(\cdot)$ to sample a random $c_v \in Q_v$ from the distribution $\mu_v^{X_{V \setminus \{v\}}}$, and then use the projection scheme $\+P$ to project $c_v \in Q_v$ to $X_v \in [m_v]$ such that $c_v \in P_v(X_v)$. It is easy to verify $X_v \sim \mu_{v,\+P}^{X_{V \setminus \{v\} }}$ if $c_v \sim \mu_v^{X_{V \setminus \{v\}}}$. After the Glauber dynamics, given  $X \in \Omega_{\+P}$, we sample the random assignment $Y \in \otimes_{v\in V}Q_v$ from $\mu^X$.

%We now given the subroutine $\sample\tp{\Phi, \Proj{h}, \delta, y_{\Lambda}, S}$.
We define some notions to describe the subroutine.
Let $c \in \Cons{C}$ be a constraint in CSP formula $\Phi$. 
Recall that $c$ is atomic.
Let 
\begin{align*}
\Ass{F}^c \triangleq c^{-1}(\False) 
\end{align*}
denote the unique ``forbidden configuration'' in $\Dom{Q}_{\vbl{c}}$ that violates $c$. 
We say that an atomic constraint $c \in \Cons{C}$ is \emph{satisfied} by $y_{\Lambda} \in \Sigma_{\Lambda}$ for $\Lambda\subseteq V$, if 
\begin{align}
\label{eq-def-sat-by-y}
\Proj{h}\tp{F^c_{\Lambda\cap\vbl{c}}} \neq y_{\Lambda \cap \vbl{c}},
\end{align}
where the function $\Proj{h}(\cdot)$ is formally defined in~\eqref{eq-def-proj-set}.
For atomic constraint $c\in\Cons{C}$, the above condition~\eqref{eq-def-sat-by-y} implies that $c$ is satisfied by any $\Ass{x} \in \Dom{Q}$ that $\Proj{h}({x_{ \Lambda }})= y_{\Lambda}$. 
Hence, the constraint $c$ must be satisfied by any configuration in the support of the distribution $\mu^{y_{\Lambda}} = \mu^{y_{\Lambda}}_{V}$. %where $\mu^{y_{\Lambda}}_{V}$ is defined in~\eqref{eq:conditional-marginal-distribution}.

The key idea of the subroutine is that we can remove all the constraints that have already been satisfied by $y_{\Lambda}$ to obtain a new CSP formula $\Phi' = (V, \Dom{Q},\Cons{C}')$, where $\Cons{C}' \triangleq \{c \in \Cons{C} \mid c \text{ is not satisfied by } y_\Lambda\}$.

Define $\mu_{\Phi'}^{y_{\Lambda}}$ to be the distribution of $\Ass{X} \sim \pi^{y_{\Lambda}}$ conditioned on $\Phi'(\Ass{X})$, where the product distribution $\pi^{y_{\Lambda}}$ is as in \Cref{def:proj-distr}.
It is straightforward to verify that $\mu_{\Phi'}^{y_{\Lambda}} \equiv \mu^{y_{\Lambda}}$.

Furthermore, the new CSP formula $\Phi'$ can be \emph{factorized} into a set of disjoint formulas:
\begin{align*}
\Phi' = \Phi'_1 \wedge \Phi'_2 \wedge \ldots \wedge \Phi'_m.
\end{align*}
Our plan is to show that it almost always holds that the size of every sub-formula $\Phi'_i$ is logarithmically bounded.
Thus, we can apply the na\"ive rejection sampling independently on each sub-formula $\Phi'_i$, which remains to be efficient altogether.

\begin{algorithm}[t]
  \SetKwInOut{Input}{Input}%
  \SetKwInOut{Output}{Output}%
  \Input{
  a CSP formula $\Phi=(V,\Dom{Q},\Cons{C})$ with atomic constraints, a projection scheme $\Proj{h}$, an error bound $\delta >0$, a configuration $y_{\Lambda} \in \Sigma_{\Lambda}$ specified on $\Lambda \subseteq V$, and a subset $S \subseteq V$; %and a parameter $0 < \eta < 1$.
   }%
  \Output{a random assignment $\Ass{X} \in \Dom{Q}_S$;}%
  % delete all the clauses
  % $c \in \Cons{C}$ from $\Phi$ such that $c$ is satisfied by $X$, delete all the variables in $\Lambda$
  % from the remaining clauses, and let $\Phi'=(V',C')$ be the new CNF formula after deletion\;
  % delete all the variables in $\Lambda$ from $\Phi$\;
  let $\Phi'$ be the new formula obtained by removing all the constraints in $\Phi$ already satisfied by $y_{\Lambda}$\;%
  factorize $\Phi'$ and find all the sub-formulas $\set{\Phi'_i = (V_i,\Dom{Q}_{V_i}, \Cons{C}'_i)\mid 1\leq i \leq \ell}$ s.t. each $V_i \cap S \neq \emptyset$\;%
  \If(\hspace{13pt}// \texttt{\small existence of giant component}){there exists $ 1 \leq i \leq \ell$ s.t. $|\Cons{C}'_i| > 2D\log \frac{nD}{\delta}$}
  	{\Return{a uniform random  $\Ass{X}_S\sim \pi_S$\;\label{line-bad-return-1}}}
  \For{each $i$ from 1 to $\ell$\label{line-rejection-begin} }{ 
    \Repeat( \emph{for at most $R \triangleq \ctp{ 10\tp{\frac{n}{\delta}}^{\eta} \log \frac{n}{\delta}}$ times:\label{line-v-start}}
    \hspace{18pt}// \texttt{\small rejection sampling with $\le R$ trials}){$\Phi'_i(\Ass{X}_i)=\True$\label{line-v-end}}{
    sample $\Ass{X}_i \sim \pi_{V_i}^{y_{\Lambda_i}}$, where $\Lambda_i\triangleq V_i\cap\Lambda$\;
    }
%  let $\Ass{X}_i \gets \Rejection\left(\Phi'_i, (h_v)_{v \in V_i}, y_{V_i \cap \Lambda}, R\right)$ with $R = \ctp{ 10\tp{\frac{n}{\delta}}^{\eta} \log \frac{n}{\delta}}$\;
    \If(\hspace{108pt}// \texttt{\small overflow of rejection sampling}){$\Phi'_i(\Ass{X}_i)=\False$}{
      \Return{a uniform random  $\Ass{X}_S\sim \pi_S$\;\label{line-bad-return-2}}
    } 
%  {\textbf{if} $\Ass{X}_i = \perp$ \textbf{then  return} a uniform random assignment $\Ass{X}\sim \pi_S$ \label{line-bad-return-2}\;}
  }
  \Return{$\Ass{X}'_S$, where $\Ass{X}' = \bigcup_{i=1}^\ell \Ass{X}_i$\;\label{line-good}}
  \caption{$\sample\tp{\Phi, \Proj{h}, \delta, y_{\Lambda}, S}$}\label{alg-sample}
\end{algorithm}

%\begin{algorithm}[t]
%  \SetKwInOut{Input}{Input} \SetKwInOut{Output}{Output} \Input{ a CSP formula $\Phi'_i = (V_i,\Dom{Q}_{V_i}, \Cons{C}'_i)$ with atomic constraints, a projection scheme $(h_v)_{v \in V_i}$, a configuration $y_{V_i \cap \Lambda} \in \Sigma_{V_i \cap \Lambda}$ for some subset $V_i \cap \Lambda $,  and an integer  $R > 0$.}  \Output{a random assignment $\Ass{X} \in \Dom{Q}_{V_i}$ or a special symbol
%    $\perp$.}  
%    %let $\widetilde{Q}_v \gets P_v(X_\Lambda(v))$ for all $v \in V \cap \Lambda$, and $\widetilde{Q}_u \gets Q_u$ for all $u \in V \setminus \Lambda$\label{line-construct-Q}\;
%    \For{each $i$ from 1 to $R$}{ 
%    sample a random assignment $\Ass{X} \sim \pi_{V_i}^{y_{V_i \cap \Lambda}}$, where $\pi_{V_i}^{y_{V_i \cap \Lambda}} = \bigotimes_{v \in V_i \cap \Lambda}\pi_v^{y_v} \times \bigotimes_{v \in V_i \setminus \Lambda}\pi_v$\;
%    %for each $v \in V$, sample $Y_v \in \widetilde{Q}_v$ uniformly and independently\;
%    \If{$\Phi'_i(\Ass{X})=\True$}{ \Return{$\Ass{X}$\;} } } \Return{$\perp$\;}
%  \caption{$\Rejection\left(\Phi'_i, (h_v)_{v \in V_i}, y_{V_i \cap \Lambda}, R\right)$}\label{alg-rejection}
%\end{algorithm}

Formally, let $H' = (V,\+E')$ denote the (multi-)hypergraph induced by the CSP formula $\Phi' = (V, \Dom{Q}, \Cons{C}')$, constructed by adding a hyperedge  $e_c = \vbl{c}$ into $\+E'$ for each constraint $c \in \Cons{C}'$.
Note that $H'$ may contain duplicated hyperedges.
Let $H_1',H_2',\ldots,H_m'$ denote the connected components of $H'$, where $H_i' = (V_i, \+E_i')$. 
Let $\Phi_i'=(V_i,\Dom{Q}_{V_i}, \Cons{C}'_i)$ denote sub-formula corresponding to  $H_i'$, where $\Cons{C}_i'$ is the set of constraints corresponding to hyperedges in $\+E'_i$.
This defines the \emph{factorization} $\Phi' = \Phi'_1 \wedge \Phi'_2 \wedge \ldots \wedge \Phi'_m$.
For each sub-formula $\Phi'_i = (V_i,\Dom{Q}_{V_i}, \Cons{C}'_i)$, 
let $\Lambda_i = \Lambda \cap V_i$, and define $\mu_{\Phi'_i}^{y_{\Lambda_i}}$ to be the distribution of $\Ass{X} \sim \pi^{y_{\Lambda_i}}_{V_i}$ conditioned on $\Phi_i'(\Ass{X})$, where $\pi^{y_{\Lambda_i}}_{V_i}$ denotes restriction of the product distribution $\pi^{y_{\Lambda_i}}$ on $V_i$.
It is then straightforward to verify:
\begin{align*}
\mu^{y_{\Lambda}} \equiv \mu_{\Phi'}^{y_{\Lambda}} \equiv \mu_{\Phi'_1}^{y_{\Lambda_1}}\times \mu_{\Phi'_2}^{y_{\Lambda_2}} \times \ldots \times \mu_{\Phi'_m}^{y_{\Lambda_m}}.
\end{align*}

%The goal of the subroutine is to draw a random sample $\Ass{X} \in Q_S$ from the distribution $\mu^{y_{\Lambda}}_S$.
Without loss of generality, we assume $S \cap V_i \neq \emptyset$ for  $1\leq i \leq \ell$ and $S \cap V_i = \emptyset$ for  $\ell< i \leq m$.
It suffices to draw random samples $\Ass{X}_i \sim \mu_{\Phi'_i}^{y_{\Lambda_i}}$ independently for all $1\leq i \leq \ell$, adjoin them together $\Ass{X}' = \cup_{i=1}^\ell \Ass{X}_i$, and output its restriction $\Ass{X}'_S$ on $S$, where each $\Ass{X}_i \sim \mu_{\Phi'_i}^{y_{\Lambda_i}}$ can be drawn by the \emph{rejection sampling} procedure: repeatedly and independently  sampling $\Ass{X}_i \sim \pi^{y_{\Lambda_i}}_{V_i}$ until $\Phi'_i(\Ass{X}_i)$ is true. 
%
%Let $0 < \eta < 1$ be a parameter to be determined in \Cref{section-proof-main}. 
%Let $\pi_S = \bigotimes_{v \in S}\pi_v$.

The subroutine $\sample\tp{\Phi, \Proj{h}, \delta, y_{\Lambda}, S}$  does precisely as above with two exceptions:
%Recall that $\+P$  satisfies \Cref{condition-projection} with parameters $\alpha$ and $\beta$.
%Let $D$ denote the maximum degree of $\Phi$'s dependency graph.
%Let $\eta \geq 1$ be a parameter such that
%\begin{align}
%\label{eq-def-eta}
%\forall c \in \Cons{C}, \quad \sum_{v \in \vbl{c}}\log \abs{Q_v} \geq \frac{1}{\beta} \log (20 \mathrm{e} \eta D^2).
%\end{align}
\begin{itemize}
\item
\textbf{existence of giant connected component}: 
$|\Cons{C}_i'| \geq 2D\log \frac{n D}{\delta}$ for some $1\leq i \leq \ell$, where $D$ stands for the maximum degree of the dependency graph for $\Phi$;
\item 
\textbf{overflow of rejection sampling}:
the rejection sampling from $\mu_{\Phi'_i}^{y_{\Lambda_i}}$ for some $1\leq i \leq \ell$, has used more than $R = \ctp{ 10\tp{\frac{n}{\delta}}^{\eta} \log \frac{n}{\delta}}$ trials, where $\eta$ is a parameter to be fixed in \Cref{section-proof-main}. %($\log \frac{1}{p} \geq \frac{1}{\beta}\log\tp{\frac{40\mathrm{e} D^2}{\eta}}$).
\end{itemize}
If either of the above exceptions occurs, the algorithm terminates and returns a random $\Ass{X}_S \sim \pi_S$.

In \Cref{section-subroutine}, we will show that assuming \Cref{condition-projection} for the projection scheme $\Proj{h}$ with properly chosen parameters $(\alpha,\beta)$ and by properly choosing $\eta$, for the random $y_\Lambda$ upon which the subroutine is called in  \Cref{alg-mcmc}, 
with high probability none of these exceptions occurs.
Therefore, the random sample returned by the subroutine is accurate enough when being called in \Cref{alg-mcmc}.

%Similarly, we use $H_{\Phi}^{X_\Lambda} = (V,\+E^{X_\Lambda})$ to model  $\Phi^{X_\Lambda} = (V,C^{X_\Lambda},[q])$, where
%\begin{align*}
%\+E^{X_\Lambda} \triangleq \{e = \vbl{c} \mid c \in \Cons{C} \land c \text{ is not satisfied by } X_{\Lambda} \}.	
%\end{align*}

%\todo{ $\alpha$ is not in the parameters}

\pagebreak

\section{Proofs of the main results}
\label{section-proof-main}
In this section, we prove the main theorems of this paper.
Our algorithm first constructs a projection scheme using one of the procedures in \Cref{theorem-projection-general} and \Cref{theorem-projection-uniform}, which gives us the projection oracle that can answer queries within time cost $O(\log q)$, where $q = \max_{v \in V}\abs{Q_v}$.
We then execute \Cref{alg-mcmc} for sampling $\Ass{X}$ approximately according to $\mu$.
We assume the following basic operations for uniform sampling:
\begin{itemize}
\item 
draw a variable $v \in V$ uniformly at random within time cost $O(\log n)$;
\item for any variable $v \in V$, draw a uniform sample $X\sim\pi_v$ from $Q_v$ within time cost $O(\log q)$.
\end{itemize}
When measuring the time cost of \Cref{alg-mcmc}, we count the number of calls to the projection oracle as well as the above two basic sampling operations.
The time complexity of \Cref{alg-mcmc} is dominated by these oracle costs.

Next, we prove \Cref{theorem-main} for general CSP formulas with atomic constraints, while \Cref{theorem-coloring}  and  \Cref{theorem-CNF} for specific subclasses of formulas are proved in \Cref{section-proof-applications}.

\subsection{CSP formulas with atomic constraints}
\label{section-proof-general}
For CSP formulas $\Phi = (V, \Dom{Q}, \Cons{C})$ defined by atomic constraints, we show that sampling uniform solution is efficient within the following regime:
\begin{align}
\label{eq:main-gen}
\ln \frac{1}{p} \geq 350\ln D + 3 \ln \frac{1}{\zeta}
\end{align}
where $p = \max_{c \in \Cons{C}}\prod_{v \in \vbl{c}} \frac{1}{\abs{Q_v}}$ stands for the maximum probability that a constraint $c\in\Cons{C}$ is violated by uniform random assignment,
and
$D$ stands for the maximum degree of the dependency graph of $\Phi$.
The positive constant parameter $\zeta$ specifies a gap to the boundary of the regime.
%To prove \Cref{theorem-main}, we give the following general theorem.
\begin{theorem}
\label{theorem-main-gen}
The following holds for any $0 < \zeta \leq 2^{-400}$.
There is an algorithm such that given any $0<\epsilon<1$ and CSP formula $\Phi = (V, \Dom{Q}, \Cons{C})$ with atomic constraints satisfying~\eqref{eq:main-gen},
%$\ln \frac{1}{p} \geq 350\ln D + 3 \ln \frac{1}{\zeta}$, where $p = \max_{c \in \Cons{C}}\prod_{v \in \vbl{c}} \frac{1}{\abs{Q_v}}$ denotes the maximum probability that a constraint is violated by uniform random assignment and $D$ denotes the maximum degree of the dependency graph of $\Phi$, 
the algorithm 
outputs 
a random assignment $\Ass{X} \in \Dom{Q}$
whose distribution is $\epsilon$-close in total variation distance to the uniform distribution $\mu$ over all solutions to $\Phi$,
using time cost
$O\tp{(D^2 k+q) n \tp{\frac{n}{\epsilon}}^{\zeta} \log ^4 \tp{\frac{nDq}{\epsilon} } }$,
where  $k = \max_{c \in \Cons{C}} \abs{\vbl{c}}$.
\end{theorem}

\Cref{theorem-main} is implied by \Cref{theorem-main-gen}, by interpreting any LLL instance with uniform random variables and atomic bad events as a CSP formula with atomic constraints.

Let $\Proj{h}= (h_v)_{v \in V}$ be a projection scheme satisfying \Cref{condition-projection} with parameters $\alpha$ and $\beta$.
To prove \Cref{theorem-main-gen}, we have the following lemma which shows that assuming a Lov\'asz local lemma condition, the Glauber dynamics for the projected distribution $\nu$ is rapidly mixing.
\begin{lemma}
\label{lemma-mixing-gen}
If $\log \frac{1}{p} \geq \frac{50}{\beta} \log \tp{\frac{2000D^4}{\beta}}$, then the Markov chain $P_{\mathrm{Glauber}}$ on $\nu$ has
%\begin{align*}
$\tmix(\epsilon) \leq \ctp{2n \log \frac{n}{\epsilon}}$.
%\end{align*}
\end{lemma}
\noindent The proof of \Cref{lemma-mixing-gen} is given in \Cref{section-mixing}.

We also need the following lemma for analyzing the subroutine $\sample(\Phi,\Proj{h}, \delta, X_\Lambda, S)$.
In \Cref{alg-mcmc} the  subroutine is called for $T+1$ times.
For $1 \leq t \leq T+1$, define the following bad events:
\begin{itemize}
\item $\+B^{(1)}_t$: in the $t$-th call of $\sample(\cdot)$,	a random assignment $\Ass{X}$ is returned in \Cref{line-bad-return-2}.
\item $\+B^{(2)}_t$: in the $t$-th call of $\sample(\cdot)$, a random assignment $\Ass{X}$ is returned in \Cref{line-bad-return-1}
\end{itemize}

\begin{lemma}
\label{lemma-simulation}
%If $\+P$ satisfies \Cref{condition-projection} with parameter $0 < \alpha \leq 1$ , 
%Let $\Phi = (V,C,(Q_v)_{v \in V})$ be the input CSP formula and $\+P$ a projection scheme satisfying \Cref{condition-projection} with parameters $\alpha$ and $\beta$.
%Let $D$ denote the maximum degree of the dependency graph of $\Phi$.
%Let $p = \max_{c \in \Cons{C}}\prod_{v \in \vbl{c}}\frac{1}{\abs{Q_v}}$.
%Let $\eta \geq 1$ be a parameter.
Let $1\leq t \leq T+1$ and $0 <\eta < 1$.
In \Cref{alg-mcmc}, for the $t$-th calling to the subroutine $\sample(\Phi, \Proj{h} ,\delta, y_\Lambda, S)$ with parameter $\eta$, it holds that
\begin{itemize}
\item given access to a projection oracle with query cost $O(\log q)$, the time cost of  $\sample(\Phi, \Proj{h} ,\delta, y_\Lambda, S)$ is bounded as %$\sample(\Phi, \Proj{h} ,\delta, y_\Lambda, S)$ is 
\begin{align*}
O\tp{\abs{S}D^2k \tp{\frac{n}{\delta}}^{\eta} \log^2\tp{\frac{nD}{\delta}} \log q },
\end{align*} 
where $k = \max_{c \in \Cons{C}}\abs{\vbl{c}}$ and $q = \max_{v \in V}\abs{Q_v}$;	
\item conditioned on $\neg\+B^{(1)}_t\land\neg\+B^{(2)}_t$, the $t$-th calling to  $\sample(\Phi, \Proj{h} ,\delta, y_\Lambda, S)$ returns a $\Ass{X}_S\in Q_S$ that is distributed precisely according to $\mu_S^{y_{\Lambda}}$.
\end{itemize}
Furthermore, 
if $\log \frac{1}{p} \geq \frac{1}{1-\alpha}\log(20D^2)$ and $\log \frac{1}{p} \geq \frac{1}{\beta}\log\tp{\frac{40\mathrm{e} D^2}{\eta}}$ %for some parameter $0 < \eta < 1$, then the subroutine with parameter $\eta$ satisfies
it holds that
\begin{align*}
\Pr[]{\+B^{(1)}_t} \leq \delta \quad \text{and} \quad 	\Pr[]{\+B^{(2)}_t} \leq \delta.
\end{align*}
\end{lemma}
\noindent
The proof of \Cref{lemma-simulation} is given in \Cref{section-subroutine}.

\begin{proof}[Proof of \Cref{theorem-main-gen}]
Let $\alpha,\beta,\eta$ be three parameters to be fixed later.
Our algorithm first uses the algorithm in \Cref{theorem-projection-general} with $\delta = \frac{\epsilon}{4}$ to construct a projection scheme satisfying \Cref{condition-projection} with parameters $\alpha$ and $\beta$. If the algorithm in \Cref{theorem-projection-general} fails to find such projection scheme, our algorithm terminates and outputs an arbitrary $\Ass{X}_{\mathrm{out}} \in \Dom{Q}$. If the algorithm finds such projection scheme, we run \Cref{alg-mcmc} to obtain the random sample $\Ass{X}_{\mathrm{out}} = \Ass{X}_{\mathrm{alg}}$, where $\Ass{X}_{\mathrm{alg}}$ denotes the output of \Cref{alg-mcmc}.

We first analyze the running time of the whole algorithm. 
By \Cref{theorem-projection-general}, the running time for constructing the projection scheme is 
\begin{align*}
T_{\mathrm{proj}} = O\tp{n(Dk+q)\log\frac{1}{\epsilon} \log q}.	
\end{align*} 
If the algorithm in \Cref{theorem-projection-general} succeeds, then it gives a projection oracle with query cost $O(\log q)$.
In \Cref{alg-mcmc}, we simulate the Glauber dynamics for $T = \Tmix$ transition steps. 
In each step, the algorithm  first picks a variable $v \in V$ uniformly at random, the cost is $O(\log n)$.
The algorithm then calls the subroutine  $\sample\tp{\Phi, \Proj{h}, \frac{\epsilon}{4(T+1)}, Y_{V \setminus \{v\}}, \{v\}}$ to draw a random $X_v \in Q_v$.
By \Cref{lemma-simulation}, the cost of the subroutine is $O\tp{D^2k\tp{\frac{n}{\delta}}^{\eta}\log^2\tp{\frac{nD}{\delta}}\log q}$, where %$q = \max_{v \in V}\abs{Q_v}$,
\begin{align*}
\delta = \frac{\epsilon}{4(T+1)} = \Theta\tp{\frac{\epsilon}{n \log \frac{n}{\epsilon}}}	 = \Omega\tp{\frac{\epsilon^2}{n^2}}.
\end{align*}
After $X_v$ is sampled in \Cref{line-sample-1}, the algorithm calls the projection oracle to map $X_v \in Q_v$ to $Y_v = h_v(v) \in \Sigma_v$, the cost of this step is $O(\log q)$.
Thus, the cost for simulating each transition step is 
\begin{align}
\label{eq-T-step}
T_{\mathrm{step}} = O\tp{D^2 k \tp{\frac{n}{\epsilon}}^{3 \eta} \log^2 \tp{\frac{nD}{\epsilon}}\log q }. 	
\end{align}
Finally, the algorithm uses $ \sample\tp{\Phi, \Proj{h}, \frac{\epsilon}{4(T+1)}, \Ass{Y}, V}$ in \Cref{line-sample-2} to sample the final output. By \Cref{lemma-simulation}, the cost is $O\tp{nD^2k\tp{\frac{n}{\delta}}^{\eta}\log^2\tp{\frac{nD}{\delta}}\log q}$, where $\delta = \frac{\epsilon}{4(T+1)} = \Omega\tp{\frac{\epsilon^2}{n^2}}$. 
Hence, the cost for the last step is 
\begin{align}
\label{eq-T-final}
T_{\mathrm{final}} = O\tp{nD^2 k \tp{\frac{n}{\epsilon}}^{3\eta} \log^2 \tp{\frac{nD}{\epsilon}}\log q }.	
\end{align}
Combining all of them together, the total running time is 
\begin{align}
\label{eq-tot-time}
T_{\mathrm{total}} &= T_{\mathrm{proj}} + T \cdot T_{\mathrm{step}} + T_{\mathrm{final}}=O\tp{n(Dk+q)\log\frac{1}{\epsilon}\log q} + O\tp{(T+n) D^2 k \tp{\frac{n}{\epsilon}}^{3\eta } \log^2 \tp{\frac{nD}{\epsilon}} \log q }\notag\\
&= O\tp{(D^2 k+q) n \tp{\frac{n}{\epsilon}}^{3\eta} \log^3 \tp{\frac{nD}{\epsilon}} \log q}.
\end{align}

Next, we prove the correctness of the algorithm, i.e., the total variation distance between the output $\Ass{X}_{\mathrm{out}}$ and the uniform distribution $\mu$ is at most $\epsilon$. 
It suffices to prove
\begin{align}
\label{eq-alg-mu}
\DTV{\Ass{X}_{\mathrm{alg}}}{\mu} \leq \frac{3\epsilon}{4}.	
\end{align}
Because if $0<\beta<\alpha<1$ and $\log \frac{1}{ p} \geq \frac{25}{(\alpha-\beta)^3} \tp{\log D + 3}$, then with probability at least $1-\frac{\epsilon}{4}$, the algorithm in \Cref{theorem-projection-general} constructs the projection scheme successfully, i.e. $\Ass{X}_{\mathrm{out}} = \Ass{X}_{\mathrm{alg}}$. 
Let $\Ass{X} \sim \mu$. By coupling lemma, we can couple $\Ass{X}$ and $\Ass{X}_{\mathrm{alg}}$ such that  $\Ass{X}\neq \Ass{X}_{\mathrm{alg}}$ with probability $\frac{3\epsilon}{4}$. Thus, we can coupling $\Ass{X}$ and  $\Ass{X}_{\mathrm{out}}$ such that $\Ass{X}\neq \Ass{X}_{\mathrm{out}}$ with probability at most $\frac{\epsilon}{4} + \frac{3\epsilon}{4} = \epsilon$. By coupling lemma,
\begin{align*}
\DTV{\Ass{X}_{\mathrm{out}}}{\mu} \leq \epsilon.	
\end{align*}
We then verify~\eqref{eq-alg-mu}. 
Consider an idealized algorithm that first runs the idealized Glauber dynamics for $T = \Tmix$ steps to obtain a random sample $\Ass{Y}_{\mathrm{G}}$, then samples $\Ass{X}_{\mathrm{idea}}$ from the distribution $\mu^{\Ass{Y}_{\mathrm{G}}}$.  
By \Cref{lemma-mixing-gen}, if $\log \frac{1}{p} \geq \frac{50}{\beta} \log \tp{\frac{2000D^4}{\beta}}$, then $\DTV{\Ass{Y}_{\mathrm{G}}}{\nu} \leq \frac{\epsilon}{4}.$
Consider the following process to draw a random sample $\Ass{X} \sim \mu$.
First sample $\Ass{Y} \sim \nu$, then sample $\Ass{X} \sim \mu^{\Ass{Y}}$.
Thus, we can couple $\Ass{Y}$ and $\Ass{Y}_G$ such that $\Ass{Y} \neq \Ass{Y}_\mathrm{G}$ with probability $\frac{\epsilon}{4}$.
Conditional on $\Ass{Y} = \Ass{Y}_{\mathrm{G}}$, $\Ass{X}$ and $\Ass{X}_{\mathrm{idea}}$ can be perfectly coupled.
By coupling lemma,
\begin{align}
\label{eq-Glau-alg}
\DTV{\Ass{X}_{\mathrm{idea}}}{\mu} \leq \frac{\epsilon}{4}.		
\end{align}
We now couple \Cref{alg-mcmc} with this idealized algorithm. 
For each transition step, they pick the same variable, then couple each transition step optimally. 
In the last step, they use the optimal coupling to draw random samples from the conditional distributions. 
Note that in \Cref{line-sample-1} of \Cref{alg-mcmc}, if the random sample $X_v \in Q_v$ returned by the subroutine is a perfect sample from $\mu_v^{Y_{V \setminus \{v\} }}$, then the $Y_v \in \Sigma_v$ constructed 
in \Cref{line-map} follows the distribution $\nu_{v}^{Y_{V \setminus \{v\} }}$.
By \Cref{lemma-coupling-gen}, if none of $\+B^{(1)}_t$ and $\+B^{(2)}_t$ for $1\leq t \leq T + 1$ occurs, then
all the $(T+1)$ executions of the subroutine $\sample(\Phi,\Proj{h} ,\delta, y_\Lambda, S)$ return perfect samples from $\mu_S^{y_{\Lambda }}$. In this case, \Cref{alg-mcmc} and the idealized algorithm can be coupled perfectly.
Note that $\delta = \frac{\epsilon}{4(T+1)}$.
By coupling lemma and \Cref{lemma-simulation}, we have
\begin{align*}
\DTV{\Ass{X}_{\mathrm{alg}}}{\Ass{X}_{\mathrm{idea}}}\leq \Pr[]{\bigvee_{i=1}^{T+1}\tp{\+B^{(1)}_t \vee \+B^{(2)}_t }} \leq 2(T+1)\delta = \frac{\epsilon}{2}.
\end{align*}
Hence, ~\eqref{eq-alg-mu} can be proved by the following triangle inequality
\begin{align*}
\DTV{\Ass{X}_{\mathrm{alg}}}{\mu}  \leq 	\DTV{\Ass{X}_{\mathrm{alg}}}{\Ass{X}_{\mathrm{idea}}} + \DTV{\Ass{X}_{\mathrm{idea}}}{\mu} \leq \frac{\epsilon}{2} + \frac{\epsilon}{4} \leq \frac{3\epsilon}{4}.
\end{align*}

We then set the parameters $\alpha,\beta$ and $\eta$. We put all the constraints in \Cref{theorem-projection-general}, \Cref{lemma-mixing-gen} and \Cref{lemma-simulation} together:
\begin{align*}
0<\beta &< \alpha < 1, \quad 0 < \eta < 1;\\
\log \frac{1}{ p} &\geq \frac{25}{(\alpha-\beta)^3} \tp{\log D + 3};\\	
\log \frac{1}{p} &\geq \frac{50}{\beta} \log \tp{\frac{2000D^4}{\beta}};\\
 \log \frac{1}{p} &\geq \frac{1}{1 - \alpha}\log(20D^2); \\
 \log \frac{1}{p} &\geq \frac{1}{\beta}\log\tp{\frac{40\mathrm{e} D^2}{\eta}}.
\end{align*}
We can take $\alpha = 0.994$ and $\beta = 0.577$. The following condition implies all the above constraints
\begin{align*}
\log \frac{1}{p} \geq 350\log D + 3\log \frac{1}{\zeta} \quad \text{and} \quad \eta = \frac{\zeta}{3}, \quad\text{where } 0 < \zeta \leq 2^{-400}.
\end{align*}
Remark that $\log \frac{1}{p} \geq 350\log D + 3\log \frac{1}{\zeta}$ is equivalent to $\ln \frac{1}{p} \geq 350\ln D + 3\ln \frac{1}{\zeta}$.
By ~\eqref{eq-tot-time}, under this condition, the total running time is 
\begin{align*}
T_{\mathrm{total}} =O\tp{(D^2 k+q) n \tp{\frac{n}{\epsilon}}^{3\eta} \log^3 \tp{\frac{nD}{\epsilon}} \log q}= O\tp{(D^2 k+q) n \tp{\frac{n}{\epsilon}}^{\zeta} \log ^4 \tp{\frac{nDq}{\epsilon} } }. &\qedhere
\end{align*}
\end{proof}

\subsection{Sharper bounds for subclasses of CSP formulas}
\label{section-proof-applications}
We prove the following theorems on specific subclasses of CSP formulas. 
Our first result is for hypergraph coloring.
\begin{theorem}
\label{theorem-coloring-gen}
There is an algorithm such that given any $k$-uniform hypergraph with maximum degree $\Delta$ and a set of colors $[q]$, 
assuming $k\geq 13$ and $q \geq 	\max\tp{\tp{7k\Delta}^{\frac{9}{k-12}}, 650}$,
the algorithm returns a random $q$-coloring $\Ass{X}\in [q]^V$ in time
$O\tp{q^2k^3\Delta^2 n \tp{\frac{n}{\epsilon}}^{\frac{1}{100(qk\Delta)^4}} \log ^4 \tp{\frac{nqk\Delta}{\epsilon} } }$
such that 
the distribution of $\Ass{X}$ is $\epsilon$-close in total variation distance to the uniform distribution of all proper $q$-colorings of the input hypergraph.	
\end{theorem}
\noindent
\Cref{theorem-coloring} is implied by \Cref{theorem-coloring-gen}: when $k \geq 30$, we have $(7k)^{\frac{9}{k-12}} \leq 15$, which means that $q \geq 15\Delta^{\frac{9}{k-12}} + 650$ suffices to imply the condition in \Cref{theorem-coloring-gen}. 

Our next result is for CNF formulas.
For a $k$-CNF formula, each clause contains $k$ variables.
And the \emph{maximum degree} of the formula is  given by the maximum number of clauses a variable belongs to.
The following theorem is is a formal restatement of \Cref{theorem-CNF}.

\begin{theorem}
\label{theorem-CNF-gen}	
The following holds for any $0< \zeta \leq 2^{-20}$. 
There is an algorithm such that given any $k$-CNF formula with maximum degree $d$, 
assuming $k \geq 13 \log d + 13 \log k + 3 \log \frac{1}{\zeta}$,  
the algorithm returns a random assignment $\Ass{X}\in \{\True,\False\}^V$ in time
$O\tp{d^2k^3 n \tp{\frac{n}{\epsilon}}^{{\zeta}/{(dk)^4}} \log ^3 \tp{\frac{ndk}{\epsilon} } }$
such that 
the distribution of $\Ass{X}$ is $\epsilon$-close in total variation distance to the uniform distribution of all satisfying assignments.
\end{theorem}
\noindent

Let $\Phi = (V, [q]^V, \Cons{C})$ denote the CSP formula where all variables have the same domain $[q]$.
Suppose that for every constraint $c \in \Cons{C}$, $c$ is atomic and $\abs{\vbl{c}} = k$, and each variable belongs to at most $d$ constraints.
Let $\Proj{h}$ denote a projection scheme satisfying \Cref{condition-projection} with parameters $\alpha$ and $\beta$.
For such special CSP formulas, we have the following lemma with an improved  mixing condition.
\begin{lemma}
\label{lemma-mixing-kd}
If $k\log q \geq \frac{1}{\beta} \log \tp{3000 q^2d^6k^6}$, then the Markov chain $P_{\mathrm{Glauber}}$ on $\nu$ has
%\begin{align*}
$\tmix(\epsilon) \leq \ctp{2n \log \frac{n}{\epsilon}}$.
%\end{align*}
\end{lemma}
\noindent
The proof of \Cref{lemma-mixing-kd} is given in \Cref{section-mixing}.
We use \Cref{lemma-simulation} and \Cref{lemma-mixing-kd} to prove our results.

\begin{proof}[Proof of \Cref{theorem-coloring-gen}]
Consider the hypergraph $q$-coloring on a $k$-uniform hypergraph $H=(V,\+E)$ with maximum degree $\Delta$.
We first transform the hypergraph coloring instance into a CSP formula $\Phi = (V, [q]^V, \Cons{C})$ with atomic constraints.
For each hyperedge $e \in \+E$, we add $q$ constraints such that the $i$-th constraint $c_i$ forbids the bad event that the hyperedge $e$ is monochromatic with color $i \in [q]$. 
Namely, $\vbl{c_i} = e$ and $c_i$ is $\False{}$ if and only if all variables in $\vbl{c_i}$ take the value $i$.
The time complexity for this reduction is $O(nq\Delta\log q)$.

In CSP formula  $\Phi = (V, [q]^V, \Cons{C})$, $c$ is atomic and $\abs{\vbl{c}} = k$ for all $c \in \Cons{C}$; each variable belongs to at most $q \Delta$ constraints.
The maximum degree  $D$ of the dependency graph of $\Phi$ is at most $qk\Delta$.
We assume $D = qk\Delta$.
If  each variable $v \in V$ draws a random value from $[q]$ uniformly and independently, then the maximum probability $p$ that one constraint becomes $\False$ is 
$p = \tp{\frac{1}{q}}^k$.

Let $\alpha,\beta,\eta$ be three parameters to be fixed later.
Our algorithm first uses the deterministic algorithm in \Cref{theorem-projection-uniform} to construct a projection scheme satisfying \Cref{condition-projection} with parameters $\alpha$ and $\beta$. 
The deterministic algorithm in \Cref{theorem-projection-uniform} always finds such a projection scheme, which gives a projection oracle with query cost $O(\log q)$.
Remark that the cost for constructing the projection scheme is 
\begin{align}
\label{eq-proj-coloring}
T_{\mathrm{proj}} = O\tp{n \log q}.	
\end{align}
We  then run \Cref{alg-mcmc} to obtain the output $\Ass{X}_{\mathrm{out}} = \Ass{X}_{\mathrm{alg}}$, where $\Ass{X}_{\mathrm{alg}}$ denotes the output of \Cref{alg-mcmc}.
The correctness result can be proved by going through the proof of \Cref{theorem-main}.

We set parameters $\alpha, \beta$ and $\eta$. 
Note that $\vbl{c} = k$ for all $c \in \Cons{C}$; $p = q^{-k}$; and each variable belongs to at most $d = q\Delta $ constraints; and $D = qk\Delta$.
We put all the constraints in \Cref{theorem-projection-uniform}, \Cref{lemma-mixing-kd} and \Cref{lemma-simulation} together:
\begin{align*}
0 < \beta < \alpha <1,\quad 7&\le q^{\frac{\alpha+\beta}{2}}\le \frac{q}{6},\quad \log q \geq \frac{1}{\alpha-\beta},\quad 0 < \eta < 1;\\
k\log q &\geq \frac{1}{\beta} \log \tp{3000 q^8\Delta^6k^6};\\	
%\log \frac{1}{p} &\geq \frac{100}{\beta} \log \tp{\frac{2000D^5}{\beta}};\\
k\log q &\geq \frac{1}{1-\alpha}\log(20q^2k^2\Delta^2); \\
k\log q &\geq \frac{1}{\beta}\log\tp{\frac{40\mathrm{e} q^2k^2\Delta^2}{\eta}}.
\end{align*}
We can take $\alpha = \frac{7}{9}$ and $\beta = \frac{2}{3}$. 
The following condition suffices to imply all the  above constraints: assume $k > 12$,
\begin{align*}	
\log q &\geq \frac{9}{k-12} \log \Delta + \frac{9}{k - 12}\log k + \frac{25}{k - 12},\quad q \geq 650, \quad \eta = \frac{1}{2^{9}(qk\Delta)^4}.
\end{align*}
The following condition suffices to imply the above one
\begin{align*}
q \geq 	\max\tp{\tp{7k\Delta}^{\frac{9}{k-12}}, 650} \quad \text{and}\quad \eta = \frac{1}{2^{9}(qk\Delta)^4}.
\end{align*}
Note that $D=kq\Delta$.
Under this condition,  by~\eqref{eq-T-step},~\eqref{eq-T-final} and~\eqref{eq-proj-coloring}, the total running time is 
\begin{align*}
T_{\mathrm{total}} = O\tp{D^2k n \tp{\frac{n}{\epsilon}}^{3\eta} \log ^3 \tp{\frac{nD}{\epsilon} } \log q } = O\tp{q^2k^3\Delta^2 n \tp{\frac{n}{\epsilon}}^{\frac{1}{100(qk\Delta)^4}} \log ^4 \tp{\frac{nqk\Delta}{\epsilon} } }. &\qedhere
\end{align*}
\end{proof}

\begin{proof}[Proof of \Cref{theorem-CNF-gen}]
Let $\Phi = (V, \{\True,\False\}^V, \Cons{C})$ be a $k$-CNF formula, where each variable belongs to at most $d$ clauses. 
Each variable takes its value for the Boolean domain $\{\True,\False\}$, thus the size of the domain is $q = 2$.
The maximum degree $D$ of the dependency graph is at most $kd$.
We  assume $D = kd$.
If  each variable $v \in V$ draws a random value from the Boolean domain $\{\True,\False\}$ uniformly and independently, the maximum probability $p$ that one clause is not satisfied is 
$p = \tp{\frac{1}{2}}^k$.	

Let $\alpha,\beta,\eta$ be three parameters to be fixed later.
Our algorithm first uses the randomized algorithm in \Cref{theorem-projection-uniform} with $\delta = \frac{\epsilon}{4}$ to construct a projection scheme satisfying \Cref{condition-projection} with parameters $\alpha$ and $\beta$. 
If the randomized algorithm in \Cref{theorem-projection-uniform} fails to find such projection scheme, 
our algorithm terminates and outputs an arbitrary $\Ass{X}_{\mathrm{out}} \in \{\True,\False\}^V$. If the randomized algorithm in \Cref{theorem-projection-uniform} succeeds, it gives a projection oracle with query cost $O(\log q)$.
By \Cref{theorem-projection-uniform}, the cost for constructing the projection scheme is 
\begin{align}
\label{eq-projection-CNF}
T_{\mathrm{proj}} = O\tp{n dk \log \frac{1}{\epsilon}}.	
\end{align}
We then run \Cref{alg-mcmc} to obtain the output $\Ass{X}_{\mathrm{out}} = \Ass{X}_{\mathrm{alg}}$, where $\Ass{X}_{\mathrm{alg}}$ denotes the output of \Cref{alg-mcmc}.
The correctness result can be proved by going through the proof of \Cref{theorem-main}.

We set parameters $\alpha, \beta$ and $\eta$. We put all the constraints in \Cref{theorem-projection-uniform}, \Cref{lemma-mixing-kd} and \Cref{lemma-simulation} together:	
\begin{align*}
0  <\beta &< \alpha < 1, \quad k \geq \frac{2\ln 2}{(\alpha - \beta)^2}\log(2\mathrm{e}kd), \quad 0 < \eta < 1;\\
k &\geq \frac{1}{\beta} \log \tp{3000 \cdot 4 \cdot d^6k^6};\\
k &\geq \frac{1}{1-\alpha}\log(20d^2k^2);\\
k &\geq \frac{1}{\beta}\log\tp{\frac{40\mathrm{e} d^2k^2}{\eta}}.
\end{align*}
We can take $\alpha = \frac{21}{25}$ and $\beta = \frac{1}{2}$. 
The following condition suffices to imply all the  above constraints %let $0 <\zeta \leq 2^{-20}$ be a parameter,
\begin{align*}
k \geq 13 \log d + 13\log k +3 \log \frac{1}{\zeta} \quad\text{and}\quad
\eta = \frac{\zeta}{3d^4k^4}, \quad \text{where }	0 <\zeta \leq 2^{-20}.
\end{align*}
Note that $D=dk$ and $q = 2$.
Under this condition,  by~\eqref{eq-T-step},~\eqref{eq-T-final} and~\eqref{eq-projection-CNF}, the total running time is 
\begin{align*}
T_{\mathrm{total}} = O\tp{D^2k n \tp{\frac{n}{\epsilon}}^{3\eta} \log ^3 \tp{\frac{nD}{\epsilon} } \log q }= O\tp{d^2k^3 n \tp{\frac{n}{\epsilon}}^{\frac{\zeta}{d^4k^4}} \log ^3 \tp{\frac{ndk}{\epsilon} } }. &\qedhere
\end{align*}

\end{proof}

\pagebreak
\section{Projection construction}
\label{section-projection-construction}

In this section, we give the algorithms to construct the projection schemes.
We first give the projection algorithm for $(k,d)$-CSP formulas (\Cref{theorem-projection-uniform}), then  give the projection algorithm for general CSP formulas (\Cref{theorem-projection-general}).
\begin{proof}[Proof of \Cref{theorem-projection-uniform}]
%\wtodo{repalce $\alpha$ with $1 - \alpha$}
We start from the first part of the lemma.
For each $v \in V$, we set $s_v$ as 
\begin{align*}
s_v = \ctp{q^{\frac{2-\alpha-\beta}{2}}}.
\end{align*}
For each variable $v \in V$, we partition $[q] = \{1,2,\ldots,q\}$ into $s_v$ intervals, where the sizes of the first $(q \mod s_v)$ intervals are $\lceil q / s_v \rceil$, and the sizes of the last $s_v - (q \mod s_v)$ intervals are $\lfloor q /s_v \rfloor$. Let $\Sigma_v = \{1,2,\ldots,s_v\}$. For each $i \in [q]$, $h_v(i) = j \in \Sigma_v$, where $i$ belongs to the $j$-th interval. This constructs the function $h_v:[q] \to \Sigma_v$. To implement the projection oracle, we only need to calculate $s_v$ for each $v \in V$, the total cost is $O(n\log q)$. Consider the two queries in \Cref{def:projection-oracle}.
\begin{itemize}
\item evaluation: given an input value $i \in [q]$ of a variable $v \in V$, the algorithm should return $j \in \Sigma_v$ such that $i$ is in the $j$-th interval, this query can be answered with the cost $O(\log q)$;
\item inversion: given a projected value $j \in \Sigma_v$ of a variable $v \in V$, the algorithm should return a random element in the $j$-th interval uniformly at random, this query can be answered with the cost $O(\log q)$.
\end{itemize}

Next, we prove that this projection scheme satisfies \Cref{condition-projection}.
For any $v \in V$, it holds that
\begin{align*}
\ctp{\frac{q}{s_v}}	\leq \ctp{q^{(\alpha+\beta)/2}} \leq q^{(\alpha+\beta)/2}+1 \overset{\diamondsuit}{\leq} \frac{7}{6}q^{(\alpha+\beta)/2},
\end{align*}
where $(\diamondsuit)$ holds because $q^{(\alpha+\beta)/2}+1 \leq \frac{7}{6}q^{(\alpha+\beta)/2}$ if $q^{(\alpha+\beta)/2} \geq 6$.
Note that $\log \frac{7}{6} \leq 0.23$.
This implies the following inequality
\begin{align}
\label{eq-large-q-1}
\sum_{v \in \vbl{c}}\ctp{\frac{q}{s_v}}	\leq k \tp{\frac{\alpha+\beta}{2} \log q + 0.23} \overset{(\star)}{\leq} k \cdot \alpha \log q = \alpha \sum_{v \in \vbl{c}}{\log q},
\end{align}
where inequality $(\star)$ holds because $\alpha > \beta$ and $\log q \geq \frac{0.8}{\alpha-\beta}$. For any $v \in V$, it holds that
\begin{align*}
\ftp{\frac{q}{s_v}} = \ftp{ \frac{q}{\ctp{q^{(2-\alpha-\beta)/2}}} }	\geq \ftp{ \frac{q}{q^{(2-\alpha-\beta)/2} + 1} }	\overset{(\ast)}{\geq} \ftp{ \frac{q}{\tp{1 + \frac{1}{6}}q^{(2-\alpha-\beta)/2}} }\geq \frac{6}{7} q^{\frac{\alpha+\beta}{2}} - 1 \overset{(\diamondsuit)}{\geq} \frac{5}{7}q^{\frac{\alpha+\beta}{2}},
\end{align*}
where inequality $(\ast)$ holds because $\tp{1 + \frac{1}{6}}q^{(2-\alpha-\beta)/2} \geq q^{(2-\alpha-\beta)/2} + 1$ if $q^{(2-\alpha-\beta)/2} \geq 6$;
inequality $(\diamondsuit)$ holds because $q^{(\alpha+\beta)/2} \geq 7$. 
Note that $\log \frac{5}{7} \geq -0.5$.
This implies
\begin{align}
\label{eq-large-q-2}
\sum_{v \in \vbl{c}}\log\ftp{\frac{q}{s_v}} \geq k \tp{ \frac{\alpha+\beta}{2} \log q  - 0.5 } \overset{(\star)}{\geq} k \cdot \beta \log q = \beta \sum_{v \in \vbl{c}}\log q,
\end{align}
where inequality $(\star)$ holds because $\alpha > \beta$ and  $\log q \geq \frac{1}{\alpha-\beta}$.
Combining~\eqref{eq-large-q-1} and~\eqref{eq-large-q-2} proves the first part of the lemma.

We then prove the second part of the lemma.
The algorithm constructs a subset of variables $\+M \subseteq V$. We call $\+M$ the set of marked variables.
If $v \in \+M$, let $\Sigma_v = [q]$, and $h_v(i) = i$ for all $i \in [q]$. 
If $v \not\in \+M$, let $\Sigma_v = \{1\}$, and $h_v(i) = 1$ for all $i \in [q]$.
Remark that $s_v = q$ if $v$ is a marked variable, and $s_v = 1$ if $v$ is an unmarked variable.
%In other words, for each variable $v$, the algorithm either partitions $[q]$ into $q$ disjoint sets (if $v \in \+M$), or leaves  $[q]$ as it is (if $v \notin \+M$). 
To implement the projection oracle, we only need to construct $\+M$.
Suppose the set $\+M$ is given (the construction will be explained later).
Consider the two queries in \Cref{def:projection-oracle}.
\begin{itemize}
\item evaluation: given an input value $i \in [q]$ of a variable $v \in V$, the algorithm should return the input~$i$ if $v \in \+M$, or return $1 \in \Sigma_v$ if $v \notin \+M$; this query can be answered with the cost $O(\log q)$;
\item inversion: given a projected value $j \in \Sigma_v$ of a variable $v \in V$, the algorithm should return $j \in [q]$ if $v \in \+M$; or return a uniform random element $X \in [q]$ if $v \notin \+M$; this query can be answered with the cost $O(\log q)$.	
\end{itemize}

Now, we construct the set of marked variables $\+M \subseteq V$.
For each constraint $c \in \Cons{C}$, define $t_c$ as the number of marked variables in $c$, i.e.
\begin{align*}
t_c \triangleq \abs{\+M \cap \vbl{c}}.	
\end{align*}
Hence, \Cref{condition-projection} becomes for each $c \in \Cons{C}$,
\begin{align*}
(1-\alpha) k \leq t_c \leq (1 - \beta)k.	
\end{align*}
In other words, each constraint contains at least $(1-\alpha) k$ marked variables and at least $
\beta k$ unmarked variables. 
We use Lov\'asz local lemma to show that such set $\+M$ exists, then use Moser-Tardos algorithm to find a set $\+M$. 
Let $\+D$ denote the product distribution such that each variables is marked independently with probability $\frac{2-\alpha-\beta}{2}$. For each constraint $c \in \Cons{C}$, let $B_c$ denote the bad event that $c$ contains less than $(1-\alpha) k$ marked variables or less than $\beta k$ unmarked variables.
We use concentration inequality to bound the probability of $B_c$. 
In~\cite{FGYZ20}, the probability of the bad event $B_c$ is bounded by the Chernoff bound. Now, we use Hoeffding's inequality to obtain a better result 
\begin{align*}
\Pr[\+D]{B_c} = \Pr[]{t_c < (1-\alpha) k \lor t_c > (1-\beta)k}	 = \Pr{\abs{t_c - \E{t_c}} \geq \frac{\alpha-\beta}{2}k}\leq 2 \exp \tp{ - \frac{(\alpha-\beta)^2}{2}k }.
\end{align*}
The maximum degree of dependency graph is at most $k(d-1)$.
By Lov\'asz local lemma (\Cref{theorem-LLL}), the set $\+M$ exist if
\begin{align*}
\mathrm{e}\cdot 2 \exp \tp{ - \frac{(\alpha-\beta)^2}{2}k } \cdot kd \leq 1.
\end{align*}
Note that $\alpha > \beta$ and $k \geq \frac{2\ln 2}{(\alpha - \beta)^2}\tp{\log k + \log d + \log 2\mathrm{e}}$ implies the above condition.

The Moser-Tardos algorithm can find such set $\+M$ within $\frac{2n}{k}$ resampling steps in expectation~\cite{moser2010constructive}.
We can run $\ctp{\log \frac{1}{\delta}}$ Moser-Tardos algorithms independently, then with probability at least $1 - \delta$, one of them finds the set $\+M$ within $\frac{4n}{k}$ resampling steps. 
The cost of each resampling step is $O(dk^2)$.
The cost for constructing data structure is $O(ndk\log\frac{1}{\delta})$.
\end{proof}

\begin{proof}[Proof of \Cref{theorem-projection-general}]
The domain of each variable $v \in V$ is $Q_v$, where $q_v = \abs{Q_v}$.
Assume each element $x \in Q_v$ can be in-coded by $O(\log q_v)$ bits.
For each $v \in V$, suppose the input provides an array $\+A_v$ of size $q_v$ containing all the elements in $Q_v$.
For each $v \in V$, we construct a data structure $\+S_v$ that can answer the following two types of the queries:
(1) given any index $i \in [q_v]$, we can access the $i$-th element in this array with cost $O(\log q_v)$.  
(2) given any $x \in Q_v$, we can find the unique index $i$ such that $\+A_v(i) = x$ with the cost $O(\log q_v)$.
For each $v \in V$, the cost of the construction is $O(q_v \log q_v)$.

The algorithm divides all variables into two parts $S_{\mathrm{large}}$ and $S_{\mathrm{small}}$ such that
\begin{align*}
S_{\mathrm{large}} = \left\{v \in V \mid  \log q_v \geq \frac{5}{\alpha-\beta} \right\}, \qquad S_{\mathrm{small}} = \left\{v \in V \mid \log q_v < \frac{5}{\alpha-\beta} \right\}.
\end{align*}
For each variable $v \in S_{\mathrm{large}}$, the algorithm sets 
\begin{align*}
\forall v\in S_{\mathrm{large}}, \quad s_v = \ctp{q_v^{\frac{2 - \alpha-\beta}{2}}}.
\end{align*}
We partition $[q] = \{1,2,\ldots,q\}$ into $s_v$ intervals, where the sizes of the first $(q \mod s_v)$ intervals are $\lceil q / s_v \rceil$, and the sizes of the last $s_v - (q \mod s_v)$ intervals are $\lfloor q /s_v \rfloor$. Let $\Sigma_v = \{1,2,\ldots,s_v\}$, where each $j \in \Sigma_v$ represents an interval $[L_j, R_j]$. For any $x \in Q_v$, let $i$ denote the unique index such that $\+A_v(i) = x$, we set $h_v(x) = j$ such that $i \in [L_j,R_j]$. This defines the function $h_v: Q_v \to \Sigma_v$. 
To implement the projection oracle for $S_{\mathrm{large}}$, the algorithm only needs to compute the value of $s_v$, where the cost is $O(\log q_v)$.
Consider the two queries of the projection oracle in \Cref{def:projection-oracle}.
\begin{itemize}
\item  evaluation: given an input value $x \in Q_v$ of a variable $v \in S_{\mathrm{large}}$, with the data structure $\+S_v$, the algorithm can return $h_v(x)$ in time $O(\log q_v)$; 	
\item inversion: given a projected value $j \in \Sigma_v$ of a variable $v \in  S_{\mathrm{large}}$, the algorithm should return a uniform element in set $\{x \in \+A_v(i) \mid  L_j \leq i \leq R_j \}$; with the data structure $\+S_v$, this query can be answered with the cost $O(\log q_v)$.
\end{itemize}
Let $q = \max_{v \in V}q_v$. For any $v \in S_{\mathrm{large}}$, the cost for answering each query is $O(\log q)$.

For variables in $S_{\mathrm{small}}$, the algorithm constructs a subset of variables $\+M \subseteq S_{\mathrm{small}}$. 
%and sets 
%\begin{align*}
%\forall v \in S_{\mathrm{small}},\quad
%s_v = \begin{cases}
%q_v &\text{if } v \in \+M\\
%1 &\text{if } v \in  S_{\mathrm{small}} \setminus \+M.	
% \end{cases}
%\end{align*}
We call $\+M$ the set of marked variables.
If $v \in \+M$, let $\Sigma_v = Q_v$, and $h_v(x) = x$ for all $x \in Q_v$. 
If $v \not\in \+M$, let $\Sigma_v = \{1\}$, and $h_v(x) = 1$ for all $x \in Q_v$.
To implement the projection oracle, the algorithm only needs to construct the set $\+M$.
The construction of $\+M$ will be explained later.
Suppose the set $\+M \subseteq S_{\mathrm{small}}$ is given.
Consider the two queries of the projection oracle in \Cref{def:projection-oracle}.
\begin{itemize}
\item evaluation: given an input value $x \in Q_v$ of a variable $v \in S_{\mathrm{small}}$, the algorithm should return the input $x$ if $v \in \+M$, or return $1 \in \Sigma_v$ if $v \notin \+M$; this query can be answered in time $O(\log q_v)$;
\item inversion: given a projected value $x \in \Sigma_v$ of a variable $v \in S_{\mathrm{small}}$, the algorithm should return the input $x$ if $v \in \+M$; or return a uniform random element $X \in Q_v$ if $v \notin \+M$; with the data structure $\+S_v$, this query can be answered in time $O(\log q_v)$.
\end{itemize}
Let $q = \max_{v \in V}q_v$. For any $v \in S_{\mathrm{small}}$, the cost for answering each query is $O(\log q)$.

Again, we use Lov\'asz local lemma to prove that there is a subset $\+M$ such that the above projection scheme satisfies \Cref{condition-projection}, then use Moser-Tardos algorithm to find such set $\+M$.
Let $\+D$ denote the product distribution such that each variable $v \in  S_{\mathrm{small}}$ is marked with probability $\frac{2-\alpha-\beta}{2}$. For each $c \in \Cons{C}$, let $B_c$ denote the bad event
\begin{align}
\label{eq-mark-condition-gen}
\sum_{v \in \vbl{c}}\log{\ctp{\frac{q_v}{s_v}}}	 > \alpha\sum_{v \in \vbl{c}}\log q_v \quad\text{or}\quad \sum_{v \in \vbl{c}}\log\ftp{\frac{q_v}{s_v}} < \beta \sum_{v \in \vbl{c}}\log q_v.
\end{align}
Fix a constraint $c \in \Cons{C}$. Suppose $v_1,v_2,\ldots,v_k$ are variables in $\vbl{c}$, where $k = k(c) = \abs{\vbl{c}}$. Let $0\leq \ell \leq k$ be an integer and assume $v_i \in S_{\mathrm{large}}$ for all $1\leq i \leq \ell$ and $v_j \in S_{\mathrm{small}}$ for all $\ell+ 1\leq j \leq k$. 
For each $1\leq i \leq k$, we define random variable
\begin{align*}
X_i \triangleq 	\log\ctp{\frac{q_{v_i}}{s_{v_i}}}.
\end{align*}
For each $1\leq i \leq \ell$, since $v_i \in S_{\mathrm{large}}$, $X_i = \log\ctp{q_{v_i} / \lceil q_{v_i}^{(2-\alpha-\beta)/2}\rceil }$ with probability 1. We have
\begin{align*}
\forall 	1\leq i \leq \ell, \quad \E{X_i} &= \log \ctp{\frac{q_{v_i}}{\lceil q_{v_i}^{(2-\alpha-\beta)/2}\rceil}} \leq \log \ctp{ q_{v_i}^{(\alpha+\beta)/2} } \leq \log\tp{ \frac{5}{4} q_{v_i}^{(\alpha+\beta)/2}},
\end{align*}
where the last inequality holds because $\log q_{v_i} \geq \frac{5}{\alpha - \beta}$, which implies $\frac{5}{4} q_{v_i}^{(\alpha+\beta)/2} \geq q_{v_i}^{(\alpha+\beta)/2} + 1 \geq \ctp{ q_{v_i}^{(\alpha+\beta)/2} }$. Note that $\log \frac{5}{4} \leq 0.33$ and $\log q_{v_i} \geq \frac{5}{\alpha - \beta}$. It holds that
\begin{align}
\label{eq-EX1}
\forall 	1\leq i \leq \ell, \quad \E{X_i} \leq 0.33 + \frac{\alpha + \beta}{2}\log q_{v_i} \leq 
\alpha \log q	_{v_i} - \frac{\alpha-\beta}{3} \log q_{v_i}.
\end{align}
For each $\ell + 1\leq j \leq k$, since  $v_j \in S_{\mathrm{small}}$, $X_j = \log q_{v_j}$ with probability $\frac{\alpha+\beta}{2}$; and $X_j = 0$ with probability $\frac{1-\alpha-\beta}{2}$. We have
\begin{align}
\label{eq-EX2}
\forall 	\ell + 1\leq j \leq k, \quad \E{X_i}	 &= \frac{\alpha+\beta}{2}\log q_{v_i} \leq \alpha \log q_{v_i} - \frac{\alpha-\beta}{3} \log q_{v_i}.
\end{align}
Consider the sum $\sum_{i = 1}^k X_i$. For any $v_i \in S_{\mathrm{large}}$, the value of $X_i$ is fixed. For any $v_j \in S_{\mathrm{small}}$, $X_j$ takes a random value and it must hold that $ X_j\in \{0, \log q_{v_i}\}$. 
By Hoeffding's inequality, 
\begin{align}
\label{eq-X-bounded}
\Pr[\+D]{\sum_{i = 1}^k X_i > \sum_{i=1}^k \E{X_i} + t } \leq \exp\tp{- \frac{2t^2}{\sum_{j=\ell + 1}^k \log^2 q_{v_j}}} \overset{(\star)}{\leq} \exp\tp{- \frac{2(\alpha-\beta)t^2}{5\sum_{j=\ell + 1}^k \log q_{v_j}}},
\end{align}
where $(\star)$ holds due to $\log q_{v_j} \leq \frac{5}{\alpha-\beta}$ for all $\ell + 1 \leq j \leq k$.
Combining~\eqref{eq-EX1},~\eqref{eq-EX2} and~\eqref{eq-X-bounded}, we have
\begin{align}
\label{eq-bad-X}
\Pr[\+D]{ \sum_{i=1}^k X_i > \alpha \sum_{i = 1}^k \log q_{v_i} } \leq \exp\tp{- \frac{\frac{2(\alpha-\beta)^3}{9} \tp{\sum_{i=1}^k \log q_{v_i} }^2 }{ 5\sum_{j = \ell + 1}^k \log q_{v_j} } } \leq \exp \tp{- \frac{(\alpha-\beta)^3}{23}\sum_{i=1}^k \log q_{v_i} }.
\end{align}

Similarly, for each $1\leq i \leq k$, we define random variable
\begin{align*}
Y_i \triangleq 	\log \ftp{\frac{q_{v_i}}{s_{v_i}}}.
\end{align*}
For each $1\leq i \leq \ell$, since $v_i \in S_{\mathrm{large}}$, $Y_i = \log \ftp{\frac{q_{v_i}}{\lceil q_{v_i}^{(2-\alpha-\beta)/2}\rceil }}$ with probability 1. We have
\begin{align*}
\forall 	1\leq i \leq \ell, \quad \E{Y_i} &=	\log \ftp{\frac{q_{v_i}}{\ctp{ q_{v_i}^{(2-\alpha-\beta)/2}}}} \geq \log \ftp{ \frac{4}{5}{q_{v_i}^{(\alpha+\beta)/2}}} \geq \log \tp{ \frac{3}{5}{q_{v_i}^{(\alpha+\beta)/2}}},
\end{align*}
where the last two inequalities hold because $0<\beta<\alpha<1$ and $\log q_{v_i} \geq \frac{5}{\alpha - \beta}$, which implies $\frac{5}{4}q_{v_i}^{(2-\alpha-\beta)/2} \geq q_{v_i}^{(2-\alpha-\beta)/2} + 1 \geq \ctp{ q_{v_i}^{(2-\alpha-\beta)/2}}$ and $\ftp{\frac{4}{5} q_{v_i}^{(\alpha+\beta)/2}} \geq \frac{4}{5} q_{v_i}^{(\alpha+\beta)/2} - 1 \geq \frac{3}{5}{q_{v_i}^{(\alpha+\beta)/2}}$. Note that $\log \frac{3}{5} \geq -0.74$. Again, by $\log q_{v_i} \geq \frac{5}{\alpha - \beta}$, we have 
\begin{align*}
\forall 	1\leq i \leq \ell, \quad \E{Y_i} \geq -0.74 + \frac{\alpha + \beta}{2}\log q_{v_i} \geq \beta \log q_{v_i} + \frac{\alpha - \beta}{3}\log q_{v_i}.
\end{align*}
For each $\ell + 1\leq j \leq k$, since  $v_j \in S_{\mathrm{small}}$, $Y_j = 0$ with probability $\frac{2-\alpha-\beta}{2}$; and $Y_j = \log q_{v_j}$ with probability $\frac{\alpha+\beta}{2}$. We have
\begin{align*}
\forall 	\ell + 1\leq j \leq k, \quad \E{Y_i}	 &= \frac{\alpha+\beta}{2}\log q_{v_i} \geq \beta \log q_{v_i} + \frac{\alpha-\beta}{3} \log q_{v_i}.
\end{align*}
Again, by Hoeffding's inequality,  we have
\begin{align}
\label{eq-bad-Y}
\Pr[\+D]{ \sum_{i=1}^k Y_i < \beta \sum_{i = 1}^k \log q_{v_i} } \leq  \exp \tp{- \frac{(\alpha-\beta)^3}{23}\sum_{i=1}^k \log q_{v_i} }.
\end{align}
Combining~\eqref{eq-bad-X} and~\eqref{eq-bad-Y} we have
\begin{align*}
\Pr[\+D]{B_c} \leq  2\exp \tp{- \frac{(\alpha-\beta)^3}{23}\sum_{i=1}^k \log q_{v_i} } \overset{(\star)}{\leq} 2\exp\tp{-\frac{25}{23}\log D - 3} \leq 2\exp\tp{-\frac{25}{23}\ln D - 3} \leq \frac{1}{\mathrm{e}(D+1)},
\end{align*}
where $(\star)$ holds because $\sum_{i=1}^k \log{q_{v_i}} \geq \frac{25}{(\alpha-\beta)^3} \tp{\log D + 3}$. By Lov\'asz local lemma, there exists a set of marked variables $\+M \subseteq S_{\mathrm{small}}$ such that the condition in~\eqref{eq-mark-condition-gen} is satisfied.

Similar to the proof of \Cref{theorem-projection-uniform}, we can use Moser-Tardos algorithm~\cite{moser2010constructive} to construct such projection scheme.  With probability at least $1 - \delta$, the algorithm constructs a projection scheme in time $O(nDk\log\frac{1}{\delta})$, where $k = \max_{c \in \Cons{C}}\abs{ \vbl{c}}$.

We now combine all the steps together. The construction of the data structures $\+S_v$ for all $v \in V$ has the cost $O(n q \log q)$. Computing the $s_v$ for all $v \in S_{\mathrm{large}}$ has the costs $O(n \log q)$. Computing the marked set $\+M \subseteq S_{\mathrm{small}}$ has the cost $O(nDk\log\frac{1}{\delta})$. The total cost is $O(n(Dk + q) \log \frac{1}{\delta} \log q)$. 
\end{proof}

\section{Analysis of the Inverse Sampling subroutine}
\label{section-subroutine}
In this section, we prove \Cref{lemma-simulation}.
Let $\Phi = (V, \Dom{Q},\Cons{C})$ be a CSP formula, where each variable $v$ takes value in $Q_v$.
Let $\Proj{h}=(h_v)_{v \in V}$ be a balanced projection scheme satisfying \Cref{condition-projection} with parameters $\alpha$ and $\beta$, where for each $v \in V$, $h_v: Q_v \to \Sigma_v$, $\abs{Q_v} = q_v$ and $\abs{\Sigma_v} = s_v$.
Let $(Y_t)_{t \geq 0}$ denote random sequence generated by \Cref{alg-mcmc}, where $Y_t \in \Sigma$ is the random  $Y$ after the $t$-th iteration of the for-loop.
Recall that for each $1\leq t \leq T+1$, we have defined the following bad events:
\begin{itemize}
\item $\+B^{(1)}_t$: in the $t$-th call of $\sample(\cdot)$,	the random assignment $\Ass{X}$ is returned in \Cref{line-bad-return-2};
\item $\+B^{(2)}_t$: in the $t$-th call of $\sample(\cdot)$, the random assignment $\Ass{X}$ is returned in \Cref{line-bad-return-1}.
\end{itemize}

In the $t$-th calling of the subroutine $\sample(\Phi, \Proj{h} ,\delta, y_\Lambda, S)$ (\Cref{alg-sample}), conditional on $\neg\+B^{(1)}_t\land\neg\+B^{(2)}_t$, all the connected components that intersect with $S$ are small, and the rejection sampling on each component succeeds.
It is straightforward to verify the subroutine  returns a perfect sample from  $\mu^{y_{\Lambda}}_S$.

Next, we analyze the running time of the subroutine $\sample(\Phi, \Proj{h} ,\delta, y_\Lambda, S)$.
Let $G=(\Cons{C} ,E)$ denote the dependency graph of $\Phi=(V, \Dom{Q},\Cons{C})$.
We assume the dependency graph is stored in an adjacent list.
We can construct such  adjacent list at the beginning of the whole algorithm.
The cost of construction is $O(nDk)$, which is dominated by the cost in \Cref{theorem-main-gen}.

Assume that the algorithm can access a projection oracle with query cost $O(\log q)$.
The first step of the subroutine is to find all the connected components that intersect with set $S$.
For each variable $v \in S$, we find all the constraints $C(v) = \{c \in \Cons{C} \mid v \in \vbl{c} \}$ (note that $\abs{C(v)} \leq D$),
then perform a deep first search (DFS) in $G$ starting from $C(v)$.
During the DFS, suppose the current constraint is $c \in \Cons{C}$.
We can find the unique configuration $\sigma \in Q_{\vbl{c}}$ forbidden by $c$, i.e. $c(\sigma) = \False$.
We call the projection oracle to obtain $\tau \in \Sigma_{\vbl{c}}$, where $\tau_v = h_v(\sigma_v)$ for each $v \in V$. 
The cost of this step is $O(k\log q)$. 
If for all $v \in \Lambda \cap \vbl{c}$, $y_{\Lambda}(v) = \tau_v$ (which means $c$ is not satisfied by $y_{\Lambda}$), we do DFS recursively starting from $c$; otherwise, we stop current DFS branch and remove $c$ from the graph $G$.
If the size of current connected component is greater than $2D \log \frac{nD}{\delta}$, the connected component is too large, we stop the whole DFS process. The total cost of DFS is 
\begin{align*}
T_{\mathrm{DFS}}  = O\tp{\abs{S}D^2 k \log \frac{nD}{\delta} \log q }.	
\end{align*} 

Another cost of the subroutine comes from the rejection sampling from \Cref{line-rejection-begin} to \Cref{line-bad-return-2}. 
To perform the rejection sampling, for each variable $v$, we either draw $X_v$ from $\pi_v^{y_v}$ or draw $X_v$ from the $\pi_v$. 
This step can be achieved by calling oracles. 
The cost is $O(\log q)$.
Since there are at most $\abs{S}$ connected components and each of the size at most $2D \log \frac{nD}{\delta}$, the total number of variables is $O(\abs{S}Dk \log \frac{nD}{\delta} )$.
For each component, the algorithm uses the rejection sampling for at most  $R = \ctp{ 10\tp{\frac{n}{\delta}}^{\eta} \log \frac{n}{\delta}}$ times. The total cost of rejection sampling is
\begin{align*}
T_{\mathrm{rej}} = O\tp{\abs{S}Dk \tp{\frac{n}{\delta}}^{\eta } \log^2\tp{\frac{nD}{\delta}} \log q}.	
\end{align*}
The total cost of the subroutine is 
\begin{align*}
T_{\mathrm{DFS}} + 	T_{\mathrm{rej}} =O\tp{\abs{S}D^2k \tp{\frac{n}{\delta}}^{\eta} \log^2\tp{\frac{nD}{\delta}}\log q }.
\end{align*}

Finally, we use the following lemma to bound the probabilities of the bad events $\+B^{(1)}_t$ and $\+B^{(2)}_t$.
\begin{lemma}
\label{lemma-subroutine}
%If $\+P$ satisfies \Cref{condition-projection} with parameter $0 < \alpha \leq 1$ , 
Let $\Phi = (V,\Dom{Q},\Cons{C})$ be the input CSP formula and $\Proj{h}$ a projection scheme satisfying \Cref{condition-projection} with parameters $\alpha$ and $\beta$.
Let $D$ denote the maximum degree of the dependency graph of $\Phi$.
Let $p = \max_{c \in \Cons{C}}\prod_{v \in \vbl{c}}\frac{1}{\abs{Q_v}}$.
Let $0<\eta<1$ be a parameter.
Suppose $\log \frac{1}{p} \geq \frac{1}{1-\alpha}\log(20D^2)$ and $\log \frac{1}{p} \geq \frac{1}{\beta}\log\tp{\frac{40\mathrm{e} D^2}{\eta}}$.
The subroutine $\sample(\Phi, \Proj{h} ,\delta, y_{\Lambda}, S)$ in \Cref{alg-sample} with parameter $\eta$ satisfies that
for any $1\leq t \leq T+1$,
\begin{align*}
\Pr[]{\+B^{(1)}_t} \leq \delta \quad\text{and}\quad 	\Pr[]{\+B^{(2)}_t} \leq \delta.
\end{align*}
\end{lemma}
The rest of this section is dedicated to the proof of \Cref{lemma-subroutine}.
Let $v_i \in V$ denote the random variable picked by \Cref{alg-mcmc} in the $i$-th iteration of the for-loop.
In the proof of \Cref{lemma-subroutine}, we always fix a  $1\leq t \leq T+1$ and a sequence $v_1,v_2,\ldots,v_T$. Hence, we always consider the probability space generated by \Cref{alg-mcmc} conditional on $v_i$ is picked in the $i$-th iteration of the for-loop.

Define (possibly partial) projected configuration
\begin{align}
\label{eq-def-X-proof}
Y = y_\Lambda \triangleq \begin{cases}
Y_{t-1}(V \setminus \{v_t\} ) &\text{if } 1 \leq t \leq T;\\
Y_{T} &\text{if } t = T+1,	
 \end{cases}
\end{align}
where $\Lambda = V \setminus \{v_t\}$ if $1\leq t \leq T$, and $\Lambda = V$ if $t = T + 1$.
We analyze $\sample(\Phi, \Proj{h} ,\delta, Y, S)$, where 
\begin{align*}
S = \begin{cases}
\{v_t\} &\text{if } 1\leq t \leq T;\\
V &\text{if } t = T+1. 	
 \end{cases}
\end{align*}

\subsection{Analysis of rejection sampling (bound $\mathrm{Pr}[\+B^{(1)}_t]$)}
We first prove that 
\begin{align}
\label{eq-bound-b2}
\Pr[]{\+B^{(1)}_t}\leq \delta.
\end{align}
Let $\Phi' = (V, \Dom{Q}, \Cons{C}')$ denote the CSP formula obtained from $\Phi=(V,\Dom{Q},\Cons{C})$ by removing constraints satisfied by $Y$. 
%Let $H = H_{\Phi} = (V, \+E)$ denote the hypergraph modeling $\Phi$, where $\+E = \{\vbl{c} \mid c\in \Cons{C}\}$ is a multi-set. 
Let $H' = H_{\Phi'} = (V,\+E')$ denote the hypergraph modeling $\Phi'$, where $\+E' = \{\vbl{c} \mid c \in \Cons{C}'\}$ is a multi-set.
Suppose $H'_\Phi$ has $\ell$ connected components $H_1',H_2',\ldots,H_\ell'$ that intersect with $S$, where $H_i' = (V_i,\+E_i')$ and $V_i \cap S \neq \emptyset$ for all $1\leq i \leq \ell$. 
Let $\Phi'_i =(V_i, \Dom{Q}_{V_i}, \Cons{C}'_i)$ denote the CSP formula represented by $H'_i$, where $\Cons{C}'_i$ denotes the set of constraints represented by $\+E'_i$.
%Since $\+B^{(1)}_t$ occurs, it must hold that $\abs{C_i^X} = \abs{\+E_i^X} \leq dk\log \frac{n}{\delta}$.

Fix an integer $1\leq i \leq \ell$. Lines~\ref{line-v-start} -- \ref{line-v-end} in \Cref{alg-sample} actually run rejection sampling on  $\widetilde{\Phi}_i = (V_i, \widetilde{\Dom{Q}}_{V_i}, \Cons{C}_i' )$, where each $\widetilde{Q}_v \subseteq Q_v$, such that
\begin{align*}
\forall v \in V_i, \quad
\widetilde{Q}_v \triangleq \begin{cases}
 h^{-1}_v(Y_v) &\text{if } v \in V_i \cap \Lambda;\\
 Q_v &\text{if } v \in V_i \setminus \Lambda.	
 \end{cases}
 \end{align*}
Since the maximum degree of the dependency graph of $\Phi$ is $D$, 
the maximum degree of the dependency graph of $\widetilde{\Phi}_i$ is at most $D$.
Let $\+D$ denote the product distribution such that each $v \in V_i$ samples a value from $\widetilde{Q}_v$ uniformly at random. 
For each constraint $c \in \Cons{C}_i'$, let $B_c$ denote the bad event that $c$ is not satisfied. 
Note that $\Proj{h}$ is a balanced projection scheme. By the definition of $\widetilde{\Dom{Q}}_{V_i}$, it holds that $|\widetilde{Q}_v| \geq \lfloor q_v/s_v \rfloor$ for all $v \in V_i$, where $q_v = \abs{Q_v}$. In other words, $\widetilde{\Phi}_i$ is the conditional LLL instance in \Cref{conditional-LLL-condition}.
By \Cref{condition-projection},
we have for each $c \in \Cons{C}'_i$,
\begin{align*}
\Pr[\+D]{B_c} = \prod_{v \in \vbl{c}}\frac{1}{\abs{\widetilde{Q}_v}} \leq \prod_{v \in \vbl{c}}\frac{1}{\lfloor q_v/s_v \rfloor} \leq \tp{\prod_{v \in \vbl{c}}\frac{1}{q_v}}^\beta,
\end{align*}
Recall that in \Cref{lemma-subroutine}, we assume that for each $c \in \Cons{C}$, $\sum_{v \in \vbl{c}} \log q_v \geq \frac{1}{\beta} \log\tp{\frac{40\mathrm{e} D^2}{\eta}}$ for $0 < \eta < 1$.  
Note that $\Cons{C}'_i \subseteq \Cons{C}$,
we have for each $c \in \Cons{C}'_i$,
\begin{align*}
\Pr[\+D]{B_c} \leq 	\frac{\eta}{40 \mathrm{e}  D^2}.
\end{align*}
For each $B_c$, define $x(B_c) = \frac{\eta}{40 D^2}$.  We have
\begin{align*}
\Pr[\+D]{B_c} &\leq  \frac{\eta}{40 \mathrm{e}  D^2} \leq 	\frac{\eta}{ 40 D^2} \tp{ 1 - \frac{\eta}{40 D^2} }^{\frac{40 D^2}{\eta} - 1} \leq \frac{\eta}{40 D^2} \tp{ 1 - \frac{\eta}{40D^2} }^{D} \\
&\leq x(B_c)\prod_{B_{c'} \in \Gamma(B_c)}\tp{1 - x(B_{c'})},
\end{align*}
where $\Gamma(\cdot)$ is defined as in the Lov\'asz local lemma (\Cref{theorem-LLL}).
Since $B^{(1)}_t$ occurs, it must hold that $\abs{\Cons{C}_i'}  \leq 2D\log \frac{nD}{\delta}$.
By Lov\'asz local lemma (\Cref{theorem-LLL}), we have
\begin{align*}
\Pr[\+D]{\bigwedge_{c \in \Cons{C}'_i}\overline{B_c} }&\geq \prod_{c \in \Cons{C}'_i}(1 - x(B_c)) \geq \prod_{c \in \Cons{C}'_i}\tp{1 - \frac{\eta}{40 D^2}}\\
\left(\text{by } \abs{\Cons{C}'_i}\leq 2D\log \frac{nD}{\delta}\right)\qquad &\geq \tp{1 - \frac{\eta}{40 D^2}}^{2D \log \frac{nD}{\delta}} \geq \exp\tp{-\frac{\eta}{5 D}\log \frac{Dn}{\delta}}\\
&=\tp{\frac{\delta}{Dn}}^{\frac{\eta}{5 D \ln 2}} \geq \tp{\frac{\delta}{Dn}}^{\frac{\eta}{2D}} \geq \frac{1}{2}\tp{\frac{\delta}{n}}^{\eta}. 
\end{align*}
Hence, each trial of the rejection sampling in Lines~\ref{line-v-start} -- \ref{line-v-end} succeeds with probability at least $\frac{1}{2}\tp{\frac{\delta}{n}}^{\eta}$.
Since the algorithm uses rejection sampling independently for $R =  \ctp{10\tp{\frac{n}{\delta}}^{\eta } \log \frac{n}{\delta}}$ times, the probability that the rejection sampling fails in one connected component is at most
\begin{align*}
\tp{1 - \frac{1}{2}\tp{\frac{\delta}{n}}^{\eta}}^R \leq \exp \tp{- \frac{R}{2} \tp{\frac{\delta}{n}}^{\eta } }	 \leq \frac{\delta}{n}.
\end{align*}
Since there are at most $n$ connected components, by a union bound,
\begin{align*}
\Pr[]{\+B^{(1)}_t}\leq \delta	
\end{align*}
This proves~\eqref{eq-bound-b2}.

\subsection{Analysis of connected component (bound $\mathrm{Pr}[\+B^{(2)}_t]$)}
We now bound the probability of bad event $\+B^{(2)}_t$. 
Consider the subroutine $\sample(\Phi, \Proj{h} ,\delta, Y, S)$.
Recall $\Phi' = (V, \Dom{Q}, \Cons{C}')$ is the CSP formula obtained from $\Phi=(V,\Dom{Q},\Cons{C})$ by removing all the constraints satisfied by $Y$.
Recall hypergraph $H' = H_{\Phi'}=(V,\+E')$ models $\Phi'$.
Let $H = H_{\Phi} = (V, \+E)$ denote the hypergraph modeling $\Phi$, where $\+E = \{\vbl{c} \mid c\in \Cons{C}\}$ is a multi-set. 
For any edge $e \in \+E$, we use $\+B_e$ to denote the bad event that $e \in \+E'$ and the number of hyperedges in the connected component in $H'$ that contains $e$ is at least $L$, where  $L = \lceil 2D \log \frac{Dn}{\delta} \rceil$.
By a union bound, we have
\begin{align*}
\Pr[]{\+B_t^{(2)}} \leq \sum_{e \in \+E}\Pr{\+B_e}.	
\end{align*}
Recall $D$ is the maximum degree of the dependency graph. 
Since $\abs{\+E} \leq n(D+1)$, it suffices to prove
\begin{align}
\label{eq-bound-Be-general}
\Pr{\+B_e}\leq \frac{\delta}{n(D+1)}.	
\end{align}

To bound the probability of $\+B_e$, we need the following lemma.
%Let $v_i$ denote the variable picked in the $i$-th step of Glauber dynamics.
%In all the analysis in this section, we always fix a sequence of variables $v_1,v_2,\ldots,v_T$.
\begin{lemma}
\label{lemma-uniform}
Let $\Phi = (V,\Dom{Q},\Cons{C})$ be a CSP formula.
Let $\Proj{h}$ be the projection scheme satisfying \Cref{condition-projection} with parameters $\alpha$ and $\beta$. 
Let $q_v = \abs{Q_v}$ and $D$ denote the maximum degree of the dependency graph of $\Phi$.
If for any constraint $c \in \Cons{C}$,
\begin{align*}
\sum_{v \in \vbl{c}}\log q_v \geq \frac{1}{\beta}\log(40 \mathrm{e} D^2),
\end{align*}
then for any subset $H \subseteq \Lambda$, any projected configuration $\sigma \in \Sigma_H = \bigotimes_{v \in H}\Sigma_v$,
\begin{align*}
\Pr{Y_H = \sigma} \leq \exp\tp{\sum_{u \in H}\frac{1}{20 D}} \prod_{v \in H} \tp{\frac{1}{q_v}\ctp{\frac{q_v}{s_v}}},
\end{align*}
where $Y \in \Sigma_{\Lambda}$ is defined in~\eqref{eq-def-X-proof}.
%Fix a sequence of variables $v_1, v_2,\ldots,v_t \in V$ such that $v_i$ is the variable picked by Glauber dynamics in $i$-th step. 
\end{lemma}
%\begin{lemma}
%\label{lemma-uniform}
%Suppose $\Phi=(V,C,[q])$ is a multi-valued $k$-CNF formula, where each variable belongs to at most $d$ clauses. If $\+P$ satisfies \Cref{condition-projection} with parameter $0 < \alpha \leq 1$, then 
%then for any subset $H \subseteq V$, any image $\sigma \in \prod_{v \in H}[m_v]$,
%\begin{align*}
%\Pr{X_H = \sigma} \leq \tp{\frac{1}{q}}^{\abs{H}} \tp{1 + \frac{\alpha}{25 k}}^{\abs{H}} \prod_{v \in H}\ctp{\frac{q}{m_v}}	.
%\end{align*}
%
%Suppose $\Phi=(V,C,(Q_v)_{v \in V})$ is a CSP formula, where the maximum degree of the dependency graph is $D$. 
%Let $q_v = \abs{Q_v}$.
%If $\+P$ satisfies \Cref{condition-projection-general} with parameter  $0<\alpha \leq 1$,
%then for any subset $H \subseteq V$, any image $\sigma \in \prod_{v \in H}[m_v]$,
%\begin{align*}
%\Pr{X_H = \sigma} \leq \tp{1 + \frac{\alpha}{60 D^{50}}}^{\abs{H}} \prod_{v \in H} \tp{\frac{1}{q_v}} \ctp{\frac{q_v}{m_v}}.
%\end{align*}
%%Fix a sequence of variables $v_1, v_2,\ldots,v_t \in V$ such that $v_i$ is the variable picked by Glauber dynamics in $i$-th step. 
%\end{lemma}
The proof of \Cref{lemma-uniform} is deferred to \Cref{section-proof-uniformity}.
Next, we introduce the following definitions of line graph and 2-tree. 

\begin{definition}[line graph]
Let $H=(V,\+E)$ be a hypergraph. The line graph $\Lin(H)$ is a graph such that each vertex represents a hyperedge in $\+E$, two vertices $e,e' \in \+E$ are adjacent iff $e \cap e' \neq \emptyset$.
\end{definition}

\begin{definition}[2-tree]
Let $G = (V,E)$ be a graph. A subset of vertices $\twotree \subseteq V$ is a 2-tree if (1) for any $u,v \in \twotree$, their distance $\dist_G(u,v)$ in graph $G$ is at least 2; (2) if one adds an edge between $u,v \in \twotree$ such that $\dist_G(u,v)=2$, then $\twotree$ becomes connected.
\end{definition}

The following two propositions are proved in the full version~\cite{FGYZ20full} of~\cite{FGYZ20}.

\begin{proposition}[\text{\cite[Corollary~5.7]{FGYZ20full}}]
\label{proposition-number}
Let $G=(V,E)$ be a graph with maximum degree $\Delta$ and $v \in V$ a vertex. The number of 2-trees in graph $G$ of size $\ell$ containing vertex $v$ is at most $\frac{(\mathrm{e}\Delta^2)^{\ell - 1}}{2}$.
\end{proposition}

\begin{proposition}[\text{\cite[Lemma~5.8]{FGYZ20full}}]
\label{proposition-2-tree}
Let $H = (V, \+E)$ be  hypergraph. Let $\Lin(H)$ denote the line graph of $H$.
Let $B \subseteq \+E$ be a subset of hyperedges that induces a connected subgraph in $\Lin(H)$ and $e \in B$ an arbitrary hyperedge. 
There exists a 2-tree $\twotree \subseteq \+E$ in $\Lin(H)$ such that $e \in \twotree$ and $\abs{\twotree}= \left\lfloor \frac{\abs{B}}{D+1} \right\rfloor$, where $D$ is the maximum degree of the line graph $\Lin(H)$.
\end{proposition}

Suppose $\Proj{h}$ satisfies \Cref{condition-projection}.
Recall $Y \in \Sigma_{\Lambda}$, where $\Lambda = V \setminus \{v_t\}$ for $1\leq t \leq T$ and $\Lambda = V$ for $t = T+1$.
We say an edge $e \in \+E$ is bad if $e$ is not satisfied by $Y$. Suppose $e$ represents the constraint $c$ such that $c(\Ass{x}) = \False$ for a unique configuration $\Ass{x} \in \Dom{Q}_{e}$. 
Given the projected configuration $Y \in \Sigma_{\Lambda}$, we have
\begin{align}
\label{eq-proof-bad}
e \text{ is bad} \quad \Longleftrightarrow \quad \forall u \in \Lambda \cap e, Y_u \neq  h_u(\Ass{x}_u).
\end{align}
In other words, if $e$ is bad, then the constraint corresponding to $c$ in the ``round-down'' CSP formula (\Cref{def:round-down}) is not satisfied by $Y$.
If $\+B_e$ occurs, there must exist a connected component $B \subseteq \+E$ in line graph $\Lin(H)$ such that $e \in B$ and all hyperedges in $B$ are bad and $\abs{B} = L$, where $L = \lceil 2D \log \frac{Dn}{\delta} \rceil$ and $D$ is the maximum degree of the dependency graph of the input formula.   By \Cref{proposition-2-tree}, there must exist a 2-tree $\twotree$ in $\Lin(H)$ with size $\ell = \left\lfloor \frac{L}{D+1} \right\rfloor$ such that $e \in \twotree$ and all edges in $\twotree$ are bad.
Fix such a 2-tree $\twotree$. 
By definition, each vertex in $\twotree$ is a hyperedge $e \in \+E$, and for all $e,e' \in \twotree$, $e \cap e' = \emptyset$.
Let $\twotree' \subseteq \twotree$ denote the subset of edges $e \in \twotree$ such that $e \subseteq \Lambda$. 
Since $Y$ is a random projected configuration, 
by~\eqref{eq-proof-bad}, we have
\begin{align*}
\Pr[]{\forall e \in \twotree, e \text{ is bad}}
&=\Pr[]{\forall e \in \twotree, \forall u \in e \cap \Lambda, Y_u \neq h_u(\Ass{x}_u) }\\
&\leq \Pr[]{\forall e \in \twotree', \forall u \in e, Y_u \neq h_u(\Ass{x}_u)}.
\end{align*}
Fix an edge $e \in \twotree'$. By \Cref{condition-projection} and the condition $\sum_{v \in e} \log q_v \geq \frac{1}{1-\alpha}\log(20D^2)$ assumed in \Cref{lemma-subroutine}, it holds that
\begin{align*}
\prod_{v \in e}	\frac{1}{q_v} \ctp{\frac{q_v}{s_v}} \leq \tp{\prod_{v \in e}\frac{1}{q_v}}^{1-\alpha} \leq \frac{1}{20D ^2}.
\end{align*}
Note that if $s_v = 1$, then $	\frac{1}{q_v} \ctp{\frac{q_v}{s_v}} = 1$. For any $v \in e$ such that $s_v > 1$(thus $q_v \geq s_v > 1$), we have $	\frac{1}{q_v} \ctp{\frac{q_v}{s_v}} \leq 	\frac{1}{q_v} \ctp{\frac{q_v}{2}} \leq \frac{2}{3}$. Let $r = \log_{2/3}\frac{1}{20D^2} + 1$. We can find a subset of variables $R(e) \subseteq e$ such that
\begin{align*}
\prod_{v \in R(e)}	\frac{1}{q_v} \ctp{\frac{q_v}{s_v}} \leq \frac{1}{20 D^2}, \quad \text{and}\quad \abs{R(e)} \leq r. 
\end{align*}
Note that \Cref{lemma-subroutine} assumes that $\sum_{v \in \vbl{c}}\log q_v \geq \frac{1}{\beta}\log\tp{\frac{40\mathrm{e} D^2}{\eta}} \geq  \frac{1}{\beta}\log(40\mathrm{e} D^2)$.
We use \Cref{lemma-uniform} on subset $H = \cup_{e \in \twotree'}R(e)$. 
Note that all hyperedges in $\twotree'$ are disjoint.
We have
\begin{align*}
\Pr[]{\forall e \in \twotree, e \text{ is bad}} &\leq \Pr[]{\forall e \in \twotree', \forall u \in R(e), Y_u \neq h_u(\Ass{x}_u) } \leq \Pr[]{\forall u \in H, Y_u \neq h_u(\Ass{x}_u) }\\
&\leq \prod_{e \in \twotree'}\prod_{v \in R(e)}\tp{\frac{1}{q_v}\ctp{\frac{q_v}{s_v}}\exp\tp{\frac{1}{20D}} }\leq  \prod_{e \in \twotree'}\tp{\frac{1}{20D^2}\exp\tp{\frac{r}{20D}}}\\
\tp{\text{by $r = \log_{2/3}\frac{1}{20D^2} + 1$)}}\quad&\leq \prod_{e \in \twotree'}\tp{\frac{1}{12D^{2}}}.
\end{align*}
Since $\abs{\Lambda} \geq n -1$ and all hyperedges in $\twotree$ are disjoint, $\abs{\twotree'} \geq \abs{\twotree} - 1 = \ell - 1$. We have
\begin{align*}
\Pr[]{\forall e \in \twotree, e \text{ is bad}}\leq 	\tp{\frac{1}{12D^{2}}}^{\ell - 1}. 
\end{align*}
Note that the maximum degree of line graph is at most $D$. By \Cref{proposition-number}, we have
\begin{align*}
\Pr[]{\+B_e} &\leq \frac{1}{2}\tp{\mathrm{e}D^2}^{\ell - 1}	\tp{\frac{1}{12D^{2}}}^{\ell - 1} \leq \frac{1}{2}\tp{\frac{1}{4}}^{\ell - 1} \leq  \tp{\frac{1}{2}}^{2\ell -1 }.
\end{align*}
Note that $\ell = \ftp{L / (D+1)}$ and $L = \ctp{2D \log \frac{nD}{\delta}}$. We have $\ell \geq \log \frac{nD}{\delta}-1$. 
We may assume $nD \geq 16$. Otherwise, the sampling problem is trivial.
The inequality~\eqref{eq-bound-Be-general} can be proved by
\begin{align*}
\Pr[]{\+B_e} \leq 	\tp{\frac{1}{2}}^{2\log\frac{nD}{\delta}-3} \leq \frac{\delta}{n(D+1)}.
\end{align*}

\subsection{Proof of \Cref{lemma-uniform}}
\label{section-proof-uniformity}
We now prove~(\Cref{lemma-uniform}). We use the following lemma to prove it. 
\begin{lemma}
\label{lemma-uniform-general}
%Let $\eta \geq 1$ be a positive parameter.
Let $\Phi = (V,\Dom{Q},\Cons{C})$ be a CSP formula.
Let $\Proj{h} = (h_v)_{v \in V}$ be the projection scheme satisfying \Cref{condition-projection} with parameters $\alpha$ and $\beta$.
Let $D$ denote the maximum degree of the dependency graph of $\Phi$.
Let $q_v = \abs{Q_v}$. Suppose for any constraint $c \in \Cons{C}$, it holds that
\begin{align*}
\sum_{v \in \vbl{c}}\log q_v \geq \frac{1}{\beta} \log(40 \mathrm{e} D^2).
\end{align*}
Fix a variable $u \in V$ and a partial projected configuration $\tau \in \Sigma_{V \setminus \{u\}}$. 
For any $y \in \Sigma_u$, it holds that
\begin{align*}
\nu_{u}^\tau(y) \leq \frac{1}{q_u}\ctp{\frac{q_u}{s_u}}\exp\tp{\frac{1}{20D}}.
\end{align*}
\end{lemma}

\begin{proof}
Define a new CSP formula $\widehat{\Phi} = (V,  \widehat{\Dom{Q}} = (\widehat{Q}_v)_{v \in V}, \Cons{C})$ by
\begin{align*}
\forall w \in V, \quad \widehat{Q}_w = \begin{cases}
h^{-1}_w(\tau_w) &\text{if } w \neq u \\
Q_w 	&\text{if }  w = u.
\end{cases}
\end{align*}
Let $\+D$ denote the product distribution that each $w \in V$ takes a value from $\widehat{Q}_w$ uniformly and independently. For each constraint $c \in \Cons{C}$, define a bad event $B_c$ as $c$ is not satisfied. 
%Let $\Gamma(B_c)$ denote the set of bad event $B_{c'}$ such that $c \neq c'$ and $\vbl{B_c} \cap \vbl{B_{c'}} \neq \emptyset$.  
Let $\+B = (B_c)_{c \in \Cons{C}}$ be the collection of bad events.
Recall that $\Gamma(\cdot)$ is defined as in the Lov\'asz local lemma (\Cref{theorem-LLL}).
It holds that $\max_{c \in \Cons{C}}\abs{\Gamma(B_c)} \leq D$.
For each $B_c$, let $x(B_c) = \frac{1}{40D^2}$. By \Cref{condition-projection}, it holds that
\begin{align*}
\Pr[\+D]{B_c \text{ is not satisfied}} &= \prod_{v \in \vbl{c}}\frac{1}{\abs{\widehat{Q}_v}} \leq \prod_{v \in \vbl{c}}\frac{1}{\ftp{q_v/s_v}}\leq \tp{\prod_{v \in \vbl{c}}\frac{1}{q_v}}^{\beta}\\
&\leq \frac{1}{40\mathrm{e}D^2}\leq \frac{1}{40D^2}\tp{1-\frac{1}{40D^2}}^{40D^2 - 1}\\
& \leq \frac{1}{40D^2}\tp{1-\frac{1}{40D^2}}^{D} \leq x(B_c) \prod_{B_{c'} \in \Gamma(B_c)}\tp{1 - x(B_{c'})}.
\end{align*}
Fix $y \in \Sigma_u$.
Let $A$ denote the event that the value of $u$ belongs to $h^{-1}_u(y)$, then $\abs{\Gamma(A)} \leq D$, where $\Gamma(A) \subseteq \+B$ is the set of bad events $B$ such that $u \in \vbl{B}$.
%We always assume $k,d \geq 2$.
Let $\widehat{\mu}$ denote the uniform distribution of all satisfying assignments to $\widehat{\Phi}$. 
By \Cref{theorem-LLL}, we have
\begin{align*}
\nu_{u}^\tau(y) = \Pr[\widehat{\mu}]{A} = \Pr[X \sim \widehat{\mu} ]{X_u \in h^{-1}_u(y)}\leq \frac{1}{q_u}\ctp{\frac{q_u}{s_u}}\tp{1 - \frac{1}{40D^2}}^{-D}  \leq \frac{1}{q_u}\ctp{\frac{q_u}{s_u}}\exp\tp{\frac{1}{20D}}. &\qedhere
\end{align*}
\end{proof}

Now we are ready to prove \Cref{lemma-uniform}. %If $\+P$ satisfies \Cref{condition-projection-general}, we use \Cref{lemma-uniform-general} to prove \Cref{lemma-uniform};
%If $\+P$ satisfies \Cref{condition-projection}, we use \Cref{lemma-uniform-kq} to prove \Cref{lemma-uniform}. 

\begin{proof}[Proof of \Cref{lemma-uniform}]
Fix a subset $H \subseteq V$, and an projected configuration $\sigma \in \Sigma_H$.
Recall $1\leq t \leq T + 1$ is a fixed integer.
Recall $Y = Y_{t-1}(\Lambda)$, where $\Lambda = V \setminus \{v_t\}$ if $1\leq t \leq T$, and $\Lambda = V$ if $t = T + 1$.
Recall that  $v_1, v_2,\ldots,v_t \in V$ is a sequence such that $v_i$ is the variable picked by \Cref{alg-mcmc} in $i$-th iteration of the for-loop.

For any variable $u \in H$, let $t(u)$ denote the last step up to step $t$ such that $u$ is picked by \Cref{alg-mcmc} of the for-loop. Formally, if $u$ appears in the sequence $v_1,v_2,\ldots,v_t$, then $t(u)$ is the largest number such that $v_{t(u)} = u$; if $u$ does not appear in the sequence $v_1,v_2,\ldots,v_t$, then $t(u) = 0$. We list all variables in $H$ as $u_1,u_2,\ldots,u_{\abs{H}}$ such that $t(u_1) \leq t(u_2)\leq \ldots \leq t(u_{\abs{H}})$, where for these variables $u$ satisfying $t(u) = 0$, we break tie arbitrarily. Thus, $Y_t(u) = Y_{t(u)}(u)$ for all $u \in H$. We have
\begin{align*}
\Pr{Y_H = \sigma}  = \Pr[]{\forall u_i \in H, Y_{u_i} = \sigma_{u_i}} \leq \prod_{i = 1}^{|H|} \Pr[]{Y_{t(u_i)}(u_i) = \sigma_{u_i} \mid \forall j <i, Y_{t(u_j)}(u_j) = \sigma_{u_j}}.
\end{align*}
We now only need to prove that, for any $1\leq i \leq |H|$,
\begin{align}
\label{eq-uniform-target-2}
\Pr[]{Y_{t(u_i)}(u_i) = \sigma_{u_i} \mid \forall j <i, Y_{t(u_j)}(u_j) = \sigma_{u_j}} \leq \frac{1}{q_{u_i}}\ctp{\frac{q_{u_i}}{s_{u_i}}} \exp\tp{\frac{1}{20 D}}.
\end{align}

Suppose $t(u_i) = 0$, then $Y_{0}(u_i) \in \Sigma_{u_i}$ is sampled independently with $\Pr[]{Y_{0}(u_i)=\sigma_{u_i}} = \frac{\abs{h_{u_i}^{-1}(\sigma_{u_i})}}{q_{u_i}}$. Since $\Proj{h}$ is balanced, we have $\abs{h^{-1}_{u_i}(\sigma_{u_i})} \leq \ctp{\frac{q_{u_i}}{s_{u_i}}}$. Inequality \eqref{eq-uniform-target-2} holds because
\begin{align*}
\Pr[]{Y_{0}(u_i) = \sigma_{u_i} \mid \forall j <i, Y_{0}(u_j) = \sigma_{u_j}} \leq \frac{1}{q_{u_i}}\ctp{\frac{q_{u_i}}{s_{u_i}}}.
\end{align*}

Suppose $t(u_i)  = \ell \neq 0$. \Cref{alg-mcmc} uses the subroutine $\sample(\cdot)$ to sample a random $X_v \in Q_v$ in \Cref{line-sample-1}, then maps $X_v$ into $Y_{\ell}(u_i)$ in~\Cref{line-map}.
If $X_v$ is returned in \Cref{line-bad-return-1} or \Cref{line-bad-return-2} in \Cref{alg-sample}, then $X_v$ is uniformly distribution over $Q_{u_i}$. In this case, inequality \eqref{eq-uniform-target-2} holds because
\begin{align*}
\Pr[]{Y_{\ell}(u_i) = \sigma_{u_i} \mid \forall j <i, Y_{t(u_j)}(u_j) = \sigma_{u_j}} = \sum_{X_v \in h^{-1}_{u_i}(\sigma_{u_i})}\frac{1}{q_{u_i}} \leq  \frac{1}{q_{u_i}} \ctp{\frac{q_{u_i}}{s_{u_i}}}.
\end{align*}
Otherwise, $X_v$ is returned in \Cref{line-good} of \Cref{alg-sample}. In this case, $Y_{\ell}(u_i)$ is sampled from the distribution $\nu_{u_i}^{Y_{\ell-1}(V \setminus \{u_i\} )}$. 
We can use \Cref{lemma-uniform-general} with $\tau=Y_{\ell-1}(V \setminus \{u_i\} )$ and $u = u_i$. Note that \Cref{lemma-uniform-general} holds for any $\tau$ and $u$. We have
%Given any $\tau$, we can define a CSP formula $\widehat{\Phi} = (V,C,(\widehat{Q}_v)_{v \in V})$ such that $\widehat{Q}_v = P_v(\tau_v)$ for $v \neq u_i$ and $\widehat{Q}_{u_i} = Q_{u_i}$. Then, the conditional distribution $\mu^\tau_{u_i}$ is precisely the distribution $\widehat{\mu}_{u_i}$, where $\widehat{\mu}$ is the uniform distribution of all satisfying assignments for $\widehat{\Phi}$.
%Suppose $\+P$ satisfies \Cref{condition-projection}, then $\prod_{v \in \vbl{c}} \lfloor{q/m_v}\rfloor \geq \frac{3000q^2d^6k^6}{\alpha}$ for any clause $c$, $\widehat{\Phi}$ satisfies 
\begin{align*}
\Pr[]{Y_{\ell}(u_i) = \sigma_{u_i} \mid \forall j <i, Y_{t(u_j)}(u_j) = \sigma_{u_j}} =
\nu_{u_i}^{Y_{\ell-1}(V \setminus \{u_i\} )}(\sigma_{u_i}) \leq \frac{1}{q_{u_i}}\ctp{\frac{q_{u_i}}{s_{u_i}}} \exp\tp{\frac{1}{20 D}}.
\end{align*}
Thus, inequality~\eqref{eq-uniform-target-2} holds.
\end{proof}

\section{Proof of rapid mixing}
\label{section-mixing}
Let $\Phi=(V,\Dom{Q},\Cons{C})$ be a CSP formula with atomic constraints and $\Proj{h}=(h_v)_{v \in V}$ be a balanced projection scheme satisfying \Cref{condition-projection} with parameter $\alpha$ and $\beta$, where $h_v: Q_v \to \Sigma_v$.
Let $\nu=\nu_{\Phi,\Proj{h}}$ be the projected distribution over $\Sigma = \bigotimes_{v \in V}\Sigma_v$ in \Cref{def:proj-distr}.
Let $(Y_t)_{t \geq 0}$ denote the Glauber dynamics $P_{\mathrm{Glauber}}$ on $\nu$. 
In this section, we show that the Glauber dynamics $P_{\mathrm{Glauber}}$ is rapid mixing,
and prove \Cref{lemma-mixing-kd} and \Cref{lemma-mixing-gen}.

%The rest of this section is dedicated to the proof of \Cref{lemma-mixing}.

\subsection{The stationary distribution}
We first proves that $\nu$  is the unique stationary distribution.
\begin{proposition}
\label{proposition-stationary}
Let $\Phi=(V,\Dom{Q},\Cons{C})$ be a CSP formula with atomic constraints. Let $\Proj{h}=(h_v)_{v \in V}$ be the projection scheme satisfying \Cref{condition-projection} with parameters $\alpha$ and $\beta$.
Let  $q_v = \abs{Q_v}$, $p = \max_{c \in \Cons{C}}\prod_{v \in \vbl{c}} \frac{1}{q_v}$ and $D$ denote the maximum degree of the dependency graph of $\Phi$.
Suppose $ \log\frac{1}{p}  \geq \frac{1}{\beta}\log(2\mathrm{e}D)$.
The Glauber dynamics $P_{\mathrm{Glauber}}$ is irreducible, aperiodic and reversible with respect to $\nu$, thus it has the unique stationary distribution $\nu$.	
\end{proposition}
\begin{proof}
By the transition rule of Glauber dynamics, it is easy to verify the Glauber dynamics is 	aperiodic and reversible with respect to $\nu$. We prove the Markov chain is irreducible.
We show that for any $\sigma \in \Sigma$, $\nu(\sigma) > 0$.
This implies that the transition probability of Glauber dynamics is always well-defined and the Markov chain is connected.
Fix a $\sigma \in \Sigma$. Define a new instance $\widehat{\Phi}=(V, \widehat{\Dom{Q}}=(\widehat{Q}_v)_{v \in V},\Cons{C})$ as $\widehat{Q}_v = h^{-1}_v(\sigma_v)$ for all $v \in V$. It suffices to show that $\widehat{\Phi}$ is satisfiable, which implies $\nu(\sigma) > 0$. The maximum degree of dependency graph of $\widehat{\Phi}$ is at most $D$. Besides, if each variable picks a value from $\widehat{Q}_v$ uniformly and independently, then for each $c \in \Cons{C}$, the probability that $c$ is not satisfied is at most
\begin{align*}
\prod_{v \in \vbl{c}} \frac{1}{\vert \widehat{Q}_v \vert}	\leq \prod_{v \in \vbl{c}} \frac{1}{\lfloor q_v /s_v \rfloor} \leq \tp{\prod_{v \in \vbl{c}}\frac{1}{q_v} }^\beta \leq \frac{1}{2\mathrm{e}D}.
\end{align*}
By Lov\'asz local lemma, $\widehat{\Phi}$ is satisfiable.
\end{proof}

\subsection{Path coupling analysis}
We use the path coupling~\cite{bubley1997path} to show that the Markov chain is rapid mixing.
Fix two projected configurations $\Ass{X}, \Ass{Y} \in \Sigma = \bigotimes_{v \in V}\Sigma_v$ such that $\Ass{X}$ and $\Ass{Y}$ disagree only at one variable $v_0 \in V$ (assume $s_{v_0} \geq 2$). We construct a coupling $(\Ass{X},\Ass{Y})\rightarrow (\Ass{X}',\Ass{Y}')$ such that $\Ass{X} \rightarrow \Ass{X}'$ and $\Ass{Y} \rightarrow \Ass{Y}'$ each individually follows the transition rule of $P_{\mathrm{Glauber}}$ such that
\begin{align}
\label{eq-path-coupling-target}
\E{d_{\mathrm{ham}}(\Ass{X}',\Ass{Y}')\mid \Ass{X}, \Ass{Y}} \leq 1-\frac{1}{2n},
\end{align}
where $d_{\mathrm{ham}}(\Ass{X}',\Ass{Y}') \triangleq |\{v \in V \mid X'_v \neq Y'_v\}|$ denotes the Hamming distance between $\Ass{X}'$ and $\Ass{Y}'$.
Note that the Hamming distance is at most $n$.
Thus, by path coupling lemma (\Cref{lemma-path-coupling}), for any $0<\epsilon < 1$,
\begin{align*}
\tmix(\epsilon) \leq \ctp{2n \log \frac{n}{\epsilon}},
\end{align*}
where $n = |V|$ is the number of variables. 

The coupling $(\Ass{X},\Ass{Y})\rightarrow (\Ass{X}',\Ass{Y}')$ is constructed as follows.
\begin{itemize}
\item Pick the same variable $v \in V$ uniformly at random, set $X'_u \gets X_u$ and $Y'_u \gets Y_u$ for all $u \neq v$.
\item Sample $(X'_v,Y'_v)$ jointly from the optimal coupling between $\nu_{v}^{X_{V\setminus \{v\} }}$ and $\nu_{v}^{Y_{V\setminus \{v\}}}$.
\end{itemize}
By the linearity of expectation, we have
\begin{align*}
\E{d_{\mathrm{ham}}(\Ass{X}',\Ass{Y}')\mid \Ass{X}, \Ass{Y}} &= \sum_{v \in V}\Pr[]{X'_v \neq Y'_v \mid \Ass{X}, \Ass{Y}}\\
\text{(by the optimal coupling)}\qquad&= \frac{1}{n}\sum_{v \in V \setminus \{v_0\} } \DTV{\nu_{v}^{ X_{V\setminus \{v\} }}}{\nu_{v}^{Y_{V\setminus \{v\} }}} + \tp{1 - \frac{1}{n}},
%&= \frac{1}{n}\sum_{v \in V: v\neq v_0} \DTV{\mu_{v,\+P}^{X_{V\setminus \{v\} }}}{\mu_{v,\+P}^{Y_{V\setminus \{v\} }}},
\end{align*}
where the last equation holds because $\DTV{\nu_{v_0}^{X_{V\setminus \{v_0\} }}}{\nu_{v_0}^{Y_{V\setminus \{v_0\} }}} = 0$.
%Define 
%\begin{align*}
%I_v \triangleq \DTV{\mu_{v,\+P}^{X_{V\setminus \{v\} }}}{\mu_{v,\+P}^{Y_{V\setminus \{v\} }}}.	
%\end{align*}
To prove~\eqref{eq-path-coupling-target}, it suffices to prove
\begin{align*}
\sum_{v \in V  \setminus \{v_0\}} \DTV{\nu_{v}^{X_{V\setminus \{v\} }}}{\nu_{v}^{Y_{V\setminus \{v\} }}}\leq \frac{1}{2}.	
\end{align*}

To prove the above inequality, we need to bound $ \DTV{\nu_{v}^{X_{V\setminus \{v\} }}}{\nu_{v}^{Y_{V\setminus \{v\} }}}$ for each $v \in V \setminus \{v_0\}$.
We use the coupling introduced by Moitra~\cite{Moi19} to do this task. 
For $k$-uniform CSP formula such that the domain of each variable is $[q]$, we construct an adaptive version~\cite{guo2019counting} of Moitra's coupling.
Compared with the analysis in~\cite{guo2019counting,FGYZ20}, this coupling  is more refined and requires a more careful analysis. This part in given in \Cref{section-adaptive-coupling}. For general CSP formula, we use the original non-adaptive version of Moitra's coupling. The analysis for general case is much more involved, because we need to deal with arbitrary domain and arbitrary size of constraints. This part is given in \Cref{section-non-adp}.

\subsection{Adaptive coupling analysis}
\label{section-adaptive-coupling}
We first analyze the simple case. Suppose the original input CSP formula of \Cref{alg-mcmc} is a $(k,d)$-CSP formula $\Phi=(V,[q]^V,\Cons{C})$ with atomic constraints, where $\abs{\vbl{c}}=k$ for all $c\in\Cons{C}$ and each variable $v\in V$ appears in at most $d$ constraints, on homogeneous domains ${Q}_v=[q]$ for all $v\in V$.
 Note that this case covers two applications: hypergraph coloring and $k$-CNF formula.  We prove the following lemma.
\begin{lemma}
\label{lemma-path-coupling-uniform}
Let  $\Phi=(V,[q]^V,\Cons{C})$ be a $(k,d)$-CSP formula with atomic constraints.
Let $\Proj{h} = (h_v)_{v \in V}$ be the projection scheme for $\Phi$ satisfying \Cref{condition-projection} with parameters $\alpha$ and $\beta$.
%Let $q_v = \abs{Q_v}$, $p = \max_{c \in \Cons{C}}\prod_{v \in \vbl{c}}\frac{1}{q_v}$ and $D$ denote the maximum degree of the dependency graph of $\Phi$. 
If
\begin{align}
\label{eq-cond-path-coupling-uniform}
k\log q \geq \frac{1}{\beta} \log \tp{3000 q^2d^6k^6},	
\end{align}
then it holds that $\sum_{v \in V  \setminus \{v_0\}} \DTV{\nu_{v}^{X_{V\setminus \{v\} }}}{\nu_{v}^{Y_{V\setminus \{v\} }}}\leq \frac{1}{2}$.
\end{lemma}

Recall that for any $\sigma \in \Sigma_{\Lambda}$, where $\Lambda \subseteq V$, the distribution $\mu^\sigma$ is the distribution of $\Ass{X} \in [q]^V$ such that $\Ass{X}$ is sampled from $\mu$ conditional on $\Proj{h}(\Ass{X}_{\Lambda}) = (h_v(\Ass{X}_v))_{v \in \Lambda} = \sigma$, where $\mu$ is the uniform distribution over all satisfying assignments to $\Phi$. We use $\mu_{v}^\sigma$ to denote the marginal distribution on $v$ projected from $\mu^\sigma$.
For any $v \in V$ and $c \in \Sigma_v$, it holds that
\begin{align*}
\nu_{v}^{X_{V\setminus \{v\} }}(c) = \sum_{j \in h_v^{-1}(c)}\mu_v^{X_{V\setminus \{v\}}}(j) \quad\text{and}\quad \nu_{v}^{Y_{V\setminus \{v\} }}(c) = \sum_{j \in h_v^{-1}(c)}\mu_v^{Y_{V\setminus \{v\}}}(j).
\end{align*}
%where $\mu_v^{X_{V\setminus \{v\}}}$ (and $\mu_v^{Y_{V\setminus \{v\}}}$) is the marginal distribution on $v$ projected from $\mu^{X_{V\setminus \{v\}}}$ (and $\mu^{Y_{V\setminus \{v\}}}$).
Note that each $h_v$ is a function from $[q]$ to $\Sigma_v$.
By triangle inequality,
it holds that
\begin{align*}
\DTV{\nu_{v}^{X_{V\setminus \{v\} }}}{\nu_{v}^{Y_{V\setminus \{v\} }}} &= \frac{1}{2}\sum_{c \in \Sigma_v}\abs{\nu_{v}^{X_{V\setminus \{v\} }}(c) - \nu_{v}^{Y_{V\setminus \{v\} }}(c)}\\
\tp{\text{by } \biguplus_{c \in \Sigma_v}h^{-1}_v(c) = [q]} \quad &\leq \frac{1}{2}\sum_{j \in [q]}\abs{\mu_{v}^{X_{V\setminus \{v\} }}(j) - \mu_{v}^{Y_{V\setminus \{v\} }}(j)} = \DTV{\mu_{v}^{X_{V\setminus \{v\} }}}{\mu_{v}^{Y_{V\setminus \{v\} }}}.	
\end{align*}
For any variable $v \in V \setminus \{v_0\}$, define the influence on $v$ caused by $v_0$ as 
\begin{align}
\label{eq-proof-Iv}
I_v \triangleq 	\DTV{\mu_{v}^{X_{V\setminus \{v\} }}}{\mu_{v}^{Y_{V\setminus \{v\} }}}.
\end{align}
To prove the rapid mixing of Glauber dynamics, it suffices to prove that
\begin{align}
\label{eq-sum-Iv}
\sum_{v \in V: v \neq v_0}I_v \leq \frac{1}{2}.	
\end{align}

Fix a variable $\vst \in V$. We will use a coupling $\cadp$ to bound the influence $I_{\vst}$.
The coupling $\cadp$ draws two random samples $\Ass{X}^{\cadp} \sim \mu^{X_{V\setminus\{\vst\} }}$ and $\Ass{Y}^{\cadp} \sim \mu^{Y_{V\setminus \{\vst\} }} $. 
By coupling lemma (\Cref{lemma-coupling-ineq}), the influence $I_{\vst}$ can be bounded by
\begin{align}
\label{eq-bound-Iv}
I_{\vst} \leq \Pr[\cadp]{X^{\cadp}_{\vst} \neq Y^{\cadp}_{\vst}}.	
\end{align}

To describe the coupling $\cadp$,
we first introduce some definitions. 
Recall $\Phi=(V,[q]^V,\Cons{C})$ is the original input CSP formula of \Cref{alg-mcmc}.
Recall two projected configurations  $\Ass{X}, \Ass{Y} \in \Sigma = \bigotimes_{v \in V}\Sigma_v$  differ only at $v_0$.
Define two CSP formulas $\Phi^{X}$ and $\Phi^{Y}$ as follows:
\begin{itemize}
\item $\Phi^{X}=(V,\Dom{Q}^{X}= (Q^{X}_u)_{u \in V}, \Cons{C})$ is a CSP formula such that
\begin{align}
\label{eq-def-Qx}
Q^{X}_u = \begin{cases}
h_u^{-1}(X_u) &\text{if } u \neq \vst;\\
[q] &\text{if } u = \vst. 	
 \end{cases}
\end{align}

\item $\Phi^{Y}=(V, \Dom{Q}^{Y}= (Q^{Y}_u)_{u \in V}, \Cons{C})$ is a CSP formula such that
\begin{align*}
Q^{Y}_u = \begin{cases}
h_u^{-1}(Y_u) &\text{if } u \neq \vst;\\
[q] &\text{if } u = \vst. 	
\end{cases}
\end{align*}
\end{itemize}
By definition, $(Q^{X}_u)_{u \in V}$ and $(Q^{Y}_u)_{u \in V}$ differ only at variable $v_0$. We then define two distributions
\begin{itemize}
\item $\mu_{\Phi^{X}}$: the uniform distribution over all satisfying assignment to $\Phi^{X}$;
\item $\mu_{\Phi^{Y}}$: the uniform distribution over all satisfying assignment to $\Phi^{Y}$. 	
\end{itemize}
It is straightforward to verify $\mu_{\Phi^{X}} = \mu^{X_{V \setminus \{\vst\}}}$ and $\mu_{\Phi^{Y}} = \mu^{Y_{V \setminus \{\vst\}}}$. For any subset $S \subseteq V$, we use $\mu_{S,\Phi^{X}}$ (and $\mu_{S,\Phi^{Y}}$) to denote the marginal distribution on $S$ projected from $\mu_{\Phi^{X}}$ (and $\mu_{\Phi^{Y}}$).

Recall that $\Phi = (V, [q]^V, \Cons{C})$ is the original input CSP formula of \Cref{alg-mcmc}.
Recall that $H=(V,\+E)$ denotes the (multi-)hypergraph that models $\Phi$, where $\+E \triangleq \{\vbl{c} \mid c \in \+C\}$.
Note that $H$ also models $\Phi^{X}$ and $\Phi^{Y}$, because  $\Phi,\Phi^{X},\Phi^{Y}$ have the same sets of variables and constraints.
We assume that given any hyperedge $e \in \+E$, we can find the unique constraint in $c \in \Cons{C}$ represented by $e$.
For each hyperedge $e \in \+E$, define the \emph{volume} of $e$ with respect to $\Phi^{X}$ and $\Phi^{Y}$ as
\begin{align*}
\vol_{\Phi^{X}}(e) \triangleq \prod_{u \in e}\abs{Q^{X}_u} \quad \text{and} \quad \vol_{\Phi^{Y}}(e) \triangleq \prod_{u \in e}\abs{Q^{Y}_u}.	
\end{align*}
By \Cref{condition-projection} and~\eqref{eq-cond-path-coupling-uniform}, initially, we have for any hyperedge $e \in \+E$,
\begin{align}
\label{eq-initial-vol}
\vol_{\Phi^{X}}(e) \geq 3000 q^2d^6k^6\quad \text{and} \quad 	\vol_{\Phi^{Y}}(e) \geq 300 0q^2d^6k^6.
\end{align}
Let $\gamma$ be a threshold such that
\begin{align}
\label{eq-def-beta-T}
\gamma \triangleq 32\mathrm{e}q^2 d^3k^3 \leq 300 0q^2d^6k^6.
\end{align}

Consider an atomic constraint $c \in \Cons{C}$.
Let $\sigma \in [q]^{\vbl{c}}$ denote the unique configuration forbidden by $c$, i.e. $c(\sigma) = \False$.
The constraint $c$ is said to be satisfied by the value $x_u \in [q]$ of variable $u$ if $u \in \vbl{c}$ and $\sigma_u \neq x_u$.
In other words, given the condition that $u$ takes the value $x_u$, the constraint $c$ must be satisfied. 
A constraint $c$ is said to be satisfied by $\tau \in [q]^S$ for some subset $S \subseteq V$ if $c$ is satisfied by some $\tau_u$, where $u \in S \cap \vbl{c}$.

The coupling procedure $\cadp$ is given in \Cref{alg-coupling-v}.
\begin{algorithm}[h]
  \SetKwInOut{Input}{Input} \SetKwInOut{Output}{Output} 
  \Input{CSP formulas $\Phi^{X}=(V,\Dom{Q}^{X}= (Q^{X}_u)_{u \in V}, \Cons{C})$ and $\Phi^{Y}=(V,\Dom{Q}^{Y}= (Q^{Y}_u)_{u \in V}, \Cons{C})$, a hypergraph
    $H= (V,\mathcal{E})$ modeling $\Phi^{X}$ and $\Phi^{Y}$, two variables $v_0,\vst \in V$, a threshold  parameter $\gamma$ in~\eqref{eq-def-beta-T};}  
  \Output{a pair of assignments $\Ass{X}^{\cadp}, \Ass{Y}^{\cadp} \in [q]^{V}$.}
  % $X_{\cadp} \gets X(\+M_v)$ and $Y_{\cadp} \gets Y(\+M_v)$, where $\+M_v = \+M \setminus\{v\}$ \;
  %$X^{\cadp}(v_0) = 0$ and $Y^{\cadp}(v_0) = 1$\; 
  $V_1 \gets \{v_0\}$, $V_2 \gets V \setminus V_1$, $\Vcol \gets \emptyset$, $\Vfro \gets \emptyset$ and $\Efro \gets \emptyset$\; 
  let $\Ass{X}^{\cadp}$ and $\Ass{Y}^{\cadp}$ be two empty assignments\;
  \While{$\exists e \in \mathcal{E}$ s.t.\ $e \cap V_1 \neq \emptyset, (e \cap V_2) \setminus (\Vcol \cup \Vfro) \neq \emptyset$\label{line-while-condition} }
    {let $e$ be the first such hyperedge and $u$ be the first variable in $(e \cap V_2) \setminus (\Vcol \cup \Vfro)$\label{line-find-u-Cv}\; 
    extend $\Ass{X}^{\cadp}$ and $\Ass{Y}^{\cadp}$ to variable $u$ by sampling $(X^{\cadp}_u, Y^{\cadp}_u)$ from the optimal coupling between $\mu_{u,\Phi^{X}}$ and $\mu_{u,\Phi^{Y}}$\label{line-set-u}\;
    update $\Phi^{X}$ by setting $Q^{X}_u \gets \{X^{\cadp}_u\}$, update $\Phi^{Y}$ by setting $Q^{Y}_u \gets \{Y^{\cadp}_u\}$\label{line-update-formula}\;
    $\Vcol \gets \Vcol \cup \{u\}$\label{line-add-vset-Cv}\; 
    \If{$X^{\cadp}_u \neq Y^{\cadp}_u$}{$V_1 \gets V_1 \cup \{u\}, V_2 \gets V \setminus V_1$\label{line-add-v1}\;}
    %sample a real number $r_u \in [0,1]$ uniformly at random\label{line-draw-u-Cv}\; 
    %let $p^{X}_u = \nu_u(0\mid X^{\cadp})$ and $p^{Y}_u = \nu_u(0\mid Y^{\cadp})$\label{line-pX-pY-Cv}\; 
    
    %extend  $X^{\cadp}$ to variable $u$ s.t.\ $X^{\cadp}(u) = 0$ if $r_u \leq p^X_u$, o.w.\ $X^{\cadp}(u) = 1$\label{line-set-X-Cv}\; 
    %extend $Y^{\cadp}$ to variable $u$ s.t.\ $Y^{\cadp}(u) = 0$ if $r_u \leq p^Y_u$, o.w.\ $Y^{\cadp}(u) = 1$\label{line-set-Y-Cv}\;
    \For{$e \in \+E$ s.t. the constraint $c$ represented by $e$ is satisfied by both $X^{\cadp}_u$ and $Y^{\cadp}_u$}
      {$\mathcal{E} \gets\mathcal{E} \setminus \{e\}$, update $\Phi^{X}$ and $\Phi^{Y}$ by removing constraint $c$ from $\Cons{C}$, i.e. $\Cons{C} \gets \Cons{C} \setminus \{c\}$ \label{line-delete}\; } 
	\For{$e \in \mathcal{E}$ s.t.\ $\vol_{\Phi^{X}}(e) \leq \gamma$ or $\vol_{\Phi^{Y}}(e) \leq \gamma$\label{line-too-small}}
      { $\Vfro \gets \Vfro \cup ((e \cap V_2) \setminus \Vcol)$\label{line-add-vfro}\;
       %$\mathcal{E} \gets \mathcal{E} \setminus \{e\}$\label{line-delete-2-Cv}\;
       }
    \For{$e \in \mathcal{E}$ s.t. $(e \cap V_2) \setminus (\Vcol \cup \Vfro) = \emptyset$\label{line-edge-frozen-1}}{
    	$\Efro \gets \Efro \cup \{e\}$\label{line-edge-frozen-2}\;
    }
	\While{$\exists e \in \Efro$ s.t. $e \cap V_1 \neq \emptyset$ and $e \cap \Vfro \neq \emptyset$ \label{line-frozen-to-v1-1} }
	{ 
		$V_1 \gets V_1 \cup ( e \cap \Vfro)$, $V_2 \gets V \setminus V_1$, $\Vfro \gets \Vfro \setminus e$\label{line-frozen-to-v1-2}\;
	}
    %\For{$e \in \mathcal{E}$ s.t.\ $e \cap V_1 \neq \emptyset, e \cap V_2 \neq \emptyset$ \label{line-fix-bug-1-Cv}}
      %{\If{$(e \setminus \Vcol) \cap V_2 = \emptyset$\label{line-fix-bug-2-Cv}}
        %{$\mathcal{E} \gets \mathcal{E} \setminus \{e\}$\label{line-delete-3-Cv} \; } } 
    } 
    extend $\Ass{X}^{\cadp}$ and $\Ass{Y}^{\cadp}$ to the set $V_2 \setminus \Vcol$ 
%    s.t.\ $X^{\cadp}(V_2 \setminus \Vcol)$ and $Y^{\cadp}(V_2 \setminus \Vcol)$ is drawn from 
   by sampling $(X^{\cadp}_{V_2 \setminus \Vcol}, Y^{\cadp}_{V_2 \setminus \Vcol})$ from the optimal coupling between $\mu_{V_2 \setminus \Vcol, \Phi^{X}}$ and $\mu_{V_2 \setminus \Vcol,\Phi^{Y}}$\label{line-sample-V2-Cv}\;
    extend $\Ass{X}^{\cadp}$ and $\Ass{Y}^{\cadp}$ to the set $V_1 \setminus \Vcol$ 
%    s.t.\ $X^{\cadp}(V_1 \setminus \Vcol)$ and $Y^{\cadp}(V_1 \setminus \Vcol)$ is drawn from 
    by sampling $(X^{\cadp}_{V_1 \setminus \Vcol}, Y^{\cadp}_{V_1 \setminus \Vcol})$ from the optimal coupling between $\mu_{V_1 \setminus \Vcol,\Phi^{X}}(\cdot\mid \Ass{X}^{\cadp})$ and $\mu_{V_1 \setminus \Vcol,\Phi^{Y}}(\cdot\mid \Ass{Y}^{\cadp})$\label{line-extend-Cv}\;
    \Return{$(\Ass{X}^{\cadp}, \Ass{Y}^{\cadp})$\;}
  \caption{The coupling procedure $\cadp$}\label{alg-coupling-v}
\end{algorithm}

The coupling procedure $\cadp$ starts from two empty assignments $\Ass{X}^{\cadp}$ and $\Ass{Y}^{\cadp}$, then 
gradually extends these assignments, finally outputs two full assignments on $V$. 
The following three basic sets of variables are maintained by the coupling.
\begin{itemize}
\item $V_1/V_2$: $V_1$ is a superset of \emph{discrepancy} variables, which contains all variables $w$ such that the coupling on $w$ may be failed i.e. $X^{\cadp}_w \neq Y^{\cadp}_w$; $V_2 = V \setminus V_1$ is the complement of set $V_1$;
\item $\Vcol$: the set of variables whose values are already assigned by the coupling procedure.	
\end{itemize}

In addition, the coupling procedure $\cadp$ also maintains two  CSP formulas $\Phi^{X}=(V, \Dom{Q}^{X}, \Cons{C})),\Phi^{Y}=(V, \Dom{Q}^{Y}, \Cons{C})$ 
and a hypergraph $H=(V,\+E)$ modeling these two formulas.
In each step, we pick a suitable variable $u$ (\Cref{line-find-u-Cv}), extend $\Ass{X}^{\cadp}$ and $\Ass{Y}^{\cadp}$ to variable $u$ (\Cref{line-set-u}). 
We then remove all the constraints (together with corresponding hyperedges\footnote{Remark that $\+E$ is a multi-set of hyperedges. Once a hyperedge $e$ is removed from $\+E$ in \Cref{line-delete}, we only remove a single copy of $e$ representing the constraint $c$.}) satisfied  by both $X^{\cadp}_u$ and $Y^{\cadp}_u$~(\Cref{line-delete}), update $\Phi^{X}$ and $\Phi^{Y}$ by setting $Q^{X}_u \gets \{X^{\cadp}_u\}$ and $Q^{Y}_u \gets \{Y^{\cadp}_u\}$~(\Cref{line-update-formula}).
In other words, we force $u$ in $\Phi^{X}$ to take the value $X^{\cadp}_u$, and force $u$ in $\Phi^{Y}$ to take the value $Y^{\cadp}_u$.

The coupling procedure $\cadp$ guarantees that the volume of all hyperedges $e \in \+E$ cannot be too small in the whole procedure. 
This property is controlled by the parameter $\gamma$. 
Thus, the coupling procedure $\cadp$ is adaptive with respect to the current volumes of hyperedges.
Specifically, the following two sets are maintained during the coupling.
\begin{itemize}
\item $\Vfro$: the set of \emph{frozen} variables, which is a set of unassigned variables in $V_2$, where each $w \in \Vfro$ is incident to a hyperedge $e$ such that the volume of $e$ is below the threshold $\gamma$.
%\item $\+E$: the multi-set of current hyperedges.
\item $\Efro$: the multi-set of \emph{frozen} hyperedges such that for each hyperedge $e \in \Efro$, all unassigned variables in $e \cap V_2$ are frozen.
\end{itemize}
Once the volume of some hyperedge $e$ is below the threshold $\gamma$ (\Cref{line-too-small}), we froze all unassigned variables in $e \cap V_2$ (\Cref{line-add-vfro}). 
Once a variable becomes frozen, the coupling cannot assign values to this variable.
If in a hyperedge $e$, all unassigned variables in $e \cap V_2$ are frozen, then the coupling cannot assign values to any unassigned variables $e$, the hyperedge $e$ becomes frozen (\Cref{line-edge-frozen-1} and \Cref{line-edge-frozen-2}).
Finally, once a frozen hyperedge both contains frozen variables and variables in $V_1$, we put all frozen variables in this hyperedge into $V_1$ (\Cref{line-frozen-to-v1-1} and \Cref{line-frozen-to-v1-2}). 

%The coupling $\cadp$ also updates CSP formulas $\Phi^{X}$ and $\Phi^{Y}$ during the procedure.
%%
%In each step of the coupling, we pick a suitable variable $u$, then extend $X^{\cadp}$ and $Y^{\cadp}$ to variable $u$ using the optimal coupling between the marginal distributions on $u$ (\Cref{line-set-u}).
%%
%After that, the value of $u$ is assigned in both $X^{\cadp}$ and $Y^{\cadp}$.
%%
%We update the CSP formula $\Phi^{X}$ by setting $Q^{X}_u \gets \{X^{\cadp}_u\}$, and update $\Phi^{Y}$ by setting $Q^{Y}_u \gets \{Y^{\cadp}_u\}$~(\Cref{line-delete}).
%

%
%The coupling $\cadp$ should also guarantee that the volume of each edge $e \in \+E$ cannot be two small during the coupling procedure. 

%

%
Once the while-loop in \Cref{alg-coupling-v} terminates, we then sample assignments for variables in $V_2 \setminus \Vcol$ and $V_1 \setminus \Vcol$ from the conditional distributions (\Cref{line-sample-V2-Cv} and \Cref{line-extend-Cv}).

\begin{lemma}
\label{lemma-coupling-Cv}
The coupling procedure $\cadp$  satisfies the following properties:
\begin{itemize}
\item the coupling procedure will terminate eventually;
\item the output $\Ass{X}^{\cadp} \in [q]^V$ follows  $\mu^{X_{V \setminus \{v\}}}$ and the output $\Ass{Y}^{\cadp} \in [q]^V$ follows  $\mu^{Y_{V \setminus \{v\}}}$;
\item for any time of the coupling procedure and any $e$ in the current set $\+E$, it holds that
\begin{align*}
\vol_{\Phi^{X}}(e) \geq \frac{\gamma}{q} \quad \text{and} \quad 	\vol_{\Phi^{Y}}(e) \geq \frac{\gamma}{q};	
\end{align*}
\item for any variable $u \in V$, if $X^{\cadp}_u \neq Y^{\cadp}_u$ in the final output, then $u \in V_1$.
\end{itemize}
\end{lemma}
\begin{proof}
We prove that the coupling $\cadp$ must terminate. %It suffices to prove the while-loop must terminate.	
Consider the while-loop in \Cref{line-frozen-to-v1-1} and \Cref{line-frozen-to-v1-2}. After the \Cref{line-frozen-to-v1-2}, the hyperedge $e$ cannot satisfy the condition in \Cref{line-frozen-to-v1-1} (because $e \cap \Vfro = \emptyset$), thus the while-loop in \Cref{line-frozen-to-v1-1} and \Cref{line-frozen-to-v1-2} will terminate eventually. 
Consider the main while-loop (\Cref{line-while-condition}). After each loop, the size of $\Vcol$ will increase by 1. 
Note that the size of $\Vcol$ cannot be greater than $n$.
Hence, the coupling $\cadp$ will terminate eventually.

We prove that the output $\Ass{X}^{\cadp} \in [q]^V$ follows the distribution $\mu^{X_{V \setminus \{v\}}}$. The result for the output $\Ass{Y}^{\cadp} \in [q]^V$ can be proved in a similar way.
Consider the input CSP formula $\Phi^{X}=(V,C,(Q^{X}_u)_{u \in V})$ defined in~\eqref{eq-def-Qx}.
It holds that the uniform distribution $\mu_{\Phi^{X}}$ of all satisfying assignments to $\Phi^{X}$ is precisely the distribution $\mu^{X_{V \setminus \{v\}}}$.
Suppose $\Vcol=\{u_1,u_2,\ldots,u_{\ell}\}$, where $u_i$ is the $i$-th variable whose value is assigned by the coupling $\cadp$. The following properties holds:
\begin{itemize}
%\item for any $0 \leq i < \ell$, given the current set $\Vcol = \{u_1,u_2,\ldots,u_{i}\}$ and current $X^{\cadp},Y^{\cadp} \in [q]^{\Vcol}$, the next variable $u_{i+1}$ is fixed;
\item the value of $u_1$ is sampled from the marginal distribution $\mu_{u_1,\Phi^{X}}$;
\item for each $1\leq i <\ell$, once $u_i$ gets the value $X^{\cadp}_{u_i}$, we fix $Q^{X}_{u_i}$ as $\{X^{\cadp}_{u_i}\}$ (\Cref{line-update-formula}) and remove a subset of constraints satisfied by current $X^{\cadp}_u$ (\Cref{line-delete}); after updated $\Phi^{X}$, we sample the value of $u_{i+1}$  from the marginal distribution $\mu_{u_{i+1},\Phi^{X}}$;
\item given the assignment of $\Vcol$, the assignments of $V_2 \setminus \Vcol$ and $V_1 \setminus \Vcol$ are sampled from the conditional distributions in \Cref{line-sample-V2-Cv} and \Cref{line-extend-Cv}. 
\end{itemize}
Note that for each $u_i$, the marginal distribution $\mu_{u_{i},\Phi^{X}}$ is precisely the distribution $\mu^{X_{V \setminus \{v\}}}$ projected on $u_i$ conditional on the value of $u_j$ is fixed as $X^{\cadp}_{u_j}$ for all $j<i$.
By the chain rule, the output $\Ass{X}^{\cadp} \in [q]^V$ follows the distribution $\mu^{X_{V \setminus \{v\}}}$.

We now prove the third property. 
By~\eqref{eq-initial-vol} and~\eqref{eq-def-beta-T}, initially, for all $e \in \+E$, it holds that $\vol_{\Phi^{X}}(e) > \frac{\gamma}{q}$	and $\vol_{\Phi^{Y}}(e) > \frac{\gamma}{q}$.
Suppose during the coupling procedure, there is a time such that some hyperedge $e$ in the current set $\+E$ satisfies $\vol_{\Phi^{X}}(e) < \frac{\gamma}{q}$ 	or $\vol_{\Phi^{Y}}(e) < \frac{\gamma}{q}$. Without loss generality, we assume $\vol_{\Phi^{X}}(e) < \frac{\gamma}{q}$. The case  $\vol_{\Phi^{Y}}(e) < \frac{\gamma}{q}$ follows from symmetry. Recall 
\begin{align*}
\vol_{\Phi^{X}}(e) \triangleq \prod_{u \in e}\abs{Q^{X}_u}.
\end{align*}
Note that the volume $\vol_{\Phi^{X}}(e)$ decreases only if we update $Q^{X}_u$ for some $u \in e$ in \Cref{line-update-formula}.
Note that for any $u \in V$, it holds that $\abs{Q^{X}_u} \leq q$. In \Cref{line-update-formula}, once the coupling sets $Q^{X}_u \gets \{X_u^{\cadp}\}$, the volume $\vol_{\Phi^{X}}(e)$ decreases by at most a factor $q$. 
If $\vol_{\Phi^{X}}(e) < \frac{\gamma}{q}$, the following event must occur
\begin{itemize}
\item event $\+B$: 	the main while-loop pick a variable $u \in e$ after $\vol_{\Phi^{X}}(e) < {\gamma}$.
\end{itemize}
We show that the event $\+B$ cannot occur. Consider the first time that $\vol_{\Phi^{X}}(e) < {\gamma}$. After \Cref{line-too-small} and \Cref{line-add-vfro}, it must hold that 
\begin{align}
\label{eq-e-subet}
e \subseteq V_1 \cup \Vcol \cup \Vfro.	
\end{align}
Note that the coupling $\cadp$ only  adds variables into $V_1$ and $\Vcol$, but never deletes variables from  $V_1$ and $\Vcol$. Also note that if a variable is removed from $\Vfro$, it must be added into $V_1$(\Cref{line-frozen-to-v1-2}). Thus, ~\eqref{eq-e-subet}  holds up to the end of the coupling.
Consider the variable $u$ in event $\+B$, $u$ must satisfy $u \in V_2 \setminus (\Vcol \cup \Vfro )$.
However, by~\eqref{eq-e-subet}, there is no such variable $u$ in hyperedge $e$. Contradiction.

Finally, we prove the last property. In this proof, we  consider $V_1,V_2,\Vcol,\Vfro,\+E,\Efro$ when the main while-loop in $\cadp$  terminates. We claim that the following properties holds:
\begin{itemize}
\item (I) for any $u \in V_2 \cap \Vcol$, $X^{\cadp}_u = Y^{\cadp}_u$;
\item (II) for any $e \in \+E$ such that $e \cap V_1 \neq \emptyset$ and $e \cap V_2 \neq \emptyset$, $e \cap V_2 \subseteq \Vcol$	.
\end{itemize}
Consider the CSP formulas $\Phi^{X}$ and $\Phi^{Y}$ in \Cref{line-sample-V2-Cv}. Note that both $\Phi^{X}$ and $\Phi^{Y}$ are modeled by hypergraph $H=(V,\+E)$. Define a set of variables 
\begin{align*}
R = \bigcup_{\substack{e \in \+E\\ e \cap V_1 \neq \emptyset, e \cap V_2 \neq \emptyset}}(e \cap V_2).
\end{align*}
Recall $\mu_{\Phi^{X}}$ and $\mu_{\Phi^{Y}}$ are the uniform distributions of satisfying assignments to $\Phi^{X}$ and $\Phi^{Y}$. By the definition of $R$, conditional on any assignment $\sigma \in [q]^R$ on set $R$, the assignment on $V_2 \setminus R$ is independent with the assignment on $V_1$. By property (I) and (II), it holds that $R \subseteq V_2\cap\Vcol$ and $X^{\cadp}_{R} = Y^{\cadp}_{R}$. Since $R \subseteq \Vcol$ and $X^{\cadp}_{R} = Y^{\cadp}_{R}$, for any $u \in R$, $\abs{Q^{X}_u} = \abs{Q^{Y}_u} = 1$ and $Q^{X}_u = Q^{Y}_u$. Hence, in $\Phi^{X}$ and $\Phi^{Y}$, variables in $R$ are fixed as a same value in $[q]$. Thus, $\mu_{V_2 \setminus \Vcol, \Phi^{X}}$ and $\mu_{V_2 \setminus \Vcol,\Phi^{Y}}$ are identical distributions. By \Cref{line-sample-V2-Cv},
\begin{align}
\label{eq-V2-same}
X^{\cadp}_{V_2 \setminus \Vcol} = Y^{\cadp}_{V_2 \setminus \Vcol}
\end{align}
Combining property (I) and~\eqref{eq-V2-same} proves that $X^{\cadp}_{V_2} = Y^{\cadp}_{V_2}$.
This proves the last property.

We finish the prove by proving properties (I) and (II). The property (I) is trivial, because for any $u \in \Vcol$, if $X^{\cadp}_u \neq Y^{\cadp}_u$, then by~\Cref{line-add-v1}, it must hold that $u \in V_1$. 
We then prove property (II). Suppose there is an hyperedge $e$ such that $e\cap V_1 \neq \emptyset$, $e \cap V_2 \neq \emptyset$ and $e$ violates property (II). 
We define a set 
\begin{align*}
S(e) = (e \cap V_2) \setminus \Vcol  = (e \setminus V_1) \setminus \Vcol \neq \emptyset.	
\end{align*}
There are only two possibilities for the set $S(e)$, we show neither of them is possible.
\begin{itemize}
\item $S(e) \not\subseteq \Vfro$: in this case, $e$ satisfies the condition in the main while-loop (\Cref{line-while-condition}), the main while-loop cannot terminate; contradiction.
\item $S(e) \subseteq \Vfro$: in this case, by \Cref{line-edge-frozen-1} and \Cref{line-edge-frozen-2}, $e \in \Efro$; hence, $e$ satisfies the condition in \Cref{line-frozen-to-v1-1}, then by \Cref{line-frozen-to-v1-2}, all variables in $e \cap \Vfro$ are removed from $\Vfro$ and added into $V_1$, thus there is no such non-empty subset $S(e) \subseteq e$ such that $S(e) \subseteq \Vfro$; contradiction.
\end{itemize}
Hence, such non-empty subset $S(e)$ does not exist, which implies property (II) holds.
\end{proof}

By \Cref{lemma-coupling-Cv} and the coupling lemma (\Cref{lemma-coupling-ineq}), to bound the $I_{\vst}$ in~\eqref{eq-proof-Iv}, we can bound
\begin{align}
\label{eq-proof-V1}
I_{\vst} = 	\DTV{\mu_{\vst}^{X_{V\setminus \{\vst\} }}}{\mu_{\vst}^{Y_{V\setminus \{\vst\} }}} \leq \Pr[\cadp]{\vst \in V_1},
\end{align}
where $V_1$ denotes the set $V_1$ at the end of the coupling $\cadp$.

In the rest of the proof,
our task is  to bounding the RHS of~\eqref{eq-proof-V1}.
From now, we use hypergraph $H=(V,\+E)$ to model the \emph{input} CSP formulas $\Phi^{X}$ and $\Phi^{Y}$ in \Cref{alg-coupling-v}.
For any  $v \in V$, define
\begin{align*}
\Nv(v) &\triangleq \{u \neq v \mid \exists e \in \mathcal{E} \text{ s.t. } u,v \in e \}.	
\end{align*}
We say a variable $u$ is incident to a hyperedge $e$ if $u \in e$; a sequence of variables $v_0,v_1,\ldots,v_{\ell}$ is a path in hypergraph $H$ if $v_i \in \Nv(v_{i-1})$ for all $1 \leq i \leq \ell$. We define the failed variables and failed edges.
\begin{definition}
Consider the time when the main while-loop in coupling procedure $\cadp$ terminates.
\begin{itemize}
\item A variable $u \in V$ is said to be failed if 	$u \in \Vcol$ and $X^{\cadp}_u \neq Y^{\cadp}_u$.
\item A hyperedge $e \in \+E$ is said to be failed if both of the following two properties hold:
\begin{enumerate}
\item the constraint represented by $e$ is not satisfied by  both $\Ass{X}^{\cadp}$ and $\Ass{Y}^{\cadp}$; 
\item $\vol_{\Phi^{X}}(e) < \gamma$ or $\vol_{\Phi^{Y}}(e) < \gamma$.
\end{enumerate}
\end{itemize}
	
\end{definition}

\begin{lemma}
\label{lemma-path}
For any $u \in V_1$, there exists a path $u_0,u_1,\ldots,u_{\ell} \in V$ in $H$ such that 
\begin{itemize}
\item $u_0=v_0$ is the initial disagreement variable, $u_\ell = u$ and $u_i \in V_1$ for all $0 \leq i \leq \ell$;
\item for any $1\leq i \leq \ell$, either $u_i$ is failed or $u_i$ is incident to a failed hyperedge $e_i$.
\end{itemize}
\end{lemma}
\begin{proof}
Suppose $V_1 = \{v_0,v_1,v_2,\ldots,v_m\}$, where $v_0$ is the initial disagreement variable and $v_i$ is the $i$-th 	variable added into set $V_1$. If a set of variables are added into $V_1$ at the same time (\Cref{line-frozen-to-v1-2}), we break tie arbitrarily.  We prove the first part of the lemma by induction on the index $i$.

The base case is $i = 0$, the first part of the lemma holds for the path that only contains $v_0$.

Assuming the lemma holds up to index $i$, we prove the lemma for index $i+1$. Consider the time when $v_{i+1}$ is added into the set $V_1$. There are following two possibilities.
\begin{itemize}
\item $v_{i+1}$ is added in \Cref{line-add-v1}. Consider  the hyperedge $e$ in \Cref{line-find-u-Cv}. It holds that $v_{i+1} \in e$ and $e \cap V_1 \neq \emptyset$, where $V_1 = \{v_0,v_1,\ldots,v_i\}$. Pick an arbitrary $v_{j} \in e \cap V_1$. By induction hypothesis, since $j < i$, there exists a path $u_0 = v_0, u_1,u_2,\ldots,u_\ell = v_j$ for $v_j$. Note that $v_{i+1} \in e$ and $v_{j} \in e$. We can find the path $u_0 = v_0, u_1,u_2,\ldots,u_\ell = v_j,u_{\ell + 1} = v_{i+1}$ for $v_{i+1}$.
\item $v_{i+1}$ is added in \Cref{line-frozen-to-v1-2}. Consider  the hyperedge $e$ satisfying the condition in \Cref{line-frozen-to-v1-1}. It holds that $v_{i+1} \in e$ and $e \cap V_1 \neq \emptyset$, where $V_1 = \{v_0,v_1,\ldots,v_i\}$. Pick an arbitrary $v_{j} \in e \cap V_1$. By induction hypothesis, since $j < i$, there exists a path $u_0 = v_0, u_1,u_2,\ldots,u_\ell = v_j$ for $v_j$. Note that $v_{i+1} \in e$ and $v_{j} \in e$. We can find the path $u_0 = v_0, u_1,u_2,\ldots,u_\ell = v_j,u_{\ell + 1} = v_{i+1}$ for $v_{i+1}$.
\end{itemize}

We now prove the second part of the lemma. It suffices to show that for any $u \in V_1 \setminus \{v_0\}$, either $u$ is failed or $u$ is incident to a failed hyperedge $e$. Note that a variable $u$ is added into $V_1$ in either \Cref{line-add-v1} or \Cref{line-frozen-to-v1-2}. If $u$ is added in \Cref{line-add-v1}, then it holds that $X^{\cadp}_u \neq Y^{\cadp}_u$, thus $u$ is a failed variable. Suppose $u$ is added in \Cref{line-frozen-to-v1-2}. Before the execution of \Cref{line-frozen-to-v1-2}, $u \in \Vfro$ must be a frozen variable. Consider the moment that $u$ becomes frozen. By \Cref{line-add-vfro}, $u$ must belong to a hyperedge $e$ such that $e$ is not satisfied by both $\Ass{X}^{\cadp}$ and $\Ass{Y}^{\cadp}$ (otherwise, $e$ is deleted in \Cref{line-delete}) and $\min\{\vol_{\Phi^{X}}(e), \vol_{\Phi^{Y}}(e)\} < \gamma$. Note that after \Cref{line-add-vfro}, $e \subseteq V_1 \cup \Vcol \cup \Vfro$. After that, in the main while-loop, the coupling $\cadp$ cannot assign values to any unassigned variables in $e$. Thus, this hyperedge $e$ is not satisfied by both $\Ass{X}^{\cadp}$ and $\Ass{Y}^{\cadp}$ up to the main while-loop in $\cadp$ terminates. Hence, $e$ is a failed hyperedge and $u$ is incident to $e$.
\end{proof}

\Cref{lemma-path} says if a variable belongs to $V_1$, there exists a path satisfying the condition in \Cref{lemma-path}. However, the failure probability of such path is not easy to bound. We next modify such path into a sequence whose failure probability is easy to bound.

Define the length of a path by the number of variables in this path  minus 1, e.g. the length of the path $v_1,v_2,\ldots,v_{\ell}$ is $\ell-1$. For any two variables $u, w \in V$, the distance between $u$ and $w$ in $H$, denoted as $\dist_H(u,w)$, is the length of the shortest path between $u$ and $w$ in $H$. 
We extend the notion of distance to subsets of variables. 
For any variable $u \in V$ and subsets $S,T \subseteq V$, define 
\begin{align*}
\dist_H(u,S) &\triangleq \min_{w \in S}\dist_H(u,w);\\
\dist_H(S, T) &\triangleq \min_{w \in S,w' \in T}\dist_H(w,w').		
\end{align*}
For such distance function $\dist_H(\cdot,\cdot)$, the triangle inequality may not hold for any subsets.
But we will use the following two specific triangle inequalities.
\begin{align}
\forall u_1,u_2,u_3 \in V, \qquad &\dist_H(u_1,u_2) \leq \dist_H(u_1,u_3) + \dist_H(u_3,u_2)\label{eq-triangle-1}\\
\forall u \in V, S, T \subseteq V \quad & \dist_H(S, T) \leq \dist_H(S, u) + \dist_H(u, T).\label{eq-triangle-2}
\end{align}
The inequality~\eqref{eq-triangle-1} holds trivially. Suppose $\dist_H(S,u) = \dist_H(u_S,u)$ for $u_S \in S$ and $\dist_H(u, T)= \dist_H(u,u_T)$ for $u_T \in T$. By~\eqref{eq-triangle-1}, we have
\begin{align*}
\dist_H(S, T) \overset{(\star)}{\leq} \dist_H(u_S,u_T) \leq \dist_H(u_S,u) + \dist_H(u,u_T) = 	\dist_H(S, u) + \dist_H(u, T),
\end{align*}
where $(\star)$ holds because $u_S \in S$ and $u_T \in T$. Remark that~\eqref{eq-triangle-2} covers~\eqref{eq-triangle-1}, because $S$ and $T$ may only contain a single variable.

%For any variable $u \in V$ and hyperedge $e \in \+E$, define
%\begin{align*}
%\dist_H(u, e) = \min_{w \in e}\dist_H(u,w).	
%\end{align*}
%Note that $\dist_H(u,e) = 0$ if $u \in e$.
%For any hyperedges $e,e'\in\+E$, define
%\begin{align*}
%\dist_H(e, e') = \min_{u \in e,w \in e'}\dist_H(u,w).	
%\end{align*}
%Note that $\dist_H(e,e') = 0$ if $e \cap e' \neq \emptyset$.
%For any variable $v \in V$ and any hyperedge $e \in \+E$, we abuse the notation $\vbl{\cdot}$ and define
%\begin{align*}
%\vbl{u} = \{u\} \quad\text{and}\quad \vbl{e} = e.	
%\end{align*}

We have the following lemma.
\begin{lemma}
\label{lemma-path-strong}
For any $u \in V_1 \setminus \{v_0\} $, there exists a sequence of sets $S_1,S_2,\ldots,S_{\ell}$, where each $S_i$ is either a hyperedge or a set containing a single variable, such that
\begin{itemize}
\item $S_1,S_2,\ldots,S_\ell$ are mutually disjoint;
\item $\dist_H(v_0,S_1) \leq 2$ and $\dist_H(u,S_{\ell}) = 0$;
\item for any $1\leq i \leq \ell - 1$, $\dist_H(S_i,S_{i+1}) \leq 2$.	
\item for each $1\leq i \leq \ell$, $S_i$ either contains a failed variable or $S_i$ is a failed hyperedge.
\end{itemize}
\end{lemma}

\begin{proof}
Fix a variable $u \in V_1 \setminus \{v_0\}$. Let $v_0,v_1,\ldots,v_m$ where $v_m = u$ denote the path in \Cref{lemma-path}. 	For each $1 \leq i \leq m$ if $v_i$ is not a failed variable, we use $e_i$ to denote the failed hyperedge incident to $v_i$; if $v_i$ is a failed variable, we let $e_i = \{v_i\}$. 
We first show that how to construct the sequence $S_1,S_2,\ldots,S_{\ell}$, then we show that such sequence satisfies the properties in the lemma.

Let $\+S$ be an empty stack. 
Let $P$ denote the path $(v_1,v_2,\ldots,v_m)$.
Remark that $P$ does not contain variable $v_0$.
We repeat the following procedure until $P$ becomes an empty path.
We pick the last variable in the path $P$, denote this variable as $v_i$.
We search for the minimum index $j$ such that $j < i$ and $e_i \cap e_j \neq \emptyset$. Here are two cases depending on whether such index $j$ exists.
\begin{itemize}
\item If such index $j$ does not exist, then push $e_i$ into the stack $\+S$, remove $v_i$ from the path $P$.
\item If such index $j$ exists, then push $e_i$ into the stack $\+S$, remove all $v_t$ for $j\leq t \leq i$ from the path $P$.
\end{itemize}
Let $S_1,S_2,\ldots,S_{\ell}$ be the elements in stack $\+S$ from top to bottom.
 
We now prove that all $S_i$ are disjoint. Suppose there are two indices $j<i$ such that $S_i \cap S_j \neq\emptyset$.
Suppose  $S_i = e_{i*}$ and $S_j = e_{j^*}$. It holds that $i^* > j^*$. $e_{j^*}$ must be removed when processing $e_{i^*}$, thus $e_{j^*}$ cannot be added into stack $\+S$. Contradiction. This proves the first property.

We now prove the second property.
Note that $u \in e_m$ and $S_\ell = e_m$, thus $\dist(u,S_\ell) = 0$. 
To bound $\dist_H(v_0,S_1)$, we consider two cases. 
\begin{itemize}
\item Case $S_1 = e_1$. Note that $v_0$ and $v_1$ are adjacent in $H$, i.e. $\dist_H(v_0,v_1) = 1$.  It holds that $v_1 \in S_1 = e_1$. Hence, $\dist_H(v_0, S_1) \leq \dist_H(v_0,v_1) = 1$;
\item Case $S_1 \neq e_1$. Suppose $S_1 = e_t$. In this case, it must hold that $e_1 \cap e_t \neq \emptyset$, thus $\dist_H(v_1,e_t) \leq \dist_H(v_1,v^*) = 1$. where $v^* \in e_1 \cap e_t$ is an arbitrary variable. Note that $\dist_H(v_0,v_1) = 1$. By triangle inequality in~\eqref{eq-triangle-2}, we have $\dist_H(v_0,e_t)\leq \dist_H(v_0,v_1)  + \dist_H(v_1,e_t)\leq 2$.
\end{itemize}

Finally, we bound the distance $\dist_H(S_i,S_{i+1})$. Suppose $S_{i+1} = e_{j}$ and $S_i = e_{j'}$. Here are  two cases.
\begin{itemize}
\item Case $j' = j - 1$: Note that $\dist_H(v_j,v_{j'}) = 1$, $v_j \in e_j$ and $v_{j'} \in e_{j'}$. We have $\dist_H(e_j, e_{j'})\leq \dist_H(v_j,v_{j'}) \leq 1$. Hence, $\dist_H(S_i,S_{i+1}) = \dist_H(e_j,e_{j'}) \leq 1$.
\item Case $j' < j - 1$: Consider the moment when $S_{i+1} = e_{j}$ is added into $\+S$. It must hold that $e_{j'+1} \cap e_j \neq \emptyset$. 
	Note that $v_{j'} \in e_{j'}$ and $\dist_H(v_{j'},v_{j'+1}) = 1$ . We have $\dist_H(e_{j'},v_{j'+1}) \leq \dist_H(v_{j'},v_{j'+1})= 1$. 
	Note that $v_{j'+1} \in e_{j'+1}$ and $e_{j'+1} \cap e_j \neq \emptyset$. It holds that $\dist_H(v_{j'+1}, e_j) \leq \dist_H(v_{j'+1},v^*) = 1$, where $v^* \in e_{j'+1} \cap e_j$ is an arbitrary variable.
	By triangle inequality in~\eqref{eq-triangle-2}, $\dist_H(e_{j'},e_j) \leq \dist(e_{j'},v_{j'+1})+ \dist_H(v_{j'+1},e_j)  \leq 2$.
\end{itemize}
Combining two cases proves the third property.

For the last property, by \Cref{lemma-path},
it is easy to see that each $S_i$ is either a failed hyperedge or a set containing a single failed variable.
\end{proof}

We say a sequence of sets $S_1,S_2,\ldots,S_\ell$ is a \emph{percolation sequence (PS)} if the following three properties are satisfied:
\begin{itemize}
\item $S_1,S_2,\ldots,S_\ell$ are mutually disjoint;
\item $\dist_H(v_0,S_1) \leq 2$;
\item for any $1\leq i \leq \ell - 1$, $\dist_H(S_i,S_{i+1}) \leq 2$.		
\end{itemize}
We say a percolation sequence $S_1,S_2,\ldots,S_\ell$ is a percolation sequence for $\vst$ if $\dist_H(\vst,e_\ell) = 0$, i.e. $\vst \in e_{\ell}$.
For any $S_{i}$ in sequence, we say $S_{i}$ fails if either $S_{i}$ contains a failed variable or $S_{i}$ is a failed hyperedge. By~\eqref{eq-bound-Iv} and \Cref{lemma-path-strong}, we have
\begin{align}
\label{eq-bound-Iv-path}
I_{\vst} \leq \Pr[\cadp]{X^{\cadp}_{\vst} \neq Y^{\cadp}_{\vst}} \leq \sum_{\text{PS for }\vst: S_1,S_2,\ldots,S_{\ell}}\Pr[\cadp]{\forall 1 \leq i \leq \ell, S_{i} \text{ fails}}.	
\end{align}
The following lemma bounds the probability that all elements in a PS fail.
\begin{lemma}
\label{lemma-failed-path}
Fix a percolation sequence (PS) $S_1,S_2,\ldots,S_\ell$ to $\vst$. It holds that
\begin{align*}
\Pr[\cadp]{\forall 1 \leq i \leq \ell, S_{i} \text{ fails}} \leq \prod_{\substack{1\leq i\leq \ell\\ \text{$S_i$ contains a single variable}}}\frac{1}{8 k^3 d^2}\ \prod_{\substack{1\leq i\leq \ell\\ \text{$S_i$ is a hyperedge}}}\frac{1}{8k^3d^3}. 
%\max\left\{ \frac{2}{\alpha k}, \frac{\beta \mathrm{e}}{\alpha q}\right\}.
\end{align*}
\end{lemma}
We need the following technical lemma to prove \Cref{lemma-failed-path}.
We introduce a parameter $s$ to write $\gamma$ defined in~\eqref{eq-def-beta-T} as
\begin{align}
\label{eq-def-s-in-gamma}
\gamma  = s \mathrm{e} q^2 d k, \quad\text{where } s \triangleq 32k^2d^2.
\end{align}
\begin{lemma}
\label{lemma-coupling-uniform}
During the coupling procedure $\cadp$, the CSP formulas $\Phi^{X} = (V,(Q^{X}_u)_{u \in V},\Cons{C})$ and $\Phi^{Y} = (V, (Q^{Y}_u)_{u \in V},\Cons{C})$ always satisfies that for any $u \in V \setminus (\Vcol \cup \{v_0\})$, $Q^{Y}_u = Q^{X}_u$ and for any $j \in Q^{X}_u=Q^{Y}_u$, 
\begin{equation}
\label{eq-coupling-local-uniform}
\begin{split}
\frac{1}{q_u}\tp{1 - \frac{\mathrm{4}}{s k}} &\leq \mu_{u,\Phi^{X}}(j) \leq \frac{1}{q_u}\tp{1 + \frac{4}{sk}}\\
\frac{1}{q_u}\tp{1 - \frac{\mathrm{4}}{s k}} &\leq \mu_{u,\Phi^{Y}}(j) \leq \frac{1}{q_u}\tp{1 + \frac{4}{s k}},
\end{split}
\end{equation}
where $q_u = \abs{Q_u^X} = \abs{Q^{Y}_u}$, thus $\DTV{\mu_{u,\Phi^{X}}}{\mu_{u,\Phi^{Y}}} \leq \frac{4}{s k}$.

Furthermore, for any optimal coupling $(x,y) \in Q^{X}_u\times Q^{Y}_u$ between $\mu_{u,\Phi^{X}}$ and $\mu_{u,\Phi^{Y}}$, it holds that
\begin{align*}
\forall j \in Q^{X}_u=Q^{Y}_u \quad \Pr{x = j \lor y = j} = \max\left\{ \mu_{u,\Phi^{X}}(j),\mu_{u,\Phi^{Y}}(j) \right\}	 \leq \frac{1}{q_u}\tp{1 + \frac{4}{s k}}.
\end{align*}
\end{lemma}
\begin{proof}
Initially, the input $\Phi^{X}$ and $\Phi^{Y}$ satisfy 	$Q^{X}_u = Q^{Y}_u$ for any $u \in V \setminus \{v_0\}$. Consider each update step in \Cref{line-update-formula}. After the value of $u$ is assigned, we put the variable $u$ into $\Vcol$ in \Cref{line-add-vset-Cv}. 
It still holds that   $Q^{Y}_v = Q^{X}_v$ for any $v \in V \setminus (\Vcol \cup \{v_0\})$.
By \Cref{lemma-coupling-Cv}, at any time, for any $e$ in current $\+E$, it holds that
\begin{align*}
\vol_{\Phi^{X}}(e) &= \prod_{u \in e}q_u \geq \frac{\gamma}{q} = s\mathrm{e}qdk \\
\vol_{\Phi^{Y}}(e) &= \prod_{u \in e}q_u \geq \frac{\gamma}{q} = s \mathrm{e}qdk.
\end{align*}
We now prove~\eqref{eq-coupling-local-uniform} for $\Phi^{X}$. The result for $\Phi^{Y}$ can be proved in a similar way.
Let $\+D$ denote the product distribution such that each variable $v \in V$ takes a value from $Q^{X}_v$ uniformly at random. 
Let $B_c$ to denote the bad event that the constraint $c$ is not satisfied. 
Let $\+B = (B_c)_{c \in \Cons{C}}$ denote the collection of bad events.
Let $\Gamma(\cdot)$ be defined as in the Lov\'asz local lemma (\Cref{theorem-LLL}).
For each $c \in \Cons{C}$, let $x(B_c) = \frac{1}{sqdk}$.
For each constraint $c \in \Cons{C}$,
\begin{align*}
\Pr[\+D]{B_c} &= \prod_{u \in \vbl{c}} \frac{1}{q_u} \leq \frac{1}{seqdk} \leq \frac{1}{sqdk}\tp{1 - \frac{1}{sqdk} }^{sqdk - 1} \leq \frac{1}{sqdk}\tp{1 - \frac{1}{sqdk} }^{dk-1}\\	
&\leq x(B_c)\prod_{B_{c'} \in \Gamma(B_c)}(1 - x(B_{c'})),
\end{align*}
where the last inequality holds because the maximum degree of the dependency graph is at most $k(d-1) \leq dk - 1$. Fix a $j \in Q^{X}_u = Q^{Y}_u$. Let $A$ denote the event that $v$ takes the value $j$. Note that $\abs{\Gamma(A)} \leq d$. By Lov\'asz local lemma (\Cref{theorem-LLL}), we have
\begin{align*}
 \mu_{u,\Phi^{X}}(j) = \Pr[\mu_{\Phi^{X}}]{A} \leq \frac{1}{q_u}\tp{1 - \frac{1}{sqdk}}^{-d}	 \leq \frac{1}{q_u} \exp \tp{\frac{2}{sqk}} \leq \frac{1}{q_u}  \tp{1 + \frac{4}{sqk}},
\end{align*}
which implies the upper bound in~\eqref{eq-coupling-local-uniform}. 
Let $A'$ denote the event that $v$ does not take the value $j$. Note that $\abs{\Gamma(A')} \leq d$. By Lov\'asz local lemma (\Cref{theorem-LLL}), we have
\begin{align*}
\Pr[\mu_{\Phi^{X}}]{A'} \leq \tp{1 - \frac{1}{q_u}}\tp{1 - \frac{1}{sqdk}}^{-d}	 \leq \tp{1-\frac{1}{q_u}} \exp \tp{\frac{2}{sqk}} \leq \tp{1-\frac{1}{q_u}}  \tp{1 + \frac{4}{sqk}}.
\end{align*}
We have
\begin{align*}
 \mu_{u,\Phi^{X}}(j) = 1 - \Pr[\mu_{\Phi^{X}}]{A'} \geq 	1 - \tp{1-\frac{1}{q_u}}  \tp{1 + \frac{4}{sqk}} = \frac{1}{q_u}\tp{1 - \frac{4q_u}{sqk} + \frac{4}{sqk}}\geq \frac{1}{q_u}\tp{1 - \frac{4}{sk}},
\end{align*}
where the last inequality holds because $q_u \leq q$. This proves the lower bound in~\eqref{eq-coupling-local-uniform}. The inequalities in~\eqref{eq-coupling-local-uniform} imply
\begin{align*}
\DTV{\mu_{u,\Phi^{X}}}{\mu_{u,\Phi^{Y}}} \leq \frac{1}{2}\sum_{j \in Q_u^X = Q_u^Y}\abs{\mu_{u,\Phi^{X}}(j) - \mu_{u,\Phi^{Y}}(j)} = \frac{4}{s k}.
\end{align*}

Let $(x,y)\in Q_u^X \times Q_u^Y$ be the optimal coupling between  $\mu_{u,\Phi^{X}}$ and $\mu_{u,\Phi^{Y}}$. It holds that
\begin{align*}
\Pr{x = {y}} = 1 - \DTV{\mu_{u,\Phi^{X}}}{\mu_{u,\Phi^{Y}}}	
\end{align*}
Define a set $S = \{j \in Q^{X}_u = Q^{Y}_u \mid \mu_{u,\Phi^{X}}(j) \geq \mu_{u,\Phi^{Y}}(j) \}$. Note that $\sum_{j\in Q^{X}_u}\mu_{u,\Phi^{X}}(j) = \sum_{j \in Q^{Y}_u}\mu_{u,\Phi^{Y}}(j) = 1$. We have
$
\DTV{\mu_{u,\Phi^{X}}}{\mu_{u,\Phi^{Y}}} = \sum_{j \in S}(\mu_{u,\Phi^{X}}(j) - \mu_{u,\Phi^{Y}}(j))
$,
which implies
\begin{align}
\label{eq-opt-coupling-1}
\Pr{{x} = {y}} &= 1 - \sum_{j \in S}(\mu_{u,\Phi^{X}}(j) - \mu_{u,\Phi^{Y}}(j)) = \tp{1 -  \sum_{j \in S}\mu_{u,\Phi^{X}}(j)} + \sum_{j \in S} \mu_{u,\Phi^{Y}}(j)\notag\\
&= \sum_{j \in Q_u^X \setminus S}\mu_{u,\Phi^{X}}(j) +  \sum_{j \in S} \mu_{u,\Phi^{Y}}(j)\notag\\
& = \sum_{j \in Q_u^X} \min\{\mu_{u,\Phi^{X}}(j), \mu_{u,\Phi^{Y}}(j)\}.
\end{align}
On the other hand, since $({x}, {y}) \in Q_u^X \times Q_u^Y$ as a valid coupling, we have
\begin{align*}
\forall j \in Q_u^X,\quad \Pr{{x}={y}=j} \leq \min\{\mu_{u,\Phi^{X}}(j), \mu_{u,\Phi^{Y}}(j)\}.	
\end{align*}
This implies that
\begin{align}
\label{eq-coupling-min}
\forall j \in Q_u^X	\quad \Pr{{x}={y}=j} = \min\{\mu_{u,\Phi^{X}}(j), \mu_{u,\Phi^{Y}}(j)\}.	
\end{align}
Fix a $j \in Q_u^X$. Without loss of generality, assume $\mu_{u,\Phi^{X}}(j) \geq \mu_{u,\Phi^{Y}}(j)$ (the case $\mu_{u,\Phi^{X}}(j)< \mu_{u,\Phi^{Y}}(j)$ follows from symmetry). By~\eqref{eq-coupling-min}, ${y} = j$ implies ${x} = j$. Thus ${x} = j \lor {y} = j$ if and only if ${x} = j$. Thus,
\begin{align*}
\Pr{{x} = j \lor {y} = j} = \max\left\{ \mu_{u,\Phi^{X}}(j),\mu_{u,\Phi^{Y}}(j) \right\}	 \leq \frac{1}{q_u}\tp{1 + \frac{4}{s k}}. &\qedhere
\end{align*}
\end{proof}

Now, we are ready to prove \Cref{lemma-failed-path}.
\begin{proof}[Proof of \Cref{lemma-failed-path}]
Given $\+S = S_1,S_2,\ldots,S_\ell$, we define a set of variables
$
\vbl{\+S} = \cup_{i=1}^\ell S_i.
$
For each  $1\leq i \leq \ell$, sample a random real number $r_i \in [0,1]$ uniformly and independently.

Consider the following implementation of coupling $\cadp$. In \Cref{line-set-u}, we need to sample $X^{\cadp}_u$ and $Y^{\cadp}_u$ from the optimal coupling between marginal distributions $\mu_{u,\Phi^{X}}$ and $\mu_{u,\Phi^{Y}}$. If $u \in \vbl{\+S}$, then we use the following implementation.
We can find a unique $S_i$ such that $u \in S_i$, because all $S_i$ are mutually disjoint.
We use random number $r_i$ to implement the optimal coupling between $\mu_{u,\Phi^{X}}$ and $\mu_{u,\Phi^{Y}}$. 
Here are two case for $S_i$: (1) $S_i = \{u\}$; (2) $S_i$ is a hyperedge and $u \in S_i$. We handle two cases separately. 

Suppose $S_i = \{u\}$. The optimal coupling satisfies $\Pr[\cadp]{X^{\cadp}_u \neq Y^{\cadp}_u} = \DTV{\mu_{u,\Phi^{X}}}{\mu_{u,\Phi^{Y}}}$. The optimal coupling can be implemented as follows.
\begin{itemize}
\item If $r_i \leq \DTV{\mu_{u,\Phi^{X}}}{\mu_{u,\Phi^{Y}}}$, then sample a pair $(X^{\cadp}_u,Y^{\cadp}_u)$ from the optimal coupling conditional on $X^{\cadp}_u \neq Y^{\cadp}_u$;
\item If $r_i > \DTV{\mu_{u,\Phi^{X}}}{\mu_{u,\Phi^{Y}}}$, then sample a pair $(X^{\cadp}_u,Y^{\cadp}_u)$ from the optimal coupling conditional on $X^{\cadp}_u = Y^{\cadp}_u$.
\end{itemize}
By \Cref{lemma-coupling-uniform}, it holds that
$\DTV{\mu_{u,\Phi^{X}}}{\mu_{u,\Phi^{Y}}} \leq \frac{4}{s k} = \frac{1}{8k^3d^2}$.  Define the following event for $S_i$:
\begin{align}
\label{eq-def-B-variable}
\+B_i: \quad r_i \leq 	\frac{4}{s k} = \frac{1}{8k^3d^2}.
\end{align}
According to the implementation, if variable $u$ fails in $\cadp$, then event $\+B_i$ must occur. 

Suppose $S_i = e$ is a hyperedge. 
Suppose $e$ represents the constraint $c$ such that $c$ forbids a unique configuration $\sigma \in [q]^{\vbl{c}}$, i.e. $c(\sigma) = \False$.
In addition to $r_i$, we maintain two variables $M_i$ and $D_i$ for $S_i$, where $M_i \in [0,1]$ is a real number, $D_i \in \{0,1\}$ is a Boolean variable. Initially, $M_i = 1$ and $D_i = 0$. Suppose the coupling $\cadp$ pick a variable $u \in e$.
% and the constraint $c$ contains literal $(u \neq c_u)$.
 We sample $X_u^{\cadp}$ and $Y_u^{\cadp}$ via following procedure $\mathsf{Couple}(u)$.
\begin{itemize}
\item If $D_i = 1$, sample $X_u^{\cadp}$ and $Y_u^{\cadp}$ from the optimal coupling between $\mu_{u,\Phi^{X}}$ and $\mu_{u,\Phi^{Y}}$. We does not need to use $r_i$ to implement this sampling step.
\item If $D_i = 0$, let $p_u = \max\{\mu_{u,\Phi^{X}}(\sigma_u), \mu_{u,\Phi^{Y}}(\sigma_u)\}$, then check whether $r_i \leq M_ip_u$.
\begin{enumerate}
\item if $r_i > M_i p_u$, sample  $X_u^{\cadp}$ and $Y_u^{\cadp}$ from the optimal coupling between $\mu_{u,\Phi^{X}}$ and $\mu_{u,\Phi^{Y}}$ conditional on $X_u^{\cadp} \neq \sigma_u \land Y_u^{\cadp} \neq\sigma_u$; then set $D_i \gets 1$;
\item if $r_i \leq M_i p_u$, sample  $X_u^{\cadp}$ and $Y_u^{\cadp}$ from the optimal coupling between $\mu_{u,\Phi^{X}}$ and $\mu_{u,\Phi^{Y}}$ conditional on $X_u^{\cadp}= \sigma_u \lor Y_u^{\cadp} = \sigma_u$; then set $M_i \gets M_ip_u$.
\end{enumerate}
\end{itemize}

We first prove that above implementation is a valid coupling between $\mu_{u,\Phi^{X}}$ and $\mu_{u, \Phi^{Y}}$.
Note that if $D_i =1 $, then there is a variable $u \in e = S_i$ such that $e$ is satisfied by both $X^{\cadp}_u$ and $Y^{\cadp}_u$, thus $D_i$ indicates whether $e$ is removed by the coupling. 
We claim 
\begin{align}
\label{eq-property-proof}
\text{conditional on $D_i = 0$ and $M_i = m_i$, $r_i$ is a uniform random real number in $[0,m_i]$	}.
\end{align}
Let $\+R$ denote all the randomness of the coupling $\cadp$ except the randomness of $r_i$.
We first fix $\+R$, then prove~\eqref{eq-property-proof} by induction. Initially, $r_i$ is sampled from $[0,1],M_i=1,D_i=0$, the property holds. 
Consider one execution of $\mathsf{Couple}(u)$. 
Suppose $D_i = 0$ and $M_i = m_i$ before the execution.
We show that~\eqref{eq-property-proof} still holds after we  sampled $X_u^{\cadp}$ and $Y_u^{\cadp}$ according to $\mathsf{Couple}(u)$. 
By induction hypothesis, $r_i$ is  a uniform random real number in $[0,m_i]$.
Note that conditional on $\+R$ and $D_i = 0$,  the value of $p_u$ is fixed.\footnote{This is because $\+R$ fixes all the randomness except the randomness of $r_i$. In our implementation, we only use $r_i$ to compare with a threshold $M_ip_u$ when we couple $X^{\cadp}_u$ and $Y^{\cadp}_u$ in \Cref{line-set-u} for some $u \in e = S_i$. Conditional further on $D_i = 0$, the results of all previous comparisons are fixed, namely, $r_i$ is smaller or equal to all the thresholds $m_ip_u$.  Hence, given $\+R$ and $D_i = 0$, the previous procedure of $\cadp$ is fully determined, which implies $p_u$ is fixed.}
After the procedure $\mathsf{Couple}(u)$, $D_i = 0$ if and only if $r_i \leq m_ip_u$.
Since $r_i$ is a uniform random real number in $[0,m_i]$, conditional on $r_i \leq m_ip_u$, $r_i$ is a uniform random real number in $[0,m_ip_u]$.
Since we set $m_i \gets m_ip_u$ at the end of the procedure, thus $r_i$ is a uniform random real number in $[0,m_i]$ after the procedure $\mathsf{Couple}(u)$, and~\eqref{eq-property-proof} still holds.

To prove the validity of the implementation. First note that if $D_i = 1$, the validity holds trivially. If $D_i = 0$, by~\eqref{eq-property-proof}, $r_i$ is a uniform random real number in $[0,M_i]$. Thus $r_i > M_i p_u$ with probability $1-p_u$, and $r_i \leq M_ip_u$ with probability $p_u$. By \Cref{lemma-coupling-uniform}, in the optimal coupling, the event $X^{\cadp}_u = c_u \lor Y^{\cadp}_u = c_u $ has probability $p_u$. Thus, the validity holds due to the chain rule.

Next, for hyperedge $S_i=e$, we define the following bad event
\begin{align}
\label{eq-def-B-edge}
\+B_i:\quad r_i \leq 	\frac{1}{8d^3k^3}.
\end{align}
We show that if the hyperedge $S_i = e$ fails, then $\+B_i$ must occur.

Suppose $S_i = e$ is a hyperedge.
Consider the input CSP formulas $\Phi^{X} = (V, (Q^{X}_u)_{u \in V}, \Cons{C})$ and $\Phi^{Y} = (V, (Q^{Y}_u)_{u \in V},\Cons{C})$. For any $u \neq v_0$, let $q_u = \abs{Q^{X}_u} =\abs{Q^{Y}_u}$.
Suppose $e$ represents the atomic constraint $c$ such that $c(\sigma) = \False$ for some unique $\sigma \in [q]^{e}$.
Suppose after the coupling procedure $\cadp$, variables $u_1,u_2,\ldots,u_m \in \Vcol \cap e$.
Since the hyperedge $S_i$ fails, it holds that
\begin{itemize}
\item after the coupling procedure, $\vol_{\Phi^{X}}(e) < \gamma$ or $\vol_{\Phi^{Y}}(e) < \gamma$;
\item for any $1\leq i \leq m$, $X^{\cadp}_{u_i} = \sigma_{u_i}$ or $Y^{\cadp}_{u_i} = \sigma_{u_i}$. 	
\end{itemize}
The second property holds because otherwise $e$ is satisfied by both $\Ass{X}^{\cadp}$ and $\Ass{Y}^{\cadp}$, thus must be removed by the coupling. According to our implementation, at the end of the coupling, we have
\begin{align*}
D_i = 0 \quad\text{and}\quad r_i \leq M_i = \prod_{j = 1}^m p_{u_j}. 	
\end{align*}
Note that $s > 32$ ($s$ is defined in \eqref{eq-def-s-in-gamma}), $m \leq k$ because $\abs{e} = k$.
By \Cref{lemma-coupling-uniform}, we have
\begin{align*}
\prod_{j = 1}^m p_{u_j} \leq \prod_{j=1}^m \frac{1}{q_{u_j}}\tp{1 + \frac{4}{s k}} \leq \exp\tp{\frac{4m}{s k}}\prod_{j=1}^m \frac{1}{q_{u_j}}\leq \mathrm{e}\prod_{j=1}^m \frac{1}{q_{u_j}}.
\end{align*}
At the end of the coupling, we have $\vol_{\Phi^{X}}(e) < \gamma$ or $\vol_{\Phi^{Y}}(e) < \gamma$. But in the beginning of the coupling, by~\eqref{eq-initial-vol}, we have $\vol_{\Phi^{X}}(e) \geq {3000q^2d^6k^6} $ and $\vol_{\Phi^{Y}}(e)\geq {3000 q^2d^6k^6}$. The volume of $e$ decreases because we update $\Phi^{X}$ and $\Phi^{Y}$ in \Cref{line-update-formula} for $u = u_1,u_2,\ldots,u_m$. Note that $v_0 \notin \Vcol$, thus $u_j \neq v_0$ for all $1\leq j \leq m$. We have
\begin{align*}
\prod_{j = 1}^m q_{u_j} \geq \frac{3000 q^2d^6k^6}{ \gamma} = \frac{3000q^2d^6k^6}{32\mathrm{e}q^2 d^3k^3} = \frac{3000d^3 k ^3}{32\mathrm{e}}	.
\end{align*}
If the hyperedge $S_i$ fails, then it holds that
\begin{align*}
r_i \leq \prod_{j = 1}^m p_{u_j} \leq \mathrm{e}\prod_{j=1}^m \frac{1}{q_{u_j}} \leq \frac{32\mathrm{e}^2 }{3000d^3k^3} \leq \frac{1}{8d^3k^3}.
\end{align*}
Thus the event $\+B_i$ must occur.

Combining two cases together, we have
\begin{align*}
\Pr[\cadp]{\forall 1 \leq i \leq \ell, S_{i} \text{ fails}} &\leq \Pr{\forall 1 \leq i \leq \ell, \+B_i}\\
\text{(all $r_i$ are mutually independent)}\quad &\leq \prod_{i=1}^\ell\Pr{\+B_i}\\
\text{(by~\eqref{eq-def-B-variable} and~\eqref{eq-def-B-edge})}\quad &\leq \prod_{\substack{1\leq i\leq \ell\\ \text{$S_i$ contains a single variable}}}\frac{1}{8d^2k^3}\ \prod_{\substack{1\leq i\leq \ell\\ \text{$S_i$ is a hyperedge}}}\frac{1}{8d^3k^3}. &\qedhere
\end{align*}
\end{proof}

Recall a sequence of sets $S_1,S_2,\ldots,S_\ell$ is called   a \emph{percolation sequence (PS)} to $u \in V$ if it satisfies first three properties in \Cref{lemma-path-strong}.
We call a  sequence of sets $S_1,S_2,\ldots,S_\ell$   a \emph{percolation sequence (PS)} if it satisfies first three properties in \Cref{lemma-path-strong} except $\dist_H(u,s_\ell) = 0$. 
For any $S_i$, let
\begin{align}
\label{eq-def-pfail}
p_{\mathrm{fail}}(S_i) = \begin{cases}
  \frac{1}{8d^2k^3} &\text{if $S_i$ contains a single variable};\\
  \frac{1}{8d^3k^3} &\text{if $S_i$ is a hyperedge}.	
 \end{cases}
 \end{align}
 Combining~\eqref{eq-bound-Iv-path} and \Cref{lemma-failed-path}, we have
\begin{align*}
I_{\vst} \leq \sum_{\text{PS for }\vst: e_1,e_2,\ldots,e_{\ell}}\Pr[\cadp]{\forall 1 \leq i \leq \ell, S_{i} \text{ fails}}\leq 	\sum_{\text{PS for }\vst: e_1,e_2,\ldots,e_{\ell}} \prod_{i=1}^{\ell} p_{\mathrm{fail}}(S_i).
\end{align*}
Note that the hypergraph $H$ is same for any $\vst \in V \setminus \{v_0\}$.
We can use the above inequality with $\vst = v$ for all $v \in V \setminus \{v_0\}$. This implies
\begin{align*}
\sum_{v \in V:v \neq v_0}I_v \leq \sum_{v\in V: v\neq v_0}\sum_{\text{PS to $v$}:S_1,S_2,\ldots,S_\ell} \prod_{1\leq i\leq \ell}p_{\mathrm{fail}}(S_i)\leq k\sum_{\text{PS}:S_1,S_2,\ldots,S_\ell}  \prod_{1\leq i\leq \ell}p_{\mathrm{fail}}(S_i),
\end{align*}
where the last inequality holds because there are at most $k$ variables $v$ that satisfies $\dist(v,S_{\ell}) = 0$ (if $S_{\ell}$ contains a single variable, there are only one variable $v$; if $S_{\ell}$ is a hyperedge, there are $k$ variables $v$). 
We can enumerate all the PSs according the length. We have
\begin{align*}
\sum_{v \in V:v \neq v_0}I_v  \leq 	k\sum_{\ell = 1}^\infty \sum_{\substack{\text{PS of length $\ell$}\\S_1,S_2,\ldots,S_\ell}} \prod_{1\leq i\leq \ell}p_{\mathrm{fail}}(S_i) = k\sum_{\ell = 1}^\infty N(\ell),
\end{align*}
where 
\begin{align*}
N(\ell) \triangleq \sum_{\substack{\text{PS of length $\ell$}\\S_1,S_2,\ldots,S_\ell}} \prod_{1\leq i\leq \ell}p_{\mathrm{fail}}(S_i).
\end{align*}
We then show that
\begin{align}
\label{eq-bound-N-l}
N(\ell) \leq \tp{k^2d^2  \frac{1}{8d^2k^3} + k^2d^3 \frac{1}{8d^3k^3}}	\tp{k^3d^2  \frac{1}{8d^2k^3} + k^3d^3   \frac{1}{8d^3k^3} }^{\ell - 1}.	
\end{align}
We need the following basic facts to prove~\eqref{eq-bound-N-l}.
%Fix a length $\ell \geq 1$, let $N(\ell)$ denote the  number of PSs of length $\ell$.
We may assume $d,k\geq 2$, otherwise the sampling problem is trivial. 
Fix a variable $v \in V$. 
The number of variables $u$ satisfying $\dist_H(v,u) \leq 2$ is at most
\begin{align*}
1 + d(k-1) + d(d-1)(k-1)^2 \leq k^2d^2.	
\end{align*}
The number of hyperedges $e'$ satisfying $\dist_H(v, e') \leq 2$ is at most
\begin{align*}
d + d(k-1)(d-1) + d(d-1)^2(k-1)^2 \leq k^2d^3. 	
\end{align*}
Fix a hyperedge $e \in \+E$. The number of variables $u$ satisfying $\dist_H(e,u) \leq 2$ is at most 
\begin{align*}
k + k(d-1)(k-1) + k(d-1)^2(k-1)^2 \leq k^3d^2.	
\end{align*}
The number of hyperedges $e'$ satisfying $\dist_H(e, e') \leq 2$ is at most
\begin{align*}
(1 + k(d-1)) + k(k-1)(d-1)^2 + k(k-1)^2(d-1)^3 \leq k^3d^3. 	
\end{align*}

We prove~\eqref{eq-bound-N-l} by induction on $\ell$. Suppose $\ell = 1$. It holds that $\dist_H(v_0, S_1) \leq 2$. By~\eqref{eq-def-pfail}, we have
\begin{align*}
N(1) \leq k^2d^2  \frac{1}{8d^2k^3}+ k^2d^3   \frac{1}{8d^3k^3}. 	
\end{align*}
Suppose~\eqref{eq-bound-N-l} holds for all $\ell \leq k$. We prove~\eqref{eq-bound-N-l} for $\ell = k + 1$. 
For PS $S_1,S_2,\ldots,S_{k+1}$ of length $k+1$, $S_1,S_2,\ldots,S_k$ is a PS of length $k$ and $\dist_H(S_k,S_{k+1}) \leq 2$.
For any $S_k$, there are at most $k^3d^2$ ways to choose $S_{k+1}$ as a variable, and at most $k^3d^3$ ways to choose $S_{k+1}$ as a hyperedge. This implies
\begin{align*}
N(k+1) &\leq N(k)\tp{k^3d^2  \frac{1}{8d^2k^3} + k^3d^3   \frac{1}{8d^3k^3}}\\
&\overset{\text{by I.H.}}{\leq} \tp{k^2d^2  \frac{1}{8d^2k^3} + k^2d^3   \frac{1}{8d^3k^3}}	\tp{k^3d^2  \frac{1}{8d^2k^3} + k^3d^3   \frac{1}{8d^3k^3}}^{k}.  	
\end{align*}
This proves~\eqref{eq-bound-N-l}. Now, we have
\begin{align*}
\sum_{v \in V:v \neq v_0}I_v \leq 	k\sum_{\ell = 1}^\infty N(\ell) \leq \sum_{\ell = 1}^\infty \tp{k^3d^2  \frac{1}{8d^2k^3} + k^3d^3   \frac{1}{8d^3k^3}}^{\ell} = \sum_{\ell = 1}^\infty \tp{\frac{1}{4}}^\ell \leq \frac{1}{2}.
\end{align*}

\subsection{Non-adaptive coupling analysis}
\label{section-non-adp}
We now analyze the general CSP formula $\Phi=(V,\Dom{Q},\Cons{C})$ with atomic constraints, where each variable $v \in V$ has an arbitrary domain $Q_v$ and each constraint contains arbitrary number of variables.
We will prove the following lemma is this section.
\begin{lemma}
\label{lemma-path-coupling-gen}
Let $\Phi=(V,\Dom{Q},\Cons{C})$ be the input CSP formula with atomic constraints in \Cref{alg-mcmc}.
Let $\Proj{h} = (h_v)_{v \in V}$ be the projection scheme for $\Phi$ satisfying \Cref{condition-projection} with parameters $\alpha$ and $\beta$.
Let $q_v = \abs{Q_v}$, $p = \max_{c \in \Cons{C}}\prod_{v \in \vbl{c}}\frac{1}{q_v}$ and $D$ denote the maximum degree of the dependency graph of $\Phi$. If
\begin{align*}
\log \frac{1}{p} \geq \frac{50}{\beta} \log \tp{\frac{2000D^4}{\beta}},	
\end{align*}
then it holds that $\sum_{v \in V \setminus \{v_0\}}\DTV{\nu_{v}^{X_{V \setminus \{v\}}}}{\nu_v^{Y_{V \setminus \{v\}}}}	\leq \frac{1}{2}$.
\end{lemma}

Fix a variable $\vst \in V \setminus \{v_0\}$. The goal of this section is to construct a non-adaptive coupling $\cnon$ to bound the total variation distance $\DTV{\nu_{\vst}^{X_{V\setminus \{\vst\} }}}{\nu_{\vst}^{Y_{V\setminus \{\vst\} }}}$.

Recall that $\Phi=(V,\Dom{Q},\Cons{C})$ is the original input CSP formula.
Recall that two CSP formulas $\Phi^{X}=(V,\Dom{Q}^{X}=(Q^{X}_u)_{u \in V},\Cons{C})$  and $\Phi^{Y}=(V,\Dom{Q}^{Y}=(Q^{Y}_v)_{v \in V},\Cons{C})$  are defined by
\begin{align}
\label{eq-def-QX-QY-gen}
Q^{X}_u = \begin{cases}
h^{-1}_u(X_u) &\text{if } u \neq \vst;\\
Q_u &\text{if } u = \vst, 	
 \end{cases}
 \qquad\qquad
 Q^{Y}_u = \begin{cases}
h^{-1}_u(Y_u) &\text{if } u \neq \vst;\\
Q_u &\text{if } u = \vst. 	
\end{cases}
\end{align}
By definition, $(Q^{X}_u)_{u \in V}$ and $(Q^{Y}_u)_{u \in V}$ differ only at variable $v_0$. 
Let $\mu_{\Phi^{X}}$ denote the uniform distribution over all satisfying assignments to $\Phi^{X}$, and 
$\mu_{\Phi^{Y}}$ denote the uniform distribution over all satisfying assignments to $\Phi^{Y}$. 
The first step for non-adaptive coupling analysis is to construct another projection schemes on instances $\Phi^{X}$ and $\Phi^{Y}$. 
Let $\PX=(h^X_v)_{v \in V}$ denote the projection scheme for $\Phi^{X}$ and $\PY = (h^Y_v)_{v \in V}$ denote the projection scheme for $\Phi^{Y}$, where $h^X_v:Q^{X}_v \to \Sigma_v^X$ and $h^Y_v:Q^{Y}_v\to \Sigma_v^Y$. 
For each $v \in V$, define
\begin{align*}
s_v^X \triangleq \abs{\Sigma_v^X},\quad s_v^Y \triangleq \abs{\Sigma_v^Y}, \quad q^X_v = \abs{Q^{X}_v}, \quad q^Y_v = \abs{Q^{Y}_v}.	
\end{align*}

In our analysis, we construct a pair of projection schemes $\PX,\PY$ satisfying the following condition.
\begin{condition}
\label{condition-projection-coupling}
Let $\Phi=(V,\Dom{Q},\Cons{C})$ be the original input CSP formula of \Cref{alg-mcmc} and $\Proj{h}=(h_v)_{v \in V}$ be the original projection scheme for $\Phi$ satisfying \Cref{condition-projection} with parameters $\alpha$ and $\beta$.
The projection scheme 	$\PX$ for $\Phi^{X}$  and the projection scheme $\PY$ for $\Phi^{Y}$ satisfy the following conditions:
\begin{itemize}
\item both $\PX$ and $\PY$ are balanced, i.e. for each $v \in V$ and $c_v^X \in \Sigma^X_v$, $\ftp{q^X_v/s^X_v} \leq \abs{(h^X_v)^{-1}(c_v^X)}\leq \ctp{q^X_v / s^X_v}$; for each $v \in V$ and $c_v^Y \in \Sigma^Y_v$, $\ftp{q^Y_v/s^Y_v} \leq \abs{(h^Y_v)^{-1}(c_v^Y)}\leq \ctp{q^Y_v / s^Y_v}$;
\item $\Sigma^X_{v_0} = \Sigma^Y_{v_0}$; and $h^X_u = h^Y_u$ for all $u \in V \setminus \{v_0\}$;
\item $h^X_{\vst} = h^Y_{\vst} = h_{\vst}$, where $h_{\vst}$ is the original projection scheme $\Proj{h}$ restricted on variable $\vst$;
\item for any constraint $c \in \Cons{C}$, 
\begin{align}
\label{eq-condition-upper-bound}
\min \tp{\sum_{v \in \vbl{c}}\log\ftp{\frac{q_v^X}{s_v^X}},\sum_{v \in \vbl{c}}\log\ftp{\frac{q_v^Y}{s_v^Y}}} \geq \frac{\beta}{10} \tp{\sum_{v \in \vbl{c}}\log q_v};
\end{align}
for any constraint $c \in \Cons{C}$ satisfying $\vst \notin \vbl{c}$,
\begin{align}
\label{eq-condition-lower-bound}
\min\tp{\sum_{v \in \vbl{c}}\log\frac{q_v^X}{\ctp{q_v^X / s_v^X}},\sum_{v \in \vbl{c}}\log\frac{q_v^Y}{\ctp{q_v^Y / s_v^Y}}}	 \geq \frac{\beta}{10} \tp{\sum_{v \in \vbl{c}}\log q_v};
\end{align}
for any constraint $c \in \Cons{C}$ satisfying $\vst \in \vbl{c}$,
\begin{align}
\label{eq-condition-coupling-projection}
&\min\tp{\log \ftp{\frac{ q^X_{\vst}}{s^X_{\vst}}} + \sum_{v \in \vbl{c} \setminus \{\vst\} }\log\frac{q_v^X}{\ctp{q_v^X / s_v^X}},\,\, \log\ftp{\frac{ q^Y_{\vst}}{s^Y_{\vst}}} + \sum_{v \in \vbl{c} \setminus \{\vst\} }\log\frac{q_v^Y}{\ctp{q_v^Y / s_v^Y}}}\notag\\
&\,\geq \frac{\beta}{10} \tp{\sum_{v \in \vbl{c}}\log q_v},
\end{align}
where $q^X_v = \abs{Q^{X}_v}$,$q^Y_v = \abs{Q^{Y}_v}$  and $q_v = \abs{Q_v}$ for all $v \in V$.
\end{itemize}
\end{condition}
\Cref{condition-projection-coupling} is a variation of \Cref{condition-projection}. The lower bound in~\eqref{eq-condition-lower-bound} can be transformed to the upper bounds on $\sum_{v \in \vbl{c}} \ctp{{q_v^X}/{s_v^X}}$ and $\sum_{v \in \vbl{c}} \ctp{{q_v^Y}/{s_v^Y}}$.
Thus,~\eqref{eq-condition-lower-bound} and~\eqref{eq-condition-upper-bound} are similar to~\eqref{eq:entropy-lower-bound} and~\eqref{eq:entropy-upper-bound} in \Cref{condition-projection}.
Moreover, for constraint $c \in \Cons{C}$ satisfying $\vst \in \vbl{c}$, we need an extra condition in~\eqref{eq-condition-coupling-projection}. The purpose of this extra condition is to handle the case that $\abs{\vbl{c}}$ can be very large.

The following lemma shows that the projection schemes satisfying \Cref{condition-projection-coupling} exist under a Lov\'asz local lemma condition. Since we only use $\PX$ and $\PY$ for analysis, we only need to show such projection schemes exist, we do not need an algorithm to construct  specific projection schemes.
\begin{lemma}
\label{lemma-find-projection-coupling}
%Let $\gamma$ and $\lambda$ be two parameters such that $\lambda + \gamma < 1$.
Let $\Phi=(V,\Dom{Q},\Cons{C})$ be the original input CSP formula of \Cref{alg-mcmc} and $\Proj{h}=(h_v)_{v \in V}$ be the original projection scheme for $\Phi$ satisfying \Cref{condition-projection} with parameters $\alpha$ and $\beta$.
Let $q_v = \abs{Q_v}$ and $D$ denote the maximum degree of the dependency graph of $\Phi$. 
Let $p \triangleq \max_{c \in \Cons{C}} \prod_{v \in \vbl{c}}\frac{1}{q_v}$.
Suppose
\begin{align*}
 \log \frac{1}{p} \geq \frac{55}{\beta}(\log D + 3).
\end{align*}
There exist  projection schemes $\PX,\PY$ for $\Phi^{X},\Phi^{Y}$ satisfying \Cref{condition-projection-coupling}.
\end{lemma}
\noindent The proof of \Cref{lemma-find-projection-coupling} is deferred to \Cref{section-proof-find-coupling-gen}.

Let $\PX=(h^X_v)_{v \in V}$ and  $\PY = (h^Y_v)_{v \in V}$ denote the projection schemes for $\Phi^{X}$ and $\Phi^{Y}$, where $h^X_v:Q^{X}_v \to \Sigma^X_v$ and $h^Y_v:Q^{Y}_v \to \Sigma^Y_v$. 
Suppose $\PX$ and $\PY$ satisfy \Cref{condition-projection-coupling}.
By \Cref{condition-projection-coupling}, for any variable $v \in V$, $\Sigma^X_v = \Sigma^Y_v$ and $s^X_v = s^Y_v = \abs{\Sigma_v^X} = \abs{\Sigma_v^Y}$. Denote
\begin{align*}
\forall v \in V, \quad s'_v &\triangleq s^X_v = s^Y_v \text{ and } \Sigma_v' \triangleq \Sigma^X_v = \Sigma^Y_v;\\
\Sigma' &\triangleq \bigotimes_{v \in V}\Sigma'_v.
\end{align*}
Recall $\mu_{\Phi^{X}}$ and $\mu_{\Phi^{Y}}$ are the uniform distributions over all satisfying assignments to $\Phi^{X}$ and $\Phi^{Y}$.
We define the following two projected distributions:
\begin{itemize}
\item $\nu_{X}$: the projected distribution (defined in \Cref{def:proj-distr}) over $\Sigma' = \bigotimes_{v \in V}\Sigma'_v$ induced from the instance $\Phi^{X}$ and the projection scheme $\Proj{h}^X$;
\item $\nu_{Y}$: the projected distribution (defined in \Cref{def:proj-distr}) over $\Sigma' = \bigotimes_{v \in V}\Sigma'_v$ induced from the instance $\Phi^{Y}$ and the projection scheme $\Proj{h}^Y$. 
\end{itemize}
For any variable $v \in V$, let $\nu_{v,X}$ and $\nu_{v,Y}$ denote the marginal distributions on $v$ projected from $\nu_{X}$ and $\nu_Y$. Recall the goal of this section is to bound $\DTV{\nu_{\vst}^{X_{V\setminus \{\vst\} }}}{\nu_{\vst}^{Y_{V\setminus \{\vst\} }}}$. By \Cref{condition-projection-coupling}, $h^X_{\vst} = h^Y_{\vst} = h_{\vst}$. By the definitions $\Phi^{X}$, $\Phi^{Y}$ and the projected distribution in~\Cref{def:proj-distr},
\begin{align*}
\nu_{\vst}^{X_{V\setminus \{\vst\} }} = \nu_{\vst, X} \quad\text{and}\quad 	\nu_{\vst}^{Y_{V\setminus \{\vst\} }} = \nu_{\vst,Y}.
\end{align*}

Recall that $\Phi = (V, \Dom{Q}, \Cons{C})$ is the original input CSP formula of \Cref{alg-mcmc}.
Recall that $H=(V,\+E)$ denotes the (multi-)hypergraph that models $\Phi$, where $\+E \triangleq \{\vbl{c} \mid c \in \+C\}$.
Note that $H$ also models $\Phi^{X}$ and $\Phi^{Y}$, because  $\Phi,\Phi^{X},\Phi^{Y}$ have the same sets of variables and constraints.
Let $e \in \+E$ be a hyperedge and $u \in e$ a variable in $e$.
Let $X_u^{\cnon},Y_u^{\cnon} \in \Sigma'_u$ be two values.
Let $c_e \in \Cons{C}$ denote the atomic constraint represented by $e$.
Let $\sigma \in \Dom{Q}_{e}$ denote the unique configuration forbidden by $c_e$, i.e. $c_e(\sigma)=\False$.
We say $e$ is satisfied by $X^{\cnon}_u$ if $\sigma_u \notin (h^X_u)^{-1}(X_u^{\cnon})$, because in the projected distribution $\nu_X$, conditional on the value of $u$ is $X^{\cnon}_u$, the constraint $c_e$ must be satisfied.
Similarly, We say $e$ is satisfied by $Y^{\cnon}_u$ if $\sigma_u \notin (h^Y_u)^{-1}(Y_u^{\cnon})$.
The coupling procedure $\cnon$ is given in \Cref{alg-coupling-gen}.
\begin{algorithm}[h]
  \SetKwInOut{Input}{Input} \SetKwInOut{Output}{Output} 
  \Input{CSP formulas $\Phi^{X}=(V,\Dom{Q}^{X}=(Q^{X}_u)_{u \in V},\Cons{C})$  and $\Phi^{Y}=(V,\Dom{Q}^{Y}=(Q^{Y}_v)_{v \in V},\Cons{C})$, the hypergraph
    $H= (V,\mathcal{E})$ modeling $\Phi^{X}$ and $\Phi^{Y}$, projection schemes $\Proj{h}^X$ and $\Proj{h}^Y$ satisfying \Cref{condition-projection-coupling}, variables $v_0, \vst \in V$, an index function $\mathrm{ID}: V \to [n]$ such that  $\mathrm{ID}(u) \neq \mathrm{ID}(v)$ for all $u \neq v$ and $\mathrm{ID}(\vst) = n$.}  
  \Output{a pair of assignments $\Ass{X}^{\cnon},\Ass{Y}^{\cnon} \in \Sigma'$.}
  % $X_{\+C_v} \gets X(\+M_v)$ and $Y_{\+C_v} \gets Y(\+M_v)$, where $\+M_v = \+M \setminus\{v\}$ \;
  %$X^{\cnon}(v_0) = 0$ and $Y^{\cnon}(v_0) = 1$\; 
  sample $X^{\cnon}_{v_0} \sim \nu_{v_0,X}$ and $Y^{\cnon}_{v_0} \sim \nu_{v_0,Y}$ independently\label{line-v0-gen}\;
    $V_1 \gets \{v_0\}$, $V_2 \gets V \setminus V_1$, $\Vcol \gets \{v_0\}$\; 
%  \For{each $v \in V \setminus \{v_0\}$ such that $m'_v = 1$}{
%  	$X^{\cnon}_v \gets 1$, $Y^{\cnon}_v \gets 1$, $\Vcol \gets \Vcol \cup \{v\}$\label{line-set-value-1}\;
%  }
  remove all $e$ from $\+E$ s.t. the constraint $c$ represented by $e$ is satisfied by both $X^{\cnon}_{v_0}$ and $Y^{\cnon}_{v_0}$\label{line-remove-e-v0}\;
  \While{$\exists e \in \mathcal{E}$ s.t.\ $e \cap V_1 \neq \emptyset, (e \cap V_2) \setminus \Vcol \neq \emptyset$ }
    {let $e$ be the first such hyperedge and $u$ the variable in $(e \cap V_2) \setminus \Vcol$ with lowest ID\label{line-pick-gen}\;
    sample $(c_X,c_Y) \in \Sigma'_u \times \Sigma'_u$ from the optimal coupling between $\nu_{u,X}(\cdot \mid \Ass{X}^{\cnon})$ and $\nu_{u,Y}(\cdot \mid \Ass{Y}^{\cnon})$ and extend $\Ass{X}^{\cnon}$ and $\Ass{Y}^{\cnon}$ to $u$ by setting $(X^{\cnon}_u, Y^{\cnon}_u) \gets (c_X,c_Y)$\label{line-optimal-coupling-gen}\;
    $\Vcol \gets \Vcol \cup \{u\}$\; 
    \If{$X^{\cnon}_u \neq Y^{\cnon}_u$}{$V_1 \gets V_1 \cup \{u\}, V_2 \gets V \setminus V_1$\label{line-add-v1-gen}\;}
    \For{$e \in \+E$ s.t. the constraint $c$ represented by $e$ is satisfied by both $X^{\cnon}_u$ and $Y^{\cnon}_u$}
      {$\mathcal{E} \gets\mathcal{E} \setminus \{e\}$ \label{line-remove-e-gen}} 
	\For{$e \in \+E$ s.t. $e \subseteq \Vcol$\label{line-add-set-cond-gen}}
      {$V_1 \gets V_1 \cup \{e\}, V_2 \gets V \setminus V_1$\label{line-add-set-gen}\;} 
    } 
    %extend $X^{\cnon}$ and $Y^{\cnon}$ to the set $V_2 \setminus \Vcol$ 
%    s.t.\ $X^{\cnon}(V_2 \setminus \Vcol)$ and $Y^{\cnon}(V_2 \setminus \Vcol)$ is drawn from 
   %by sampling $(X^{\cnon}_{V_2 \setminus \Vcol}, Y^{\cnon}_{V_2 \setminus \Vcol})$ from the optimal coupling between $\mu_{V_2 \setminus \Vcol, \Phi^{X}}$ and $\mu_{V_2 \setminus \Vcol,\Phi^{Y}}$\label{line-sample-V2-Cv}\;
    extend $\Ass{X}^{\cnon}$ and $\Ass{Y}^{\cnon}$ to the set $V_2 \setminus \Vcol$ 
%    s.t.\ $X^{\cnon}(V_1 \setminus \Vcol)$ and $Y^{\cnon}(V_1 \setminus \Vcol)$ is drawn from 
    by sampling $(X^{\cnon}_{V_2 \setminus \Vcol}, Y^{\cnon}_{V_2 \setminus \Vcol})$ from the optimal coupling between $\nu_{V_2 \setminus \Vcol,X}(\cdot\mid \Ass{X}^{\cnon})$ and $\nu_{V_2 \setminus \Vcol,Y}(\cdot\mid \Ass{Y}^{\cnon})$\label{line-extend-gen}\;
    \Return{$(\Ass{X}^{\cnon}, \Ass{Y}^{\cnon})$\;}
  \caption{The coupling procedure $\cnon$}\label{alg-coupling-gen}
\end{algorithm}

The input of the coupling $\cnon$ contains CSP formulas $\Phi^{X} $ and $\Phi^{Y}$, together with projection schemes $\Proj{h}^X$ and $\Proj{h}^Y$ satisfying \Cref{condition-projection-coupling}. 
We also give an index function $\mathrm{ID}: V \to [n]$ such that each variable has a distinct index and the variable $\vst$ has the largest index. The coupling will use this index to pick the variable  in \Cref{line-pick-gen}. 
Compared with the adaptive coupling in \Cref{alg-coupling-v}, the coupling $\cnon$ is non-adaptive, i.e. it does not need to maintain the current volume of each hyperedge. 
Instead, the coupling $\cnon$ is given  projection schemes $\Proj{h}^X$ and $\Proj{h}^Y$ in advance. 
Once the coupling $\cnon$ picks a variable $u$, it assigns the values in $\Sigma'_u$ to variable $u$, where the domain $\Sigma'_u$ is determined by $\Proj{h}^X$ and $\Proj{h}^Y$.
The coupling $\cnon$ will put $u$ into $V_1$ if the coupling on $u$ fails.
After that, the coupling will remove all the hyperedges satisfied by both $X^{\cnon}_u$ and $Y^{\cnon}_u$ in \Cref{line-remove-e-gen}. 
If all variables in a hyperedge $e$ are assigned values  and $e$ is still not satisfied, the coupling $\cnon$ puts $e$ into $V_1$ in \Cref{line-add-set-gen}.
Remark that after the while-loop, $\cnon$ only samples the value for $V_2 \setminus \Vcol$ because $V_1 \subseteq \Vcol$.
%The coupling $\cnon$ satisfies the following properties.
\begin{lemma}
\label{lemma-coupling-gen}	
The coupling procedure $\cnon$  satisfies the following properties:
\begin{itemize}
\item the coupling procedure will terminate eventually;
\item the output $\Ass{X}^{\cnon} \in \Sigma'$ follows $\nu_X$ and the output $\Ass{Y}^{\cnon} \in \Sigma'$ follows $\nu_{Y}$;
\item for any variable $u \in V$, if $X^{\cnon}_u \neq Y^{\cnon}_u$ in the final output, then $u \in V_1$.
\end{itemize}
\end{lemma}
\begin{proof}
After each execution of the while-loop, the size of $\Vcol$ will increase by $1$. The size of $\Vcol$ is at most $n$. Thus, the coupling procedure will terminate eventually.

We prove the second property for $\Ass{X}^{\cnon}$. The result for $\Ass{Y}^{\cnon}$ can be proved in a similar way.
In \Cref{line-v0-gen}, the coupling samples the $X^{\cnon}_{v_0}$ independently from the distribution $\nu_{v_0,X}$.  Given the current configuration $\Ass{X}^{\cnon}$,  the coupling picks an unassigned variable $u$, then draw $X^{\cnon}_{u}$ from the conditional marginal distribution $\nu_{u,X}(\cdot \mid \Ass{X}^{\cnon})$ in \Cref{line-optimal-coupling-gen}. 
Finally, the coupling samples $X^{\cnon}_{V \setminus V_2}$ from the conditional distribution. 
Note that $V_1 \subseteq \Vcol$. When the coupling terminates, all variables $v \in V$ gets a value $X^{\cnon}_v \in\Sigma'_v$.
By the chain rule, the output $\Ass{X}^{\cnon} \in \Sigma'$ follows the law  $\nu_X$.

To prove the last property, we show that after the while loop, it holds that
\begin{itemize}
\item $X^{\cnon}_{V_2 \cap \Vcol} = Y^{\cnon}_{V_2 \cap \Vcol}$;
\item 	$\nu_{V_2 \setminus \Vcol,X}(\cdot\mid \Ass{X}^{\cnon})$ and $\nu_{V_2 \setminus \Vcol,Y}(\cdot\mid \Ass{Y}^{\cnon})$ are identical distributions, thus all variables in $V_2 \setminus \Vcol$ can be coupled perfectly.
\end{itemize}
Combining these two properties proves the last property in the lemma.
The first property is easy to verify, because if $X^{\cnon}_u \neq Y^{\cnon}_u$, then $u$ must be added into $V_1$ in \Cref{line-add-v1-gen}.
To prove the second property, we claim that, after the while-loop, there is no hyperedge $e \in \+E$ such that $e \cap V_1 \neq \emptyset$ and $e \cap V_2 \neq \emptyset$. Suppose such hyperedge $e$ exists. There are two possibilities for such hyperedge.
\begin{itemize}
\item 	$(e \cap V_2) \setminus \Vcol \neq \emptyset$: In this case, the while-loop cannot terminate. Contradiction.
\item $(e \cap V_2) \setminus \Vcol = \emptyset$: Note that it always holds that $V_1 \subseteq \Vcol$. In this case, it holds that $e \subseteq \Vcol$. Note that $e \cap V_1 \neq  \emptyset$ and $e \cap V_2 \neq \emptyset$. Hence, after the \Cref{line-v0-gen}, there is no such hyperedge $e$. If such hyperedge $e$ exists, it must be produced by the while-loop.  Since $e \subseteq \Vcol$, such hyperedge $e$ will either be removed in \Cref{line-remove-e-gen}, or added into $V_1$ in \Cref{line-add-set-gen} (after which $e \cap V_2=\emptyset$). This implies that such hyperedge does not exist when the while-loop terminates. Contradiction.
\end{itemize}
Hence, after the while-loop, all variables are divided into two parts $V_1$ and $V_2$. Besides, all the constraints $c \in \Cons{C}$ such that $\vbl{c} \cap V_1 \neq \emptyset$ and $\vbl{c} \cap V_2 \neq \emptyset$ are satisfied by both $\Ass{X}^{\cnon}$ and $\Ass{Y}^{\cnon}$. This implies, conditional on $\Ass{X}^{\cnon}$, the variables in $V_2$ is independent with the variables in $V_1$, and the same result holds for $\Ass{Y}^{\cnon}$.  Note that two instances $\Phi^{X}$ and $\Phi^{Y}$ differ only at variable $v_0$, two projection schemes $\Proj{h}^X$ and $\Proj{h}^Y$ also differ only at $v_0$, and $v_0 \in V_1$. Since $X^{\cnon}_{V_2 \cap \Vcol} = Y^{\cnon}_{V_2 \cap \Vcol}$, $\nu_{V_2 \setminus \Vcol,X}(\cdot\mid \Ass{X}^{\cnon})=\nu_{V_2 \setminus \Vcol,X}(\cdot\mid X^{\cnon}_{V_2 \cap \Vcol})$ and $\nu_{V_2 \setminus \Vcol,Y}(\cdot\mid \Ass{Y}^{\cnon}) = \nu_{V_2 \setminus \Vcol,Y}(\cdot\mid Y^{\cnon}_{V_2 \cap \Vcol})$ are identical distributions.
\end{proof}

For each hyperedge $e \in \+E$, we say $e$ is failed in coupling $\cnon$ if the following condition holds.
\begin{definition}
\label{defintiion-fails-cnon}
A  hyperedge $e \in \+E$ fails in the coupling $\cnon$ if one of the following two events occur.
\begin{itemize}
\item \textbf{Type-I failure}: there is a variable $u \in e \setminus \{v_0\}$ such that the coupling picks $e$ and $u$ in \Cref{line-pick-gen}, and $X^{\cnon}_u \neq Y^{\cnon}_u$ after the coupling.
\item \textbf{Type-II failure}: consider the time when the while-loop terminates. It holds that $e \subseteq \Vcol$ and the constraint represented by $e$ is not satisfied by  both $\Ass{X}^{\cnon}$ and $\Ass{Y}^{\cnon}$.
\end{itemize}
\end{definition}
Let $\Lin(H)$ denote the line graph of $H$, where each vertex in $\Lin(H)$ is a hyperedge in $H$, two hyperedges $e,e' \in \+E$ are connected if $e \cap e' \neq \emptyset$. 
Let $\Lin^k(H)$ denote the $k$-th power graph of $\Lin(H)$, two hyperedges $e$ and $e'$ are adjacent in $\Lin^k(H)$ if their distance in $\Lin(H)$ is no more than $k$. For each variable, we use $N(v)$ to denote the set of hyperedges incident to $v$:
\begin{align*}
N(v) &\triangleq \{e \in \+E \mid v \in e\}.	
\end{align*}
For any $k \geq 1$, define
\begin{align}
\label{eq-def-Nkv}
N^k(v) & \triangleq \left\{e \in \+E \mid \exists e' \in N(v) \text{ s.t. } \dist_{\Lin(H)}(e,e') \leq k - 1 \right\},
\end{align}
where $\dist_{\Lin(H)}(e,e')$ denotes the length of the shortest path between $e$ and $e'$ in graph $\Lin(H)$. Remark that $N(v) = N^1(v)$ by definition. 

When the coupling $\cnon$ terminates, each variable $v \in V_1$ satisfies the following property.
\begin{lemma}
\label{lemma-path-gen}
For any $v \in V_1 \setminus \{v_0\}$, there exists a path $e_1,e_2,\ldots,e_\ell$ in $\Lin^2(H)$ such that
\begin{itemize}
\item $e_1 \in N^2(v_0)$ and $v \in e_\ell$;
\item for all $1\leq i \leq \ell$, the hyperedge $e_i$ fails in the coupling.
\end{itemize}
\end{lemma}
\begin{proof}
Let $V_1 = \{v_0, v_1,v_2,\ldots,v_m\}$ denote the variables in $V_1$, where $v_i$ is the $i$-th variables added into $V_1$. Remark that if a set of variables are added into $V_1$ at the same time (\Cref{line-add-set-gen}), we break tie arbitrarily. We prove the lemma by induction on index $i$.

The base case is $v_0$, the lemma holds for $v_0$ trivially. Suppose the lemma holds for $v_0,v_1,\ldots,v_{k-1}$. We prove the lemma for variable $v_{k}$. 
The variable $v_k$ is added into $V_1$ either in \Cref{line-add-v1-gen} or \Cref{line-add-set-gen}. 
\begin{itemize}
\item Suppose $v_k$ is added into $V_1$ in \Cref{line-add-v1-gen}. 
Variable $v_k$ must be picked in \Cref{line-pick-gen}.
Consider the hyperedge $e$ picked in \Cref{line-pick-gen}. The hyperedge $e$ fails in type-I because $v_k \in e$ and $X^{\cnon}_{v_k} \neq Y^{\cnon}_{v_k}$.
Besides, it holds that $v_k \in e$ and $v_j \in e$ for some $j < k$. If $j = 0$, the lemma holds trivially. If $0< j <k$, by induction hypothesis, there is a path $e_1,e_2,\ldots,e_{t}$ for $v_j$. Since $v_j \in e_t$ and $v_j \in e$, the lemma holds for $v_k$ with the path $e_1,e_2,\ldots,e_{t},e$.
\item Suppose $v_k$ is added into $V_1$ in \Cref{line-add-set-gen}. 
Let $e$ denote the hyperedge in \Cref{line-add-set-gen}. It holds that that $v_k \in e$. 
By \Cref{line-add-set-cond-gen}, $e \subseteq \Vcol$.
Since $e$ is not deleted in \Cref{line-remove-e-v0} or \Cref{line-remove-e-gen}, the constraint represented by $e$ is not satisfied by  both $\Ass{X}^{\cnon}$ and $\Ass{Y}^{\cnon}$.
This property holds up to the end of the coupling. Thus $e$ fails in type-II.
Since $e \subseteq \Vcol$ and $v_k \neq v_0$, the while-loop must have picked a hyperedge $e'$ and $v_k \in e'$ in \Cref{line-pick-gen}. 
Thus, $e'$ contains a variable $v_j$ for $j < k$ ($e'$ may not fail). 
If $j = 0$, then $e \in N^2(v_0)$, and the lemma holds for $v_k$ with single hyperedge $e$.
If $0< j < k$, by induction hypothesis, there is a path $e_1,e_2,\ldots,e_{t}$ for $v_j$. 
Since $e_t \cap e' \neq \emptyset$ and $e' \cap e \neq \emptyset$, $e$ and $e_t$ are adjacent in $\Lin^2(H)$. the lemma holds for $v_k$ with the path $e_1,e_2,\ldots,e_{t},e$.
\end{itemize}
Combining two cases proves the lemma.
\end{proof}
If the $X^{\cnon}_{\vst} \neq Y^{\cnon}_{\vst}$, we have the following result.
\begin{lemma}
\label{lemma-path-v-gen}
If $X^{\cnon}_{\vst} \neq Y^{\cnon}_{\vst}$, then there exists a path $e_1,e_2,\ldots,e_\ell$ in $\Lin^2(H)$ such that
\begin{itemize}
\item $e_1 \in N^2(v_0)$ and $\vst \in e_{\ell}$;
\item for all $1\leq i \leq \ell - 1$, the hyperedge $e_i$ fails in the coupling;
\item the hyperedge $e_{\ell}$ is not satisfied by both $X^{\cnon}_{S}$ and $Y^{\cnon}_{S}$, where $S = e_\ell \setminus \{\vst\}$.
\end{itemize}
\end{lemma}
\begin{proof}
If 	$X^{\cnon}_{\vst} \neq Y^{\cnon}_{\vst}$, by \Cref{lemma-coupling-gen}, it must hold that $\vst \in V_1$ and $\vst$ is added into $V_1$ in \Cref{line-add-v1-gen}, because $\vst \neq v_0$, and if  $\vst$ is added into $V_1$ in \Cref{line-add-set-gen}, then $X^{\cnon}_{\vst} = Y^{\cnon}_{\vst}$.
Consider the moment when $\vst$ is added into $V_1$. Suppose the while-loop picks the hyperedge $e_{\star}$. It must hold that $\vst \in e_{\star}$ and the while loop picks $\vst$ to sample its values in $\Ass{X}^{\cnon}$ and $\Ass{Y}^{\cnon}$. In \Cref{line-pick-gen}, the algorithm always picks the variable in $e_{\star}$ with lowest ID and the ID of $\vst$ is the $n$. This implies all $(e_{\star} \cap V_2) \setminus \Vcol = \{\vst\}$. Note that $V_1 \subseteq \Vcol$. Thus, all variables in $e_{\star} \setminus \{\vst\}$ get the value and $e_{\star}$ is not satisfied in both $X^{\cnon}_S$ and $Y^{\cnon}_S$, where $S = e_{\star} \setminus \{\vst\}$. Otherwise, $e_{\star}$ is removed in \Cref{line-remove-e-v0} or \Cref{line-remove-e-gen}, the while-loop cannot pick $e_{\star}$. 

Let $V_1 = \{v_0, v_1,v_2,\ldots,v_m\}$ denote the variables in $V_1$, where $v_i$ is the $i$-th variables added into $V_1$. Remark that if a set of variables are added into $V_1$ at the same time (\Cref{line-add-set-gen}), we break tie arbitrarily. Suppose $\vst = v_k$. Since $e_{\star}$ is picked in \Cref{line-pick-gen}, it must hold that $v_j \in e_{\star}$ for some $j < k$. If $j = 0$, the lemma holds with single hyperedge $e_{\star}$. If $0 < j < k$, there exists a path $e_1,e_2,\ldots,e_{\ell - 1}$ in $\Lin^2(H)$ satisfying the condition in \Cref{lemma-path-gen} for $v_j$. Since $v_j \in e_{\ell-1}$ and $v_j \in e_{\star}$, the lemma holds with the path $e_1,e_2,\ldots,e_{\ell - 1}, e_{\star}$. 
\end{proof}
We modify the path in \Cref{lemma-path-v-gen} to the following sequence of hyperedges, which will be used in the analysis.
\begin{corollary}
\label{corollary-path-v-gen}
If $X^{\cnon}_{\vst} \neq Y^{\cnon}_{\vst}$, then there exists a path $e_1,e_2,\ldots,e_\ell$ in $\Lin^3(H)$ such that
\begin{itemize}
\item $e_1 \in N^3 (v_0)$, $\vst \in e_{\ell}$, and $e_1,e_2,\ldots,e_\ell$  are mutually disjoint.
\item for all $1\leq i \leq \ell - 1$, the hyperedge $e_i$ fails in the coupling;
\item the hyperedge $e_{\ell}$ is not satisfied by both $X^{\cnon}_{S}$ and $Y^{\cnon}_{S}$, where $S = e_\ell \setminus \{\vst\}$.
\end{itemize}
\end{corollary}
\begin{proof}
Let $e'_1,e'_2,\ldots,e'_{m}$ denote the path in \Cref{lemma-path-v-gen}. We first show that how to construct the path $e_1,e_2,\ldots,e_{\ell}$ in $\Lin^3(H)$, then we show that such path satisfies the properties in the corollary.

Let $\+S$ be an empty stack. 
Let $P$ denote the sequence $(e'_1,e'_2,\ldots,e'_m)$.
We pick the last hyperedge in the path $P$, denote this hyperedge as $e'_i$.
We  push $e'_i$ into the stack $\+S$.
We search for the minimum index $j$ such that $j < i$ and $e'_i \cap e'_j \neq \emptyset$. Here are two cases depending on whether such index $j$ exists.
\begin{itemize}
\item If such index $j$ does not exist,  remove $e'_i$ from the path $P$.
\item If such index $j$ exists, remove all $e'_k$ for $j\leq k \leq i$ from the path $P$.
\end{itemize}
Repeat the above procedure until $P$ becomes an empty sequence.
Let $e_1,e_2,\ldots,e_{\ell}$ be the elements in stack $\+S$ from top to bottom.	

It is easy to verify $e_\ell =e'_m$. By \Cref{lemma-path-v-gen}, $\vst \in e_\ell$ and $e_\ell$ satisfies the last property in the corollary.
It is also easy to see all $e_1,e_2,\ldots,e_{\ell}$ are mutually disjoint.
By \Cref{lemma-path-v-gen}, the hyperedge $e_i$ fails in the coupling for all $1 \leq i \leq \ell - 1$.
We only need to prove the following two properties 
\begin{itemize}
\item $e_1 \in N^3(v_0)$;
\item $e_1,e_2,\ldots,e_\ell$ forms a path in $\Lin^3(H)$.	
\end{itemize}

We first prove $e_1 \in N^3(v_0)$. If $e_1 = e'_1$, then the property holds trivially.
Suppose $e_1 = e'_k$ for some $k > 1$. When the procedure adds $e'_k$ into the stack, the hyperedge $e'_1$ must be removed. This implies $e'_k \cap e'_1 \neq \emptyset$. By \Cref{lemma-path-v-gen}, $e'_1 \in N^2(v_0)$. It holds that $e_1 = e'_k \in N^3(v_0)$.

Next, we prove that  $e_1,e_2,\ldots,e_\ell$ forms a path in $\Lin^3(H)$. Consider two adjacent hyperedges $e_{i-1}$ and $e_{i}$. Suppose $e_i = e'_j$ and $e_{i-1} = e'_{k}$. If $j = k + 1$, since $e'_j$ and $e'_k$ are adjacent in $\Lin^2(H)$, $e_i$ and $e_{i-1}$ are adjacent in $\Lin^3(H)$. Suppose $j > k + 1$. In this case, $e'_{k+1}$ is removed and $e'_k$ is not removed, thus $e'_j \cap e'_{k+1} \neq \emptyset$. Since $e'_k$ and $e'_{k+1}$ are adjacent in $\Lin^2(H)$, $e'_j$ and $e'_k$ are adjacent in $\Lin^3(H)$.
\end{proof}

Fix a path $e_1,e_2,\ldots,e_{\ell}$ in $\Lin^3(H)$ such that it satisfies the first property except $\vst \in e_{\ell}$ in \Cref{corollary-path-v-gen}, i.e. $e_1 \in N^3(v_0)$, and $e_1,e_2,\ldots,e_{\ell}$  are mutually disjoint. 
We call such path a \emph{percolation path}~(PP). 
We say a percolation path $e_1,e_2,\ldots,e_{\ell}$ is a percolation path for $\vst$ if $\vst \in e_\ell$.
\begin{definition}
\label{definition-bad-edge-gen}
Fix a percolation path $e_1,e_2,\ldots,e_\ell$.
For each $1 \leq i \leq \ell$, a hyperedge $e_i$ is \emph{bad} if
\begin{itemize}
\item for $1\leq i \leq \ell - 1$: the hyperedge $e_i$ fails in the coupling $\cnon$ (\Cref{defintiion-fails-cnon});
\item for $i = \ell$: the hyperedge $e_{\ell}$ is not satisfied by both $X^{\cnon}_{S}$ and $Y^{\cnon}_{S}$, where $S = e_\ell \setminus \{\vst\}$; and $\vst$ is assigned  different values in $\Ass{X}^{\cnon}$ and $\Ass{Y}^{\cnon}$, i.e. $X^{\cnon}_{\vst} \neq Y^{\cnon}_{\vst}$.	
\end{itemize}
\end{definition}
By \Cref{corollary-path-v-gen}, if $X^{\cnon}_{\vst} \neq Y^{\cnon}_{\vst}$ in coupling $\cnon$, then there is a percolation path for $\vst$: $e_1,e_2,\ldots,e_{\ell}$ such that $e_i$ is bad for all $1 \leq i \leq \ell$.
We give the following key lemma in this proof.
\begin{lemma}
\label{lemma-PP}
Suppose the original input CSP formula $\Phi = (V,\Dom{Q},\Cons{C})$ of \Cref{alg-mcmc} satisfies
\begin{align}
\label{eq-cond-coupling-gen}
\log \frac{1}{p} \geq \frac{50}{\beta}\log\tp{\frac{2000 D^{4}}{\beta}}.		
\end{align}
Fix a percolation path (PP) $e_1,e_2,\ldots,e_{\ell}$ for $\vst$ in $\Lin^3(H)$	. It holds that
\begin{align*}
\Pr[\cnon]{\forall 1 \leq i \leq \ell, e_i \text{ is bad}}	 \leq  \tp{\frac{1}{4 D^3}}^{\ell} \frac{\beta}{50}\tp{\frac{1}{2}}^{\frac{\beta \abs{e_\ell}}{50}},
\end{align*}
which implies
\begin{align*}
\Pr[\cnon]{X^{\cnon}_{\vst} \neq Y^{\cnon}_{\vst} } \leq \sum_{e_1,e_2,\ldots e_{\ell} \text{ is a PP for }\vst}\tp{\frac{1}{4 D^3}}^{\ell} \frac{\beta}{50}\tp{\frac{1}{2}}^{\frac{\beta \abs{e_\ell}}{50}}.	
\end{align*}
\end{lemma}
The proof of \Cref{lemma-PP} is deferred to \Cref{section-proof-lemma-PP}. We now use \Cref{lemma-PP} to prove \Cref{lemma-path-coupling-gen}.

\begin{proof}[Proof of \Cref{lemma-path-coupling-gen}]
We will use \Cref{lemma-PP} to show that
\begin{align*}
\sum_{v \in V \setminus \{v_0\}}\DTV{\nu_{v}^{X_{V \setminus \{v\}}}}{\nu_{v}^{Y_{V \setminus \{v\}}}}	\leq \frac{1}{2}. 
\end{align*}
By the assumption in \Cref{lemma-path-coupling-gen}, it holds that $\log \frac{1}{p} \geq \frac{50}{\beta} \log \tp{\frac{2000D^4}{\beta}}$. Note that the condition in \Cref{lemma-PP} holds. 
Note that $\log \frac{1}{p} \geq \frac{50}{\beta} \log \tp{\frac{2000D^4}{\beta}}\geq\frac{55}{\beta}\tp{\log D + 3}.$
By \Cref{lemma-find-projection-coupling}, the projection schemes satisfying \Cref{condition-projection-coupling} exists.	
By \Cref{lemma-coupling-gen}, the $\Ass{X}^{\cnon}$ in $\cnon$ follows the distribution $\nu_{X}$ and  the $\Ass{Y}^{\cnon}$ in $\cnon$ follows the distribution $\nu_{Y}$. By the definition of $\nu_X$ and $\nu_Y$, it holds that $\nu_{\vst,X}=\nu_{\vst}^{X_{V \setminus \{\vst\}}}$ and $\nu_{\vst,Y}=\nu_{\vst}^{Y_{V \setminus \{\vst\}}}$. By the coupling lemma and \Cref{lemma-PP}, it holds that
\begin{align*}
\DTV{\nu_{\vst}^{X_{V \setminus \{\vst\}}}}{\nu_{\vst}^{Y_{V \setminus \{\vst\}}}}		\leq \Pr[\cnon]{X^{\cnon}_{\vst} \neq Y^{\cnon}_{\vst} }\leq \sum_{e_1,e_2,\ldots e_{\ell} \text{ is a PP for }\vst}\tp{\frac{1}{4 D^3}}^{\ell} \frac{\beta}{50}\tp{\frac{1}{2}}^{\frac{\beta \abs{e_\ell}}{50}}.
\end{align*}
Note that the hypergraph $H$ is same for any $\vst \in V \setminus \{v_0\}$.
We can use the above inequality with $\vst =v$ for all $v \in V \setminus \{v_0\}$. Thus,
\begin{align*}
\sum_{v \in V \setminus \{v_0\}}\DTV{\nu_{v}^{X_{V \setminus \{v\}}}}{\nu_{v}^{Y_{V \setminus \{v\}}}}	&\leq \sum_{v \in V \setminus \{v_0\}}\sum_{e_1,e_2,\ldots e_{\ell} \text{ is a PP for }v}\tp{\frac{1}{4 D^3}}^{\ell} \frac{\beta}{50}\tp{\frac{1}{2}}^{\frac{\beta \abs{e_\ell}}{50}}\\
\tp{\text{by double counting}}\quad&\leq \sum_{e_1,e_2,\ldots e_{\ell} \text{ is a PP}}\tp{\frac{1}{4 D^3}}^{\ell} \frac{\beta\abs{e_\ell}}{50}\tp{\frac{1}{2}}^{\frac{\beta \abs{e_\ell}}{50}}.
%&= \sum_{e_1,e_2,\ldots e_{\ell} \text{ is a PP}}\tp{\frac{1}{4 D^3}}^{\ell}\frac{\beta\abs{e_\ell}}{50}\tp{\frac{1}{2}}^{\frac{\beta \abs{e_\ell}}{50}}. 
\end{align*}
Note that $x\tp{\frac{1}{2}}^x \leq 1$ for all $x \geq 0$. We have
\begin{align*}
\sum_{v \in V \setminus \{v_0\}}\DTV{\nu_{v}^{X_{V \setminus \{v\}}}}{\nu_{v}^{Y_{V \setminus \{v\}}}} \leq \sum_{e_1,e_2,\ldots e_{\ell} \text{ is a PP}}\tp{\frac{1}{4 D^3}}^{\ell}.	
\end{align*}
If $e_1,e_2,\ldots e_{\ell}$ is a percolation path, then $e_1,e_2,\ldots e_{\ell}$ is a path in $\Lin^3(H)$ and $e_1 \in N^3(v_0)$. 
Note that $\abs{N^3(v_0)}\leq D + D(D-1) + D(D-1)^2 \leq D^3$ (due to~\eqref{eq-def-Nkv}) and the maximum degree of $\Lin^3(H)$ is at most $D^3$.
The number of such paths is at most $D^{3\ell}$. We have
\begin{align*}
\sum_{v \in V \setminus \{v_0\}}\DTV{\nu_{v}^{X_{V \setminus \{v\}}}}{\nu_{v}^{Y_{V \setminus \{v\}}}} \leq \sum_{e_1,e_2,\ldots e_{\ell} \text{ is a PP}}\tp{\frac{1}{4 D^3}}^{\ell} \leq \sum_{\ell = 1}^{\infty}D^{3\ell}\tp{\frac{1}{4 D^3}}^{\ell} \leq \frac{1}{2}.	&\qedhere
\end{align*}
\end{proof}

\subsubsection{Proof of \Cref{lemma-PP}}
\label{section-proof-lemma-PP}
We first introduce some notations for proving \Cref{lemma-PP}. 
Let $\Phi = (V,\Dom{Q},\Cons{C})$ to denote the original input CSP formula of \Cref{alg-mcmc}.
Let $D$ denote the maximum degree of the dependency graph of $\Phi$.
For each $v \in V$, let $q_v = \abs{Q_v}$. Let
\begin{align*}
p \triangleq \max_{c \in \Cons{C}}\prod_{v \in \vbl{c}}\frac{1}{q_v}.	
\end{align*}
Let $\Proj{h}$ denote the original projection scheme for $\Phi$ satisfying \Cref{condition-projection} with parameters $\alpha$ and $\beta$.
Recall that $\Phi^{X}=(V,\Dom{Q}^{X}=(Q^{X}_u)_{u \in V},\Cons{C})$  and $\Phi^{Y}=(V,\Dom{Q}^{Y}=(Q^{Y}_v)_{v \in V},\Cons{C})$ are defined in~\eqref{eq-def-QX-QY-gen}.
Recall that $\PX=(h^X_v)_{v \in V}$ and  $\PY = (h^Y_v)_{v \in V}$ denote the projection schemes for $\Phi^{X}$ and $\Phi^{Y}$, where $h^X_v: Q^{X}_v \to \Sigma'_v$ and $h^Y_v: Q^{Y}_v \to \Sigma'_v$. 
Recall that $\PX$ and $\PY$ satisfy \Cref{condition-projection-coupling}.
For each $v \in V$, $s^X_v =s^Y_v =s'_v$.
The following lemma gives the key property for $\nu_{v,X}$ and $\nu_{v ,Y}$ in \Cref{line-optimal-coupling-gen}.
\begin{lemma}
\label{lemma-local-uniform-coupling-gen}
Suppose the original input CSP formula $\Phi$ of \Cref{alg-mcmc} satisfies
\begin{align*}
\log \frac{1}{p} \geq \frac{50}{\beta}\log\tp{\frac{2000 D^{4}}{\beta}}.	
\end{align*}
Let $\Lambda \subseteq V$ and $v \in V \setminus \Lambda$. 
Let $\sigma_X, \sigma_Y \in \Sigma'_{\Lambda} = \bigotimes_{u \in \Lambda}\Sigma'_u$ be two partial assignments on $\Lambda$.
For any $c_X,c_Y \in \Sigma'_v$, 
\begin{align*}
\frac{\abs{(h^X_v)^{-1}(c_X)}}{q_v^X} \tp{1 - \frac{\beta}{500 D^3}} \leq \nu_{v,X}(c_X \mid \sigma_X) \leq \frac{\abs{(h^X_v)^{-1}(c_X)}}{q_v^X} \tp{1 + \frac{\beta}{500 D^3}},\\
\frac{\abs{(h^Y_v)^{-1}(c_Y)}}{q_v^Y} \tp{1 - \frac{\beta}{500 D^3}} \leq \nu_{v,Y}(c_Y \mid \sigma_Y) \leq \frac{\abs{(h^Y_v)^{-1}(c_Y)}}{q_v^Y} \tp{1 + \frac{\beta}{500 D^3}}.
\end{align*}
Furthermore, if the variable $v$ satisfies $\log \ftp{\frac{q^X_v}{s'_v}} \geq t + \frac{5}{4}\log \tp{\frac{2000 D^{4}}{\beta}}$ and $\log \ftp{\frac{q^Y_v}{s'_v}} \geq t + \frac{5}{4}\log \tp{\frac{2000 D^{4}}{\beta}}$ for some $t \geq 0$,then for any $c_X,c_Y \in \Sigma'_v$,
 \begin{align*}
\frac{\abs{(h^X_v)^{-1}(c_X)}}{q_v^X} \tp{1 - \frac{\beta 2^{-t}}{500 D^3}} &\leq \nu_{v,X}(c_X \mid \sigma_X) \leq \frac{\abs{(h^X_v)^{-1}(c_X)}}{q_v^X} \tp{1 + \frac{\beta 2^{-t}}{500 D^3}},\\
\frac{\abs{(h^Y_v)^{-1}(c_Y)}}{q_v^Y} \tp{1 - \frac{\beta 2^{-t}}{500 D^3}} &\leq \nu_{v,Y}(c_Y \mid \sigma_Y) \leq \frac{\abs{(h^Y_v)^{-1}(c_Y)}}{q_v^Y} \tp{1 + \frac{\beta 2^{-t}}{500 D^3}}.
\end{align*}
\end{lemma}
\begin{proof}
We prove the lemma for $\nu_{v,X}(c_X \mid \sigma_X)$. The result for  $\nu_{v,Y}(c_Y \mid \sigma_Y) $ can be proved in a similar way.
To simplify the notation, denote $\sigma = \sigma_X$, $c^\star  = c_X$.
We define a new instance $\widetilde{\Phi}=(V,\widetilde{\Dom{Q}}=(\widetilde{Q}_u)_{u \in V}, \Cons{C})$:
\begin{align*}
\forall u \in V, \quad \widetilde{Q}_u = \begin{cases}
(h^X_u)^{-1}(\sigma_u) &\text{if } u \in \Lambda;\\
Q^{X}_u &\text{if } u \notin \Lambda.
 \end{cases}
\end{align*}
Let $\widetilde{\mu}$ denote the uniform distribution of all satisfying assignments to $\widetilde{\Phi}$.
By the definition of the projected distribution, if $X \sim \widetilde{\mu}$, then $\Pr[]{X_v \in (h^X_v)^{-1}(c^\star )}$ equals to $\nu_{v,X}(c^\star  \mid \sigma)$.
By \Cref{condition-projection-coupling}, for any constraint $c \in \Cons{C}$, it holds that
\begin{equation}
\label{eq-qxv-lower-bound}
\begin{split}
\sum_{v \in \vbl{c}} \log \ftp{\frac{q^X_v}{s'_v}} \geq \frac{\beta}{10}\log \frac{1}{p} \geq  {5}\log\tp{\frac{2000 D^{4}}{\beta}}.	
\end{split}
\end{equation}
Let $\+D$ denote the product distribution such that each variable $u \in V$ takes a value from $\widetilde{Q}_u$ uniformly at random. 
For each constraint $c \in \Cons{C}$, let $B_c$ denote the bad event that $c$ is not satisfied. 
Let $\+B$ denote the collection of bad events $(B_c)_{c \in \Cons{C}}$.
Let $\Gamma(\cdot)$ be defined as in the Lov\'asz local lemma (\Cref{theorem-LLL}).
We define a function $x: \+B \to (0,1)$ such that
\begin{align*}
\forall c \in \Cons{C} \text{ s.t. } v \notin \vbl{c},\quad 	x(B_c) &= \frac{\beta}{2000 D^4};\\
\forall c \in \Cons{C} \text{ s.t. } v \in \vbl{c}, \quad x(B_c) &= \frac{\beta\ftp{ q^X_v / s'_v }}{2000 D^4 q^X_v}.
\end{align*}
Since $\Proj{h}^X$ is a balanced projection scheme, $\abs{\widetilde{Q}_u} \geq \ftp{\frac{q^X_u}{s'_u}}$ for all $u \in V$.
For any constraint $c \in \Cons{C}$ such that $v \notin \vbl{c}$, it holds that
\begin{align}
\label{eq-LLL-not-v}
\Pr[\+D]{B_c} &= \prod_{u \in \vbl{c}}\frac{1}{\vert{\widetilde{Q}_u}\vert}\leq \prod_{u \in \vbl{c}}\frac{1}{\ftp{ q^X_u / s'_u }} \leq \frac{\beta}{2000 ^ 5D^{20}} \leq \frac{\beta}{2000 D^4}\tp{1 - \frac{\beta}{2000 D^4}}^{2000D^4/\beta - 1}\notag\\
&\leq  \frac{\beta}{2000 D^4}\tp{1 - \frac{\beta}{2000 D^4}}^{D} \leq x(B_c) \prod_{B_{c'} \in \Gamma(B_c)}\tp{1 - x(B_{c'})},
\end{align}
where the last inequality holds because $x(B_c) \leq \frac{\beta}{2000 D^4}$ for all $c \in \Cons{C}$. 
Note that $v \notin \Lambda$.
For any $c \in \Cons{C}$ such that $v \in \vbl{c}$, by~\eqref{eq-qxv-lower-bound}, it holds that
\begin{align*}
\Pr[\+D]{B_c} &= \frac{1}{q^X_v}\prod_{u \in \vbl{c}: u \neq v}\frac{1}{\vert{\widetilde{Q}_u}\vert} \leq 	 \frac{\ftp{q^X_v / s'_v} }{q^X_v} \prod_{u \in \vbl{c}}\frac{1}{\ftp{ q^X_u / s'_u }} \leq \frac{\ftp{q^X_v / s'_v} }{q^X_v} \cdot \frac{\beta}{2000^5D^{20}} \\
&\leq \frac{\beta\ftp{q^X_v / s'_v} }{2000 D^4 q^X_v} \tp{1 - \frac{\beta}{2000 D^4}}^{2000D^4/\beta - 1}
\leq  \frac{\beta\ftp{q^X_v / s'_v} }{2000 D^4 q^X_v}\tp{1 - \frac{\beta}{2000 D^4}}^{D} \\
&\leq x(B_c) \prod_{B_{c'} \in \Gamma(B_c)}\tp{1 - x(B_{c'})}.
\end{align*}
Fix a value $c^\star  \in \Sigma'_v$.
Let $A$ denote the event that $v$ takes a value in $(h^X_v)^{-1}(c^\star )$. It holds that $\abs{\Gamma(A)} \leq D$. For any $B_c \in \Gamma(A)$, it holds that $v \in \vbl{c}$ and $x(B_c) = \frac{\beta\ftp{ q^X_v / s'_v }}{2000 D^4 q^X_v}$.
Recall that $\widetilde{\mu}$ denotes the uniform distribution of all satisfying assignments to $\widetilde{\Phi}$.
By Lov\'asz local lemma (\Cref{theorem-LLL}),
\begin{align*}
\Pr[\widetilde{\mu}]{A} = \nu_{v,X}(c^\star \mid \sigma) &\leq \frac{\abs{(h^X_v)^{-1}(c^\star )} }{q^X_v} \tp{1 - \frac{\beta\ftp{ q^X_v / s'_v }}{2000 D^4 q^X_v} }^{-D}\leq \frac{\abs{(h^X_v)^{-1}(c^\star )} }{q^X_v} \exp \tp{ \frac{\beta\ftp{ q^X_v / s'_v }}{1000 D^3 q^X_v}  } \\
&\leq \frac{\abs{(h^X_v)^{-1}(c^\star )} }{q^X_v} \tp{1 + \frac{\beta\ftp{ q^X_v / s'_v }}{500 D^3 q^X_v}  }
\leq  \frac{\abs{(h^X_v)^{-1}(c^\star )} }{q^X_v} \tp{1 + \frac{\beta}{500 D^3 }  }.
\end{align*}
This proves the upper bound.
Let $A'$ denote the event that $v$ does not take any value in $(h^X_v)^{-1}(c^\star )$, then $\abs{\Gamma(A')} \leq D$. For any $B_c \in \Gamma(A')$, it holds that $v \in \vbl{c}$ and $x(B_c) = \frac{\beta\ftp{ q^X_v / s'_v }}{2000 D^4 q^X_v}$.
By \Cref{theorem-LLL},
\begin{align*}
\Pr[\widetilde{\mu}]{A'} &= 1 - \nu_{v,X}(c^\star \mid \sigma) \leq \tp{1-\frac{\abs{(h^X_v)^{-1}(c^\star )} }{q^X_v}} \tp{1 - \frac{\beta\ftp{ q^X_v / s'_v }}{2000 D^4 q^X_v} }^{-D}\\
&\leq \tp{1 -\frac{\abs{(h^X_v)^{-1}(c^\star )} }{q^X_v}} \exp \tp{ \frac{\beta\ftp{ q^X_v / s'_v }}{1000 D^3 q^X_v}  } \leq \tp{1-\frac{\abs{(h^X_v)^{-1}(c^\star )} }{q^X_v} }\tp{1 + \frac{\beta\ftp{ q^X_v / s'_v }}{500 D^3 q^X_v}  }. 
\end{align*}
Let $a = \abs{(h^X_v)^{-1}(c^\star )}/q^X_v$ and $b = \ftp{ q^X_v / s'_v } / q^X_v$. Since $\Proj{h}^X$ is a balanced projection scheme (\Cref{condition-projection-coupling}), it holds that $\abs{(h^X_v)^{-1}(c^\star )} \geq  \ftp{ q^X_v / s'_v } $  and $a \geq b$. Thus
\begin{align}
\label{eq-proof-lower-bound-v}
\nu_{v,X}(c^\star \mid \sigma) &\geq 1 - \tp{1 - a}\tp{1 + \frac{\beta b}{500D^3}}	  = a\tp{1 + \frac{\beta b}{500D^3} - \frac{\beta b}{500aD^3}} \geq {a}\tp{1  - \frac{\beta b}{500aD^3}}\notag\\
\tp{\text{by }a \geq b}\quad &\geq a\tp{1 - \frac{\beta}{500 D^3}} =\frac{\abs{(h^X_v)^{-1}(c^\star )}}{q^X_v} \tp{1 - \frac{\beta}{500 D^3}}.
\end{align}
This proves the lower bound.

Next, we assume 
\begin{align}
\label{eq-asm-stronger}
\log \ftp{\frac{q^X_v}{s'_v}} \geq t + \frac{5}{4}\log \tp{\frac{2000 D^{4}}{\beta}}. 
\end{align}
For each bad event $B_c$, 
we define a function $x: \+B \to (0,1)$ such that
\begin{align*}
\forall c \in \Cons{C} \text{ s.t. } v \notin \vbl{c},\quad 	x(B_c) &= \frac{\beta}{2000 D^4};\\
\forall c \in \Cons{C} \text{ s.t. } v \in \vbl{c}, \quad x(B_c) &= \frac{\beta2^{-t}\ftp{ q^X_v / s'_v }}{2000 D^4 q^X_v}.
\end{align*}
Note that for any $c \in \Cons{C}$, it holds that $x(B_c) \leq \frac{\beta}{2000 D^4}$. By the same proof, for any constraint $c \in \Cons{C}$ such that $v \notin \vbl{c}$,~\eqref{eq-LLL-not-v} still holds. 
For any constraint $c \in \Cons{C}$ such that $v \in \vbl{c}$, we have
\begin{align*}
\Pr[\+D]{B_c} &= \frac{1}{q^X_v}\prod_{u \in \vbl{c}: u \neq v}\frac{1}{\vert{\widetilde{Q}_u}\vert} \leq 	 \frac{\ftp{q^X_v / s'_v} }{q^X_v} \prod_{u \in \vbl{c}}\frac{1}{\ftp{ q^X_u / s'_u }} \leq  \frac{\ftp{q^X_v / s'_v} }{q^X_v} \frac{1}{\ftp{ q^X_v / s'_v }}\\
\tp{\text{by~\eqref{eq-asm-stronger} and $\beta \leq 1$}}\quad&\leq \frac{\ftp{q^X_v / s'_v} }{q^X_v} \cdot \frac{\beta2^{-t}}{2000^{5/4}D^5} 
\leq \frac{\beta2^{-t}\ftp{q^X_v / s'_v} }{2000 D^4 q^X_v} \tp{1 - \frac{\beta}{2000 D^4}}^{2000D^4/\beta - 1}\\
&\leq  \frac{\beta2^{-t}\ftp{q^X_v / s'_v} }{2000 D^4 q^X_v}\tp{1 - \frac{\beta}{2000 D^4}}^{D} \leq x(B_c) \prod_{B_{c'} \in \Gamma(B_c)}\tp{1 - x(B_{c'})}.	
\end{align*}
Thus, the function $x: \+B \to (0,1)$ satisfies the Lov\'asz local lemma condition. By~\Cref{theorem-LLL}, 
\begin{align*}
\Pr[\widetilde{\mu}]{A} = \nu_{v,X}(c^\star \mid \sigma) &\leq \frac{\abs{(h^X_v)^{-1}(c^\star )} }{q^X_v} \tp{1 - \frac{\beta 2^{-t}\ftp{ q^X_v / s'_v }}{2000 D^4 q^X_v} }^{-D}\leq \frac{\abs{(h^X_v)^{-1}(c^\star )} }{q^X_v} \exp \tp{ \frac{\beta 2^{-t}\ftp{ q^X_v / s'_v }}{1000 D^3 q^X_v}  } \\
&\leq \frac{\abs{(h^X_v)^{-1}(c^\star )} }{q^X_v} \tp{1 + \frac{\beta 2^{-t}\ftp{ q^X_v / s'_v }}{500 D^3q^X_v}  }
\leq  \frac{\abs{(h^X_v)^{-1}(c^\star )} }{q^X_v} \tp{1 + \frac{\beta 2^{-t}}{500 D^3 }  }.
\end{align*}
Furthermore, 
\begin{align*}
\Pr[\widetilde{\mu}]{A'} &= 1 - \nu_{v,X}(c^\star \mid \sigma) \leq \tp{1-\frac{\abs{(h^X_v)^{-1}(c^\star )} }{q^X_v}} \tp{1 - \frac{\beta 2^{-t}\ftp{ q^X_v / s'_v }}{2000 D^4 q^X_v} }^{-D}\\
&\leq \tp{1 -\frac{\abs{(h^X_v)^{-1}(c^\star )} }{q^X_v}} \exp \tp{ \frac{\beta 2^{-t}\ftp{ q^X_v / s'_v }}{1000 D^3 q^X_v}  } 
\leq \tp{1-\frac{\abs{(h^X_v)^{-1}(c^\star )} }{q^X_v} }\tp{1 + \frac{\beta 2^{-t}\ftp{ q^X_v / s'_v }}{500 D^3 q^X_v}  }. 
\end{align*}
By the same proof in~\eqref{eq-proof-lower-bound-v}, we have
\begin{align*}
\nu_{v,X}(c^\star \mid \sigma) \geq 	 \frac{\abs{(h^X_v)^{-1}(c^\star )} }{q^X_v}  \tp{1 - \frac{\beta 2^{-t}}{500 D^3 }  }. &\qedhere
\end{align*}
\end{proof}
Now, we are ready to prove \Cref{lemma-PP}. 
Fix a percolation path (PP) $e_1,e_2,\ldots,e_{\ell}$ in $\Lin^3(H)$	.
We bound the probability that all $e_i$ are bad for $1\leq i \leq \ell$. 
Recall $s'_v = s^X_v =s^Y_v$ for all $v \in V$.
For each hyperedge $e_i$, define 
\begin{align*}
V(e_i) \triangleq \{v \in e_i \mid s'_v \neq 1 \text{ and } v \neq v_0\}.	
\end{align*}
Note that for variables  $v \in e_i \setminus (V(e_i) \cup \{v_0\})$, it must hold that $s'_v =\abs{\Sigma_v'} = 1$. 
It must hold that $X^{\cnon}_v = Y^{\cnon}_v$, which implies the coupling on $v$ cannot be failed.
Hence, if there is a variable $u \in e_i \setminus \{v_0\}$ such that $X^{\cnon}_u \neq Y^{\cnon}_u$, it must hold that $u \in V(e_i)$.
In the while-loop, the coupling $\cnon$ assigns values to variables one-by-one, using the optimal coupling between marginal distributions. Let 
\begin{align*}
k(e_i) \triangleq \abs{V(e_i)}.	
\end{align*}
Fix an index $1 \leq i \leq \ell - 1$.
Let $c(e_i)$ denote the constraint represented by $e_i$.
We can define $k(e_i) + 1$ bad events $B_i^{(j)}$ for $1\leq j \leq k(e_i) + 1$:
\begin{itemize}
\item if $1\leq j \leq k(e_i)$: the constraint $c(e_i)$ is not satisfied by both $\Ass{X}^{\cnon}$ and $\Ass{Y}^{\cnon}$ after $j-1$ variables in $V(e_i)$ are assigned values by $\cnon$, and the coupling on $j$-th variable fails, i.e. $X^{\cnon}_{v_j} \neq Y^{\cnon}_{v_j}$, where $v_j\in V(e_i)$ is the $j$-th variable in $V(e_i)$ whose value is assigned by the coupling $\cnon$;
\item if $j = k(e_i)+1$: the constraint $c(e_i)$ is not satisfied  by both $\Ass{X}^{\cnon}$ and $\Ass{Y}^{\cnon}$  after all variables in $e_i$ are assigned values by the coupling $\cnon$.
\end{itemize}
Let $B_i$ denote the event $\bigvee_{j=1}^{k(e_i)+1}B_{i}^{(j)}$.
By \Cref{definition-bad-edge-gen}, we have the following relation
\begin{align*}
e_i\text{ is bad} \quad \Longleftrightarrow \quad e_i\text{ fails} \quad \Longrightarrow \quad B_i = \bigvee_{j=1}^{k(e_i)+1}B_i^{(j)}.
\end{align*}
By \Cref{defintiion-fails-cnon}, if $e_i$ fails in type-I, then there must exist $1 \leq j \leq k(e_i)$ such that the coupling of $j$-th variable in $V(e_i)$ fails and $e_i$ is not satisfied by both $\Ass{X}^{\cnon}$ and $\Ass{Y}^{\cnon}$ after $j-1$ variables in $V(e_i)$ are assigned values (otherwise, $e_i$ will be removed in~\Cref{line-remove-e-v0} or~\Cref{line-remove-e-gen}). Hence, if $e_i$ fails in type-I, $\bigvee_{j=1}^{k(e_i)}B_i^{(j)}$ must occur. If $e_i$ fails in type-II, then $B_i^{(k(e_i)+1)}$ must occur. This proves the above relation.

For hyperedge $e_\ell$, let $c(e_{\ell})$ denote the constraint represented by $e_\ell$, we define the bad event $B_\ell$ as
\begin{itemize}
\item $B_\ell$: the constraint $c(e_{\ell})$ is not satisfied by both $\Ass{X}^{\cnon}$ and $\Ass{Y}^{\cnon}$  after all variables in $e_\ell \setminus \{\vst\}$ are assigned values by the coupling $\cnon$, and the coupling on $\vst$ fails, i.e. $X^{\cnon}_{\vst} \neq Y^{\cnon}_{\vst}$.
\end{itemize}
By \Cref{definition-bad-edge-gen}, we have the following relation
\begin{align*}
e_\ell\text{ is bad} \quad \Longrightarrow \quad B_{\ell}.
\end{align*}
Let $\Omega_{B} = \bigotimes_{i = 1}^{\ell - 1}[k(e_i) + 1]$, where $[k({e_i}) + 1] = \{1,2,\ldots,k(e_i) + 1\}$.
We have the following relation
\begin{align*}
\Pr[\cnon]{\forall 1\leq i \leq \ell: e_i \text{ is bad}} \leq \Pr[\cnon]{\forall 1\leq i \leq \ell: B_i}\leq \sum_{z \in \Omega_{B}}\Pr[\cnon]{B_{\ell} \land \forall 1\leq i \leq \ell-1: B_i^{(z_i)}},
\end{align*}
where $z \in \Omega_B$ is a $(\ell - 1)$-dimensional vector and $z_i \in [k(e_i) + 1]$.
Fix a vector $z \in \Omega_B$. Let 
\begin{align*}
\+E_1 &= \{e_i \mid 1 \leq  i \leq \ell -1 \land z_i \leq k(e_i) \}\\
\+E_2 &= \{e_i \mid 1 \leq  i \leq \ell - 1 \land z_i = k(e_i) + 1\}.	
\end{align*}
We will prove that 
\begin{align}
\label{eq-gen-target}
&\Pr[\cnon]{B_{\ell} \land \forall 1\leq i \leq \ell-1: B_i^{(z_i)}}\notag\\
\leq&\, \prod_{e_i \in \+E_1}\tp{\tp{\frac{3}{4}}^{z_i - 1} \frac{1}{200 D^3}} \times \prod_{e_j \in \+E_2}\tp{\frac{1}{200D^3}}	\times \tp{\frac{\beta}{200 D^3} \tp{\frac{1}{2}}^{\frac{\beta \abs{e_\ell}}{50}}}.
\end{align}
By~\eqref{eq-gen-target}, we have
\begin{align*}
\Pr[\cnon]{\forall 1\leq i \leq \ell: e_i \text{ is bad}} &\leq 	\sum_{z \in \Omega_{B}}\Pr[\cnon]{B_{\ell} \land \forall 1\leq i \leq \ell-1: B_i^{(z_i)}}\\
\tp{\text{by }~\eqref{eq-gen-target}}\quad&\leq \tp{\frac{1}{200 D^3} + \frac{1}{200 D^3} \sum_{j=1}^{k(e_i)}\tp{\frac{3}{4}}^{j-1}}^{\ell - 1}\times \tp{ \frac{\beta}{200 D^3} \tp{\frac{1}{2}}^{\frac{\beta \abs{e_\ell}}{50}}}\\
&\leq \tp{\frac{1}{40 D^3}}^{\ell - 1}\frac{\beta}{200 D^3} \tp{\frac{1}{2}}^{\frac{\beta \abs{e_\ell}}{50}} \leq \tp{\frac{1}{4 D^3}}^{\ell} \frac{\beta}{50}\tp{\frac{1}{2}}^{\frac{\beta \abs{e_\ell}}{50}}.
\end{align*}
This proves \Cref{lemma-PP}. The rest of this section is dedicated to the proof of~\eqref{eq-gen-target}.

Note that the RHS of~\eqref{eq-gen-target} is a product.
Although all hyperedges in a percolation path are mutually disjoint, we cannot show that all bad events $B_i^{(z_i)}$ and $B_{\ell}$ are mutually independent. Because all the bad events are defined by $\cnon$, they may have some correlations with each other. To prove~\eqref{eq-gen-target}, we will use an independent random process to dominate the event that all $B_i^{(z_i)}$ and $B_{\ell}$ occur.

To prove~\eqref{eq-gen-target}, we first divide the bad event $B_{\ell}$ into two parts $B_{\ell}^{(1)}$ and $B_{\ell}^{(2)}$, where $B_{\ell}^{(1)}$ denotes the event that  the constraint $c(e_{\ell})$ is not satisfied by both $X^{\cnon}_S$ and $Y^{\cnon}_S$, where $S = e_\ell \setminus \{\vst\}$, and $B_{\ell}^{(2)}$ denotes the event that the coupling on $\vst$ fails, i.e. $X^{\cnon}_{\vst} \neq Y^{\cnon}_{\vst}$. It is easy to see $B_{\ell} = B_{\ell}^{(1)} \land B_{\ell}^{(2)}$. Note that $\vst \in e_\ell$ and $q^X_{\vst} = q^Y_{\vst}$. By~\eqref{eq-condition-coupling-projection} in \Cref{condition-projection-coupling}, one of the following two conditions must be satisfied:
\begin{align}
\label{eq-cond-1}
&\min\tp{\sum_{v \in \vbl{c} \setminus \{\vst\} }\log\frac{q_v^X}{\ctp{q_v^X / s_v'}},\sum_{v \in \vbl{c} \setminus \{\vst\} }\log\frac{q_v^Y}{\ctp{q_v^Y / s_v'}}}\geq \frac{\beta}{20} \tp{\sum_{v \in \vbl{c}}\log q_v},
\end{align}
\begin{align}
\label{eq-cond-2}
\log \ftp{\frac{ q^X_{\vst}}{s'_{\vst}}} = \log\ftp{\frac{ q^Y_{\vst}}{s'_{\vst}}} \geq \frac{\beta}{20} \tp{\sum_{v \in \vbl{c}}\log q_v}.
\end{align}
If~\eqref{eq-cond-1} holds, we can prove~\eqref{eq-gen-target} by bounding the RHS of the following inequality
\begin{align}
\label{eq-cond-1-RHS}
\Pr[\cnon]{B_{\ell} \land \forall 1\leq i \leq \ell-1: B_i^{(z_i)}} \leq \Pr[\cnon]{B_{\ell}^{(1)} \land \forall 1\leq i \leq \ell-1: B_i^{(z_i)}}.
\end{align}
If~\eqref{eq-cond-2} holds, we can prove~\eqref{eq-gen-target} by bounding the RHS of the following inequality
\begin{align}
\label{eq-cond-2-RHS}
\Pr[\cnon]{B_{\ell} \land \forall 1\leq i \leq \ell-1: B_i^{(z_i)}} \leq \Pr[\cnon]{B_{\ell}^{(2)} \land \forall 1\leq i \leq \ell-1: B_i^{(z_i)}}.
\end{align}
In the rest of the proof, we mainly focus on the case when~\eqref{eq-cond-1} holds. If~\eqref{eq-cond-2} holds, we can modify our proof to bound the RHS of~\eqref{eq-cond-2-RHS}, this part will be discussed later. 

Assume~\eqref{eq-cond-1} holds. We start to bound the RHS of~\eqref{eq-cond-1-RHS}.  To do this, we will give a particular implementation of the coupling $\cnon$ such that if $B^{(1)}_\ell$ and all $B_i^{(z_i)}$ occur, then some independent events must occur in our implementation and their probabilities are easy to bound. We first sample a set  $\+R$ of real numbers from $[0,1]$ uniformly and independently.
\begin{itemize}
\item For each $e_i \in \+E_1$, sample $k(e_i)$ random real numbers $r_{e_i}(j) \in [0,1]$ for $1 \leq j \leq k(e_i)$ uniformly and independently.
\item For each $e_i \in \+E_2 \cup \{e_{\ell}\}$, for each variable $v \in e_i$, sample a random real number $r_{v} \in [0,1]$ uniformly and independently. 
\end{itemize}
We then run the coupling $\cnon$ in \Cref{alg-coupling-gen}, but in some particular steps, we will use the random numbers in $\+R$ to implement the sampling step in $\cnon$.

We start from the special variable $v_0$.
Note that if $v_0$ appears in the percolation path, then $v_0 \in e_1$. 
The coupling $\cnon$ will sample the values of $v_0$ in \Cref{line-v0-gen}.
We use the real number $r_{v_0}$ to implement this sampling step if and only if $v_0 \in e_1$ and $e_1 \in \+E_2$.
Let $c(e_1)$ denote the constraint represented by $e_1$. 
Suppose $c(e_1)$ forbids the configuration $\sigma \in Q_{e_1}$, i.e. $(c(e_1))(\sigma) = \False$. 
By definition, in $\Phi^{X}$, $Q^{X}_{v_0} = h^{-1}_{v_0}(X_{v_0})$ and  in $\Phi^{Y}$, $Q^{Y}_{v_0} = h^{-1}_{v_0}(Y_{v_0})$. Note that $Q^{X}_{v_0} \neq Q^{Y}_{v_0}$. 
Thus, $e_1$ must be satisfied in $\Phi^{X}$ or $\Phi^{Y}$, because it must hold that $\sigma_{v_0} \notin Q^{X}_{v_0}$ or $\sigma_{v_0} \notin Q^{Y}_{v_0}$. 
If $e_1$ is satisfied in both $\Phi^{X}$ and $\Phi^{Y}$, then the hyperedge $e_1$ cannot be bad.
We may assume $e_1$ is not satisfied in $\Phi^{X}$ (i.e. $\sigma_{v_0} \in Q^{X}_{v_0}$) and $e_1$ is satisfied in $\Phi^{Y}$ (i.e. $\sigma_{v_0} \notin Q^{Y}_{v_0}$). Otherwise, we can swap the roles of $X$ and $Y$ in the whole analysis.
We use $r_{v_0}$ to sample $X^{\cnon}_{v_0}$ in \Cref{line-v0-gen} of $\cnon$. 
Note that there is only one $j \in \Sigma'_{v_0}$ such that $\sigma_{v_0} \in (h^X_{v_0})^{-1}(j)$. We can set $X^{\cnon}_{v_0} = j$ if $r_{v_0} \leq \nu_{v_0,X}(j)$. By \Cref{lemma-local-uniform-coupling-gen}, $\nu_{v_0,X}(j)\leq (1+\frac{1}{500D^3})\ctp{q^{X}_{v_0}/s'_{v_0}}/{q^X_{v_0}}$. 
Note that if $s'_{v_0} = 1$, then $\nu_{v_0,X}(j) = 1$, which implies
$
\nu_{v_0,X}(j) = 1 = (\frac{\lceil{q^{X}_{v_0}/s'_{v_0}}\rceil}{q^X_{v_0}})^{0.95}.
$
If $s_{v_0}' \geq 2$, then $\ctp{q^{X}_{v_0}/s'_{v_0}}/{q^X_{v_0}} \leq \ctp{q^{X}_{v_0}/2}/{q^X_{v_0}}\leq  \frac{2}{3}$ (because $q^X_{v_0} \geq s'_{v_0} \geq 2$), which implies 
\begin{align*}
\nu_{v_0,X}(j)\leq (1+\frac{1}{500D^3})\frac{\ctp{q^{X}_{v_0}/s'_{v_0}}}{q^X_{v_0}} \leq	\frac{501}{500}\frac{\ctp{q^{X}_{v_0}/s'_{v_0}}}{q^X_{v_0}}\leq \tp{\frac{\ctp{q^{X}_{v_0}/s'_{v_0}}}{q^X_{v_0}}}^{0.95}.
\end{align*}
After \Cref{line-v0-gen}, if $e_1$ is not satisfied by both $X^{\cnon}_{v_0}$ and $Y^{\cnon}_{v_0}$, then the following event must occur
\begin{align}
\label{eq-bad-v0}
r_{v_0} \leq \tp{\frac{\ctp{q^X_{v_0}/s'_{v_0}}}{q^X_{v_0}}}^{0.95}. 	
\end{align}
 
During the while-loop of $\cnon$, we maintain an index $j_i$ for each hyperedge $e_i \in \+E_1$. Initially, all $j_i = 0$. Suppose the coupling $\cnon$ picks a variable $u$ in \Cref{line-pick-gen}. Suppose $u\in e_i$ for some $1 \leq i \leq \ell$. Note that such hyperedge $e_i$ is unique because all hyperedges in a percolation path are mutually disjoint. Let $c(e_i)$ denote the constraint represented by $e_i$. 
Suppose $c(e_i)$ forbids the configuration $\tau \in \Dom{Q}_{e_i}$, i.e. $(c(e_i))(\tau) = \False$.
Since $u \neq v_0$, by \Cref{condition-projection-coupling}, it holds that $h^X_u = h^Y_u$.
Let $c^{\star} \in \Sigma_u'$ denote the value such that $\tau_u \in (h^X_u)^{-1}(c^{\star}) = (h^Y_u)^{-1}(c^{\star})$.
We need to sample $c_x \in \Sigma'_u$ and $c_y \in \Sigma'_u$ from the optimal coupling between $\nu_{u,X}(\cdot \mid \Ass{X}^{\cnon})$ and $\nu_{u,Y}(\cdot \mid \Ass{Y}^{\cnon})$ in \Cref{line-optimal-coupling-gen}. 
By~\eqref{eq-opt-coupling-1} and~\eqref{eq-coupling-min},
the optimal coupling satisfies the following properties,
\begin{align*}
\Pr[]{c_x = c_y} =  \sum_{j \in \Sigma'_u}\Pr[]{c_x = c_y = j} &= \sum_{j \in \Sigma'_u }\min \tp{\nu_{u,X}(j \mid \Ass{X}^{\cnon}), \nu_{u,Y}(j \mid \Ass{Y}^{\cnon})}\\
&= 1 - \DTV{\nu_{u,X}(\cdot \mid \Ass{X}^{\cnon})}{\nu_{u,Y}(\cdot \mid \Ass{Y}^{\cnon})},\\
\Pr[]{c_x = c^\star \lor c_y = c^\star} &= \max \tp{\nu_{u,X}(c^\star \mid \Ass{X}^{\cnon}), \nu_{u,Y}(c^\star \mid \Ass{Y}^{\cnon})}.	
\end{align*}

Let $t_{\max} \triangleq  \max \tp{\nu_{u,X}(c^\star \mid \Ass{X}^{\cnon}), \nu_{u,Y}(c^\star \mid \Ass{Y}^{\cnon})}$ and $d_{\mathsf{TV}}\triangleq \DTV{\nu_{u,X}(\cdot \mid \Ass{X}^{\cnon})}{\nu_{u,Y}(\cdot \mid \Ass{Y}^{\cnon})}$.
Note that either $e_i \in \+E_1$ or $e_i \in \+E_2 \cup \{e_{\ell}\}$.
We will use the following  procedure to implement the sampling step in \Cref{line-optimal-coupling-gen}.
\begin{itemize}
\item Case $e_i \in \+E_1$ and $u \in V(e_i)$. Set $j_i \gets j_i + 1$ and let $r = r_{e_i}(j_i)$. If $j_i < z_i$, we sample $c_x$ and $c_y$ such that $c_x = c^\star \lor c_y = c^\star$ if and only if $r \leq t_{\max}$.
	If $j_i = z_i$, we sample $c_x$ and $c_y$ such that $c_x \neq c_y$ if and only if $r \leq d_{\mathsf{TV}}$.
	If $j_i > z_i$, we arbitrarily sample $c_x$ and $c_y$ from their optimal coupling.

\item Case $e_i \in \+E_2 \cup \{e_\ell\}$. Let $r = r_u$.  Sample $c_x$ and $c_y$ such that $c_x = c^\star \lor c_y = c^\star$ if and only if $r \leq t_{\max}$.
\item Otherwise, we do not use random numbers in $\+R$ to implement the coupling.
\end{itemize}

We will use the following properties to analysis our implementation. 
Note that after we assigned the values to variable $u$, if $c(e_i)$ is not satisfied by both $X^{\cnon}_u$ and $Y^{\cnon}_u$, then it must hold that $c_x = c^\star$ or $c_y = c^{\star}$. 
Since $u \neq v_0$, by \Cref{condition-projection-coupling}, $Q^{X}_u = Q^{Y}_u$ and $h^X_u = h^Y_u$.
By \Cref{lemma-local-uniform-coupling-gen}, we can prove the following properties. For any $u$ with $s'_u > 1$, we have $q^X_u =q^Y_u \geq s'_u > 1$, thus
\begin{align}
\label{eq-tmax-1}
t_{\max} &\leq \frac{\ctp{q^X_u / s'_u} }{q^X_u}\tp{1 + \frac{1}{500 D^3}} \leq  \frac{\ctp{q^X_u/ 2} }{q^X_u}\tp{1 + \frac{1}{500}} \leq \frac{2}{3}\tp{1 + \frac{1}{500}} \leq \frac{3}{4}.
\end{align}
For any $u \in V \setminus \{v_0\}$,  since $Q^{X}_u = Q^{Y}_u$ and $h^X_u = h^Y_u$, by \Cref{lemma-local-uniform-coupling-gen}, it holds that
\begin{align}
t_{\max} &\leq \min\tp{1, \frac{\ctp{q^X_u / s'_u} }{q^X_u}\tp{1 + \frac{1}{500 D^3}}} \leq  \min \tp{1, \frac{501\ctp{q^X_u / s'_u} }{500q^X_u}}\leq \tp{\frac{\ctp{q^X_u / s'_u} }{q^X_u}}^{0.95} \label{eq-tmax-2};\\
d_{\mathsf{TV}} &\leq \frac{1}{2}\sum_{j \in \Sigma'_u}\frac{\abs{(h^X_u)^{-1}(j)}}{q^X_u} \tp{\frac{2}{500D^3}} = \frac{1}{500D^3} \leq \frac{1}{200D^3}.\label{eq-d-max}
\end{align}
Inequality~\eqref{eq-tmax-2} can be proved by considering two cases. 
If $s'_u = 1$, then $\tp{\frac{\ctp{q^X_u / s'_u} }{q^X_u}}^{0.95} = 1$, the inequality holds trivially.
If $s'_u > 1$, then $\frac{\ctp{q^X_u / s'_u} }{q^X_u} \leq \frac{2}{3}$, this implies~\eqref{eq-tmax-2}.
To prove~\eqref{eq-d-max}, note that  $Q^{X}_u = Q^{Y}_u$ (thus, $q^X_u =q^Y_u$); and $\Proj{h}^X$ and $\Proj{h}^Y$ use the same way to map $Q^{X}_u = Q^{Y}_u$ to $\Sigma'_u$ (i.e. $h^X_u=h^Y_u$). Hence, we can use the upper and lower bound in \Cref{lemma-local-uniform-coupling-gen} to bound the total variation distance $d_{\mathsf{TV}}$.

Consider a hyperedge $e_i \in \+E_1$. If the event $B_i^{(z_i)}$ occurs, then by definition, $c(e_i)$ is not satisfied after $z_i-1$ variables in $V(e_i)$ get the values and the coupling on $z_i$-th variable in $V(e_i)$ fails. Note that for all $v \in V(e_i)$, $s'_v > 1$. By~\eqref{eq-tmax-1} and~\eqref{eq-d-max}, the bad event $B_i^{(z_i)}$ implies the following event:
\begin{itemize}
\item $\+A_i$: for all $1 \leq j \leq z_i - 1$, $r_{e_i}(j) \leq \frac{3}{4}$ and $r_{e_i}(z_i) \leq \frac{1}{200D^3}$.	
\end{itemize}
This bad event $\+A_i$ occurs with probability
\begin{align}
\label{eq-A-1}
\Pr[]{\+A_i} = 	\tp{\frac{3}{4}}^{z_i - 1} \frac{1}{200 D^3}.
\end{align}

Consider a hyperedge $e_i \in \+E_2$. If the event $B_i^{(z_i)} =  B_i^{(k(e_i) + 1)}$ occurs, then by definition, $c(e_i)$ is not satisfied after all variables in $e_i$ get the value.
In our implementation, for any $v \in e_i$, we use $r_v$ to sample values for $X^{\cnon}_{v}$ and $Y^{\cnon}_v$.
 By~\eqref{eq-bad-v0} and~\eqref{eq-tmax-2}, the bad event $B_i^{(z_i)}$ implies
\begin{itemize}
\item $\+A_i$: for all $v \in e_i$, $r_v \leq \tp{\frac{\ctp{q^X_v / s'_v} }{q^X_v}}^{0.95}$.	
\end{itemize}
Since $e_i \in \+E_2$, it holds that $\vst \notin e_i$. By \Cref{condition-projection-coupling} and~\eqref{eq-cond-coupling-gen}, it holds that
\begin{align*}
\sum_{v \in e_i}\log\frac{q_v^X}{\ctp{q_v^X / s_v'}}  \geq \frac{\beta}{10}\sum_{v \in e_i} \log q_v	 \geq 5 \log \tp{\frac{2000D^4}{\beta}},
\end{align*}
This bad event $\+A_i$ occurs with probability
\begin{align}
\label{eq-A-2}
\Pr[]{\+A_i} = 	\prod_{v \in e_i}\tp{\frac{\ctp{q^X_u / s'_u} }{q^X_u}}^{0.95} \leq \tp{\frac{1}{2000 D^{20}}}^{0.95} \leq \frac{1}{200 D^3}.
\end{align}

Consider the hyperedge $e_{\ell}$. If the event $B_{\ell}^{(1)}$ occurs, then by definition, $c(e_\ell)$ is not satisfied after all variables in $e_i \setminus \{\vst\}$ get the value.
In our implementation, for any $v \in e_{\ell}$, we use $r_v$ to sample values for $X^{\cnon}_{v}$ and $Y^{\cnon}_v$.
 By~\eqref{eq-bad-v0} and~\eqref{eq-tmax-2}, the bad event $B_{\ell}^{(1)}$ implies
\begin{itemize}
\item $\+A_{\ell}$: for all $v \in e_\ell \setminus \{\vst\}$, $r_v \leq \tp{\frac{\ctp{q^X_v / s'_v} }{q^X_v}}^{0.95}$.	
\end{itemize}
By~\eqref{eq-cond-1}, we have
\begin{align*}
\sum_{v \in e_\ell \setminus \{\vst\} }\log\frac{q_v^X}{\ctp{q_v^X / s_v'}}  \geq \frac{\beta}{20}\sum_{v \in e_\ell} \log q_v,
\end{align*}
Note that in the original input CSP formula $\Phi = (V, \Dom{Q},\Cons{C})$ of \Cref{alg-mcmc}, the domain size of each variable is at least 2 (otherwise,the value of such variable is fixed and we can remove such variable), it holds that $q_v \geq 2$ for all $v \in V$. This implies $\sum_{v \in e_\ell} \log q_v \geq \abs{e_\ell}$. By~\eqref{eq-cond-coupling-gen}, it holds that $\sum_{v \in e_\ell} \log q_v \geq \log \frac{1}{p} \geq \frac{50}{\beta}\log\tp{\frac{2000D^4}{\beta}} $. We have
\begin{align*}
\sum_{v \in e_\ell \setminus \{\vst\}}\log\frac{q_v^X}{\ctp{q_v^X / s'_v}} \geq \frac{\beta}{40} \abs{e_\ell} + \frac{\beta}{40}\cdot \frac{50}{\beta}\log\tp{\frac{2000D^4}{\beta}} 	= \frac{\beta}{40} \abs{e_\ell}  + \frac{5}{4}\log\tp{\frac{2000D^4}{\beta}}.
\end{align*}
Hence, this bad event $\+A_\ell$ occurs with probability
\begin{align}
\label{eq-A-3}
\Pr[]{\+A_{\ell}} = 	\prod_{v \in e_\ell \setminus \{\vst\} }\tp{\frac{\ctp{q^X_v / s'_v} }{q^X_v}}^{0.95} \leq \tp{\frac{1}{2}}^{\frac{0.95\beta}{40}\abs{e_\ell}} \cdot \tp{\frac{\beta^{5/4}}{2000^{5/4} D^{5}}}^{0.95} \leq \tp{\frac{1}{2}}^{\frac{\beta}{50}\abs{e_\ell}}{\frac{\beta}{200 D^{3}}},
\end{align}
where the last inequality holds because $\beta \leq 1$.

Finally, if $B_{\ell}^{(1)}$ and all  $B_i^{(z_i)}$ for $1\leq i \leq \ell-1$ occur, then $\+A_i$ occurs for all $1\leq i \leq \ell$. 
By definition, the event $\+A_i$ is determined by a subset of random variables $S_i \subseteq \+R$. For any $i \neq j$, the subset $S_i$ and $S_j$ are disjoint, thus all events $\+A_i$ are mutually independent.
Combining ~\eqref{eq-cond-1-RHS},~\eqref{eq-A-1},~\eqref{eq-A-2} and~\eqref{eq-A-3},
\begin{align*}
\Pr[\cnon]{B_{\ell} \land \forall 1\leq i \leq \ell-1: B_i^{(z_i)}} &\leq \Pr[\cnon]{B_{\ell}^{(1)} \land \forall 1\leq i \leq \ell-1: B_i^{(z_i)}}\\
&\leq \Pr[]{\forall  1\leq i \leq \ell, \+A_i} = \prod_{i = 1}^\ell \Pr[]{\+A_i}\\
&\leq \prod_{e_i \in \+E_1}\tp{\tp{\frac{3}{4}}^{z_i - 1} \frac{1}{200 D^3}} \times \prod_{e_j \in \+E_2}\tp{\frac{1}{200D^3}}	\times \tp{\frac{\beta}{200 D^3} \tp{\frac{1}{2}}^{\frac{\beta \abs{e_\ell}}{50}}}.
\end{align*}
This proves~\eqref{eq-gen-target} in case of~\eqref{eq-cond-1}.

Suppose the condition in~\eqref{eq-cond-2} holds. In this case, we need to bound the RHS of~\eqref{eq-cond-2-RHS}. Compared with the above proof, the only difference is that we need to bound the probability of $B^{(2)}_{\ell}$, where $B^{(2)}_\ell$ denotes the coupling on $\vst$ fails, i.e. $X^{\cnon}_{\vst} \neq Y^{\cnon}_{\vst}$. 
In this case, we have $\log \lfloor{\frac{ q^X_{\vst}}{s'_{\vst}}}\rfloor = \log\lfloor {\frac{ q^Y_{\vst}}{s'_{\vst}}}\rfloor \geq \frac{\beta}{20} \tp{\sum_{v \in e_{\ell}}\log q_v}$. 
Note that in the original input CSP formula of \Cref{alg-mcmc}, it holds that $q_v \geq 2$ for all $v \in V$. This implies $\sum_{v \in e_\ell} \log q_v \geq \abs{e_\ell}$. By~\eqref{eq-cond-coupling-gen}, it holds that $\sum_{v \in e_\ell} \log q_v \geq \log \frac{1}{p} \geq \frac{50}{\beta}\log\tp{\frac{2000D^4}{\beta}} $. Thus, we have
\begin{align*}
\log \ftp{\frac{ q^X_{\vst}}{s'_{\vst}}} = \log\ftp{\frac{ q^Y_{\vst}}{s'_{\vst}}} \geq \frac{\beta}{20} \tp{\sum_{v \in e_{\ell}}\log q_v} \geq \frac{\beta}{40} \abs{e_\ell}  + \frac{5}{4}\log\tp{\frac{2000D^4}{\beta}}.
\end{align*}
Note that $Q^{X}_{\vst} = Q^{Y}_{\vst}$ and $h^X_{\vst} = h^Y_{\vst}$.
In \Cref{lemma-local-uniform-coupling-gen}, 
we can set the parameter $t =\frac{\beta}{40} \abs{e_\ell}$.
This implies that when $\cnon$ couples $X^{\cnon}_{\vst}$ and $Y^{\cnon}_{\vst}$, the probability that the coupling fails is at most
\begin{align*}
\frac{1}{2}\sum_{j \in \Sigma_u'}\frac{\abs{(h^X_{\vst})^{-1}(j)}}{q^X_{\vst}} \tp{\frac{2\beta }{500D^3}}\tp{\frac{1}{2}}^{\frac{\beta}{40}\abs{e_{\ell}}}  \leq \tp{\frac{1}{2}}^{\frac{\beta}{50}\abs{e_\ell}}{\frac{\beta}{200 D^{3}}}.
\end{align*}
The proof of this case is almost the same as the above proof. The only difference is that when coupling $\vst$, we sample a random real number $r_{\vst} \in [0,1]$ uniformly and independently. We use $r_{\vst}$ to implement the coupling such that $X^{\cnon}_{\vst} \neq Y^{\cnon}_{\vst}$ only if $r_{\vst} \leq \tp{\frac{1}{2}}^{\frac{\beta}{50}\abs{e_\ell}}{\frac{\beta}{200 D^{3}}}$. We define the bad event $\+A_{\ell}$ as $r_{\vst} \leq \tp{\frac{1}{2}}^{\frac{\beta}{50}\abs{e_\ell}}{\frac{\beta}{200 D^{3}}}$. By the same proof, we have
\begin{align*}
\Pr[\cnon]{B_{\ell} \land \forall 1\leq i \leq \ell-1: B_i^{(z_i)}} &\leq \Pr[\cnon]{B_{\ell}^{(2)} \land \forall 1\leq i \leq \ell-1: B_i^{(z_i)}}\\
&\leq \prod_{e_i \in \+E_1}\tp{\tp{\frac{3}{4}}^{z_i - 1} \frac{1}{200 D^3}} \times \prod_{e_j \in \+E_2}\tp{\frac{1}{200D^3}}	\times \tp{\frac{\beta}{200 D^3} \tp{\frac{1}{2}}^{\frac{\beta \abs{e_\ell}}{50}}}.
\end{align*}
This proves~\eqref{eq-gen-target} in case of~\eqref{eq-cond-2}.

\subsubsection{Proof of \Cref{lemma-find-projection-coupling}}
\label{section-proof-find-coupling-gen}
Without loss of generality, we assume $\abs{Q^{X}_{v_0}} \leq \abs{Q^{Y}_{v_0}}$. 
Otherwise, we can swap the roles of $X$ and $Y$ in this proof. 
Since the original projection scheme $\Proj{h}$ is uniform,
\begin{align}
\label{eq-QX-QY-v0}
0 \leq \abs{Q^{Y}_{v_0}} - \abs{Q^{X}_{v_0}} \leq 1. 
\end{align}
We first construct the projection scheme $\PX$ for $\Phi^{X}$. To do this, we introduce a CSP formula $\widetilde{\Phi}^X = (V,\widetilde{\Dom{Q}}^X=(\widetilde{Q}^X_v)_{v \in V},\Cons{C})$. We first construct a projection scheme $\widetilde{\Proj{h}}^X$ for  $\widetilde{\Phi}^X$, then transform $\widetilde{\Proj{h}}^X$ to the projection scheme $\PX$. 
Recall the original projection scheme is $\Proj{h} = (h_v)_{v \in V}$, where $h_v:Q_v \to \Sigma_v$. Recall $q_v = \abs{Q_v}$.
The CSP formula  $\widetilde{\Phi}^X$ is define as follows:
\begin{align*}
\widetilde{Q}^X_u = \begin{cases}
h_u^{-1}(X_u) &\text{if } u \neq \vst;\\
h_u^{-1}(j) &\text{if } u = \vst, 	
 \end{cases}	
\end{align*}
where $j\in \Sigma_{\vst}$ is an arbitrary value satisfying $\abs{h^{-1}_{\vst}(j)} = \lfloor{q_{\vst}}/{s_{\vst}}\rfloor$. For each $v \in V$, let $\widetilde{q}^X_v = \abs{\widetilde{Q}^X_v}$. Let $\widetilde{p}$ denote $\max_{c \in \Cons{C}}\prod_{v \in \vbl{c}}\frac{1}{\widetilde{q}^X_v}$. 
By \Cref{condition-projection}, we have for any constraint $c \in \Cons{C}$,
\begin{align*}
\sum_{v \in \vbl{c}}\log\widetilde{q}^X_v \geq \beta\sum_{v \in \vbl{c}}\log q_v	.
\end{align*}
By the condition assumed in \Cref{lemma-find-projection-coupling},  it holds that
\begin{align}
\label{eq-condition-p-tilde}
\log \frac{1}{\widetilde{p}} \geq \beta \log \frac{1}{p} \geq 55(\log D + 3).	
\end{align}
Recall that the maximum degree of the dependency graph of $\widetilde{\Phi}^X$ is also $D$. 
We can use \Cref{theorem-projection-general} on instance $\widetilde{\Phi}^X$ such that the parameter $\alpha$ and $\beta$ in \Cref{theorem-projection-general} are set as $\alpha = 8/9$ and $\beta = 1/9$. 
Remark that in the proof of \Cref{theorem-projection-general}, we use Lov\'asz loca lemma to prove that the projection scheme described in theorem must exist. When $\alpha = 8/9$ and $\beta = 1/9$, the condition in \Cref{theorem-projection-general}  becomes
\begin{align*}
\log \frac{1}{\widetilde{p}} \geq \frac{25 \cdot 9^3}{7^3}(\log D + 3).	
\end{align*}
This implies that under the condition in~\eqref{eq-condition-p-tilde}, there exists a balanced projection scheme $\widetilde{\Proj{h}}^X = (\widetilde{h}_v^X)_{v \in V}$, where $\widetilde{h}_v^X: \widetilde{Q}^X_v \to \widetilde{\Sigma}_v^X$ and $\widetilde{s}_v^X = \abs{\widetilde{\Sigma}_v^X}$ such that for any $c \in \Cons{C}$,
\begin{equation}
\label{eq-proof-PX}
\begin{split}
\sum_{v \in \vbl{c}}\log\frac{\widetilde{q}^X_v}{\ctp{\widetilde{q}^X_v/ \widetilde{s}^X_v}}	 &\geq \tp{1-\frac{8}{9}}\sum_{v \in \vbl{c}}\log\widetilde{q}^X_v \geq \frac{\beta}{9} \sum_{v \in \vbl{c}}\log q_v;\\ 
\sum_{v \in \vbl{c}}\log\ftp{\frac{\widetilde{q}^X_v}{\widetilde{s}^X_v}} &\geq \frac{1}{9} \sum_{v \in \vbl{c}}\log \widetilde{q}^X_v\geq 	\frac{\beta}{9} \sum_{v \in \vbl{c}}\log q_v.
\end{split}
\end{equation}
%and for any clause $c \in \Cons{C}$, let $k_c = \vbl{c}$, if $k_c > 100D$, then it holds that
%\begin{align}
%\label{eq-proof-PX-2}
%\abs{\{v \in \vbl{c} \mid \widetilde{m}^X_v > 1\}} \geq 0.11 k_c.	
%\end{align}
Note that $\widetilde{\Phi}^X$ and $\Phi^{X}$ differ only at variable $\vst$. 
Given the projection scheme $\widetilde{\Proj{h}}^X$ and the original projection scheme $\Proj{h}$, the projection scheme $\PX$ can be constructed as follows
\begin{align*}
h^X_u = \begin{cases}
\widetilde{h}^X_u &\text{if } u \neq \vst;\\
h_u &\text{if } u = \vst.
 \end{cases}
\end{align*}
By definition, $\Proj{h}^X$ is a balanced projection scheme and $h^X_{\vst} = h_{\vst}$.
Since $\widetilde{\Proj{h}}^X$ and $\Proj{h}^X$ differ only at variable $\vst$, for any constraint $c \in \Cons{C}$ such that $\vst \notin \vbl{c}$, by~\eqref{eq-proof-PX},
\begin{align*}
\sum_{v \in \vbl{c}}\log \frac{q^X_v}{\ctp{q^X_v/s^X_v}} &= \sum_{v \in \vbl{c}}\log\frac{\widetilde{q}^X_v}{\ctp{\widetilde{q}^X_v/ \widetilde{s}^X_v}}\geq  \frac{\beta}{10} \sum_{v \in \vbl{c}}\log q_v;\\
\sum_{v \in \vbl{c}}\log\ftp{\frac{q^X_v}{s^X_v}} &= \sum_{v \in \vbl{c}}\log\ftp{\frac{\widetilde{q}^X_v}{\widetilde{s}^X_v}} \geq \frac{\beta}{10} \sum_{v \in \vbl{c}}\log q_v.	
\end{align*}
For variable $\vst$, it holds that $\ftp{q^X_{\vst} / s^X_{\vst}} = \ftp{q_{\vst}/s_{\vst}} = \widetilde{q}^X_{\vst}$, because $\Proj{h}^X$ uses the same way to partition $Q_{\vst}$ as in the original projection scheme $\Proj{h}$. Hence, for any constraint $c \in \Cons{C}$ such that $\vst \in \vbl{c}$, 
\begin{align}
\label{eq-b/9}
\sum_{v \in \vbl{c}} \log \ftp{ \frac{q^X_v}{s^X_v}} &\geq
\sum_{v \in \vbl{c}}\log\ftp{\frac{\widetilde{q}^X_v}{\widetilde{s}^X_v}} \geq 	\frac{\beta}{10} \sum_{v \in \vbl{c}}\log q_v;\notag\\
 \log \ftp{\frac{q^X_{\vst}}{s^X_{\vst}}} + \sum_{v \in \vbl{c} \setminus \{\vst\} } \log \frac{q^X_v}{\ctp{q^X_v / s^X_v}} &= \log \ftp{\frac{q^X_{\vst}}{s^X_{\vst}}}  + \sum_{v \in \vbl{c} \setminus \{\vst\} } \log \frac{\widetilde{q}^X_v}{\ctp{\widetilde{q}^X_v / \widetilde{s}^X_v}}\notag\\
\tp{\text{by $\ftp{q^X_{\vst} / s^X_{\vst}} = \widetilde{q}^X_{\vst}$}}\quad&\geq \sum_{v \in \vbl{c} } \log \frac{\widetilde{q}^X_v}{\ctp{\widetilde{q}^X_v / \widetilde{s}^X_v}}\notag\\
& \geq \frac{\beta}{9} \sum_{v \in \vbl{c}}\log q_v \geq 	\frac{\beta}{10} \sum_{v \in \vbl{c}}\log q_v.
\end{align}
This implies that $\Proj{h}^X$ satisfies all the conditions in \Cref{condition-projection-coupling}. 

Given the projection scheme  $\Proj{h}^X$, the projection scheme $\Proj{h}^Y$ for $\Phi^{Y}$ can be defined as follows.
For each variable $v \in V \setminus \{v_0\}$, $h^Y_v = h^X_v$.
For variable $v_0$, we construct $\Sigma^Y_{v_0} = \Sigma^X_{v_0}$ and $s^Y_{v_0}  = \abs{\Sigma^Y_{v_0}}$, then
arbitrarily map $Q^{Y}_{v_0}$ to $\Sigma^Y_{v_0}$ such that for any $j \in \Sigma^Y_{v_0}$, $\ftp{q^Y_{v_0}/s^Y_{v_0}}\leq  \abs{(h^Y_{v_0})^{-1}(j)} \leq \ctp{q^Y_{v_0}/s^Y_{v_0}}$.
It is easy to see $\Proj{h}^Y$ is also a balanced projection scheme and $h^Y_{\vst}=h_{\vst}$.
It is also easy to see $\Sigma_{v_0}^X = \Sigma_{v_0}^Y$, and $h^X_u = h^Y_u$ for all $u \in V \setminus \{v_0\}$.
We now only need to verify that 
for any $c \in \Cons{C}$,
\begin{align}
\label{eq-qY-1}
\sum_{v \in \vbl{c}}\log\ftp{\frac{q_v^Y}{s_v^Y}} \geq \frac{\beta}{10} \sum_{v \in \vbl{c}}\log q_v;
\end{align}
for any  $c \in \Cons{C}$ satisfying $\vst \notin \vbl{c}$,
\begin{align}
\label{eq-qY-2}
\sum_{v \in \vbl{c}}\log\frac{q_v^Y}{\ctp{q_v^Y / s_v^Y}}	 \geq \frac{\beta}{10}\tp{\sum_{v \in \vbl{c}}\log q_v} ;
\end{align}
and for any $c \in \Cons{C}$ satisfying $\vst \in \vbl{c}$,
\begin{align}
\label{eq-qY-3}
\log \ftp{\frac{q^Y_{\vst}}{s^Y_{\vst}}} + \sum_{v \in \vbl{c} \setminus \{\vst\}}\log\frac{q_v^Y}{\ctp{q_v^Y / s_v^Y}}	 \geq \frac{\beta}{10}\tp{\sum_{v \in \vbl{c}}\log q_v}.
\end{align}
Note that for all $u \in V \setminus \{v_0\}$, it holds that $s^X_u =s^Y_u$ and $q^X_u = q^Y_u$.	
Also note that $s^X_{v_0} = s^Y_{v_0}$.
If $q^X_{v_0} = q^Y_{v_0}$, ~\eqref{eq-qY-1},~\eqref{eq-qY-2} and~\eqref{eq-qY-3} hold trivially. By~\eqref{eq-QX-QY-v0}, we assume $q^Y_{v_0} = q^X_{v_0} + 1$.
Since $q^Y_{u} \geq q^X_{u}$ and $s^X_u = s^Y_u$ for all $u \in V$, for any $c \in \Cons{C}$,
\begin{align*}
\sum_{v \in \vbl{c}}\log\ftp{\frac{q_v^Y}{s_v^Y}} \geq \sum_{v \in \vbl{c}}\log\ftp{\frac{q_v^X}{s_v^X}} \geq \frac{\beta}{10}\sum_{v \in \vbl{c}}\log q_v.
\end{align*}
This proves~\eqref{eq-qY-1}. Note that for all $u \neq v_0$, $q^X_u = q^Y_u$ and $s^X_u = s^Y_u$. Also note that $\vst \neq v_0$. It holds that
\begin{align}
\label{eq-proof-claim-0}
\ftp{\frac{q^Y_{\vst}}{s^Y_{\vst}}}  =  \ftp{\frac{q^X_{\vst}}{s^X_{\vst}}}\quad \text{and}\quad
\forall v \in V \setminus \{v_0\},\quad
\frac{q_{v}^Y}{\ctp{q_{v}^Y / s_{v}^Y}} = \frac{q_{v}^X}{\ctp{q_{v}^X / s_v^X}}.
\end{align}
To prove~\eqref{eq-qY-2} and~\eqref{eq-qY-3}, we only need to compare $\frac{q_{v_0}^X}{\ctp{q_{v_0}^X / s_{v_0}^X}}$ with $\frac{q_{v_0}^Y}{\ctp{q_{v_0}^Y / s_{v_0}^Y}}$. We claim
\begin{align}
\label{eq-proof-claim-1}
\frac{q_{v_0}^Y}{\ctp{q_{v_0}^Y / s_{v_0}^Y}} = \frac{q_{v_0}^X + 1}{\ctp{(q_{v_0}^X +1)/ s_{v_0}^X}} \geq \frac{1}{2}\frac{q_{v_0}^X}{\ctp{q_{v_0}^X / s_{v_0}^X}}.
\end{align}
By~\eqref{eq-proof-PX},~\eqref{eq-proof-claim-0} and~\eqref{eq-proof-claim-1}, for any $c \in \Cons{C}$ such that $\vst \notin \vbl{c}$, we have
\begin{align*}
\sum_{v \in \vbl{c}}\log\frac{q_v^Y}{\ctp{q_v^Y / s_v^Y}} \geq 	\tp{\sum_{v \in \vbl{c}}\log\frac{q_v^X}{\ctp{q_v^X / s_v^X}}} - 1 \geq \frac{\beta}{9}\tp{\sum_{v \in \vbl{c}}\log q_v} - 1 \geq \frac{\beta}{10}\tp{\sum_{v \in \vbl{c}}\log q_v},
\end{align*}
where the last inequality holds because ${\beta}\sum_{v \in \vbl{c}}\log q_v \geq \beta \log \frac{1}{p} \geq 55(\log D + 3) \geq 165$.
This proves~\eqref{eq-qY-2}. Similarly, for any $c \in \Cons{C}$ such that $\vst \in \vbl{c}$, we have
\begin{align*}
\log \ftp{\frac{q^Y_{\vst}}{s^Y_{\vst}}} + \sum_{v \in \vbl{c} \setminus \{\vst\}}\log\frac{q_v^Y}{\ctp{q_v^Y / s_v^Y}} &\geq 	\log \ftp{\frac{q^X_{\vst}}{s^X_{\vst}}} + \tp{\sum_{v \in \vbl{c} \setminus \{\vst\}}\log\frac{q_v^X}{\ctp{q_v^X / s_v^X}}} - 1\\
\tp{\text{by~\eqref{eq-b/9}}}\quad&\geq \frac{\beta}{9}\tp{\sum_{v \in \vbl{c}}\log q_v} - 1 \geq \frac{\beta}{10}\tp{\sum_{v \in \vbl{c}}\log q_v}.
\end{align*}

To prove~\eqref{eq-proof-claim-1}, we consider two case. Recall $s^X_{v_0} =s^Y_{v_0}$. If $q^X_{v_0}$ cannot be divided by $s_{v_0}^X$, then  $\ctp{(q_{v_0}^X +1)/ s_v^X} = \ctp{q_{v_0}^X / s_{v_0}^X}$ and~\eqref{eq-proof-claim-1} holds trivially. If $q^X_{v_0}$ can be divided by $s_{v_0}^X$, then we need to show
\begin{align*}
\frac{q_{v_0}^X + 1}{1 + q_{v_0}^X / s_{v_0}^X} \geq \frac{1}{2}s_{v_0}^X,
\end{align*}
which is equivalent to $q^X_{v_0} \geq s^X_{v_0} - 2$, then~\eqref{eq-proof-claim-1} holds because $q^X_{v_0} \geq s^X_{v_0}$.

\subsection{Proofs of \Cref{lemma-mixing-gen} and \Cref{lemma-mixing-kd}}
\Cref{lemma-mixing-gen} is proved by combining \Cref{lemma-path-coupling}, \Cref{proposition-stationary} and \Cref{lemma-path-coupling-gen}.
Note that the condition in \Cref{lemma-mixing-gen} is $\log \frac{1}{p} \geq \frac{50}{\beta} \log \tp{\frac{2000D^4}{\beta}}$, which suffices to imply the conditions in \Cref{proposition-stationary}  and \Cref{lemma-path-coupling-gen}. This implies the Glauber dynamics has the unique stationary distribution $\nu$ and the mixing rate is $\tmix(\epsilon) \leq \ctp{2n \log \frac{n}{\epsilon}}$.

\Cref{lemma-mixing-kd} is proved by combining \Cref{lemma-path-coupling}, \Cref{proposition-stationary} and \Cref{lemma-path-coupling-uniform}. Given a $(k,d)$-CSP formula, the maximum degree $D$ of the dependency graph is at most $dk$, thus the condition in \Cref{proposition-stationary} becomes $k \log q \geq \frac{1}{\beta}\log(2\mathrm{e}dk)$.
The condition in \Cref{lemma-mixing-kd} is $k\log q \geq \frac{1}{\beta} \log \tp{3000 q^2d^6k^6}$, which suffices to imply the conditions in \Cref{proposition-stationary}  and \Cref{lemma-path-coupling-uniform}. This implies the Glauber dynamics has the unique stationary distribution $\nu$ and the mixing rate is $\tmix(\epsilon) \leq \ctp{2n \log \frac{n}{\epsilon}}$.

\bibliographystyle{alpha}
\bibliography{refs.bib}
\end{document}